\def\year{2021}\relax
\documentclass[letterpaper]{article} % DO NOT CHANGE THIS
\usepackage{aaai21}  % DO NOT CHANGE THIS
\usepackage{times}  % DO NOT CHANGE THIS
\usepackage{helvet} % DO NOT CHANGE THIS
\usepackage{courier}  % DO NOT CHANGE THIS
\usepackage[hyphens]{url}  % DO NOT CHANGE THIS
\usepackage{graphicx} % DO NOT CHANGE THIS
\usepackage{natbib}  % DO NOT CHANGE THIS AND DO NOT ADD ANY OPTIONS TO IT
\usepackage{caption} % DO NOT CHANGE THIS AND DO NOT ADD ANY OPTIONS TO IT
\usepackage{tikz}
\usepackage{pgfplots}
\usepgfplotslibrary{ternary}
\usepackage{amsmath,amssymb,amsthm}
\usepackage{thm-restate}
\usepackage{dsfont}
\usepackage{subcaption}
\usepackage{mathtools}
\usepackage[switch]{lineno}
\newtheorem{theorem}{Theorem}

\newtheorem{definition}{Definition}
\newtheorem{proposition}{Proposition}

\renewcommand{\Pr}{\mathbb{P}} % \Pr è già definito e lo vogliamo modificare

\renewcommand{\l}{\lambda}
\renewcommand{\t}{\tau}
\renewcommand{\d}{\delta}
\newcommand{\E}{\mathbb{E}}
\newcommand{\R}{\mathcal{R}}
\newcommand{\CR}{\rho}
\newcommand{\M}{\mathcal{M}}

\newcommand{\dd}{\textnormal{d}}
\newcommand{\F}{\mathcal{F}}

\definecolor{cp}{rgb}{0.93, 0.23, 0.51}
\definecolor{rose}{rgb}{0.98, 0.8, 0.91}
\definecolor{cerise}{rgb}{0.87, 0.19, 0.39}

\title{Online Posted Pricing with Unknown Time-Discounted Valuations}

\author{
	Giulia Romano\textsuperscript{\textnormal 1}, Gianluca Tartaglia\textsuperscript{\textnormal 2}, Alberto Marchesi\textsuperscript{\textnormal 1}, Nicola Gatti\textsuperscript{\textnormal 1} \\
}	
\affiliations{
	Politecnico di Milano, Piazza Leonardo da Vinci 32, I-20133, Milan, Italy \\
	\textsuperscript{1} \texttt{\{name.surname\}@polimi.it}\\
	\textsuperscript{2} \texttt{gianluca.tartaglia@mail.polimi.it} 
}

\begin{document}

% \linenumbers

\maketitle

\begin{abstract}
We study the problem of designing \emph{posted-price mechanisms} in order to sell a single unit of a single item within a finite period of time.
Motivated by real-world problems, such as, \emph{e.g.}, long-term rental of rooms and apartments, we assume that customers arrive \emph{online} according to a Poisson process, and their valuations  are drawn from an \emph{unknown} distribution and \emph{discounted} over time.
We evaluate our mechanisms in terms of competitive ratio, measuring the worst-case ratio between their revenue and that of an optimal mechanism that knows the distribution of valuations.
First, we focus on the \emph{identical valuation} setting, where all the customers value the item for the same amount.
In this setting, we provide a mechanism $\M_\textsc{c}$ that achieves the best possible competitive ratio, discussing its dependency on the parameters in the case of linear discount.
Then, we switch to the \emph{random valuation} setting.
We show that, if we restrict the attention to distributions of valuations with a monotone hazard rate, then the competitive ratio of $\M_\textsc{c}$ is lower bounded by a strictly positive constant that does not depend on the distribution.
Moreover, we provide another mechanism, called $\M_\textsc{pc}$, which is defined by a piecewise constant pricing strategy and reaches performances comparable to those obtained with $\M_\textsc{c}$.
This mechanism is useful when the seller cannot change the posted price too often.
Finally, we empirically evaluate the performances of our mechanisms in a number of experimental settings.
\end{abstract}

\section{Introduction}

Posted-price mechanisms try to sell an item by proposing a \emph{take-it-or-leave-it} price to each arriving agent, who then decides whether to buy the item or not~\cite{chawla2010multi}.
If an agent opts for purchasing the item, then the mechanism terminates; otherwise, the agent leaves without any further possibility of buying the item, and the mechanism goes on by proposing prices to upcoming agents.
Over the last years, growing attention has been devoted to the analysis of posted-price mechanisms, both in the classical economic literature~\cite{seifert2006posted} and in computer science~\cite{babaioff2015dynamic,babaioff2017posting,adamczyk2017sequential,correa2017posted}, within artificial intelligence and machine learning in particular~\cite{kleinberg2003value,Shah2019Semi}.
This is mainly motivated by the overwhelming number of online economic transactions carried out by posted-price mechanisms. This happens, for example, in online travel agencies (\emph{e.g., Expedia}), accommodation websites (\emph{e.g., Booking.com}), and e-commerce platforms (\emph{e.g., Amazon, eBay}).
As studied by~\citet{einav2018auctions}, an increasing number of \emph{eBay} users prefer buying goods via posted prices rather than participating in auctions.%\gtodo{check esempi.}

Posted-price mechanisms provide many advantages over traditional auction-style mechanisms.
% that are traditionally studied in mechanism design.
%
From the designer's perspective, posting prices requires a much lower effort than running an auction, since it avoids the burden of first \emph{eliciting information} (the bids) from the agents, and then \emph{collecting the payments}.
At the same time, posted-price mechanisms retain most of the desirable properties of classical auctions, such as truthfulness.
Indeed, even though the agents are not required to report their valuations for the item, they are always better off deciding whether to buy the item or not on the basis of their true valuations, without acting strategically~\cite{babaioff2017posting}.
From the agents' perspective, participating in a posted-price mechanism is preferable over competing in an auction, for several reasons.
For instance, agents may prefer revealing minimal information about their true preferences if they plan to participate in similar markets in the future.
Moreover, in some real-world settings, requiring the agents to figure out their true valuations for the item might need some additional efforts on their behalf, while answering a take-it-or-leave-it offer is usually much easier.
%
%\todo{Vi viene in mente altro?}  
%\todo[inline, color=blue]{\textcolor{white}{Lato buyers: non devono aspettare la fine dell'asta per conoscere l'esito, ma lo conoscono immediatamente.}}
%\gitodo{if buyers arrive gradually and they are impatient, is preferable posted pricing}
%\todo[inline, color=blue]{\textcolor{white}{Lato seller, non so se va precisato: in multi-parameter la truthfulness non è scontata, vedi chawla2010multi. Oltre a quello in single item, group strategyproofness, l'unico modo per un buyer di aiutarne un altro è quello di rifiutare un'offerta a lui vantaggiosa.}}

In this work, we study posted-price mechanisms for selling a single unit of a single item within a finite period of time, when the value of the item is discounted over time according to an arbitrary continuous and non-increasing discount function.
%
%Differently from classical mechanism design, the number of agents participating in the mechanism is not known upfront.
%
Discounting is common in many real-world applications and widely studied for a number of economic situations, such, \emph{e.g.}, bargaining \cite{10.2307/1912531,DBLP:journals/ai/GattiGM08} and auctions \cite{DBLP:conf/ijcai/MaoZWC18}.
We tackle settings in which agents arrive sequentially---a common assumptions in \emph{online} mechanism design~\cite{lavi2004competitive, parkes2007online}---and the number of agents is unknown \emph{a priori}.
In particular, following a mainstream approach in economics (see, \emph{e.g.},~\citep{MASON20111699,RosenthalRealEstate}), we assume that agents' arrivals are governed by a Poisson process.
Remarkably, posted pricing with Poisson arrivals has been previously investigated by~\citet{wang1993auctions} and~\citet{rong2018dynamic} for undiscounted settings, though without providing any theoretical result.

We assume that each agent arriving at the mechanism has a different initial (\emph{i.e.}, undiscounted) valuation for the item, which is independently drawn according to a common probability distribution.
This leads to a fundamental trade-off between setting high prices so as to achieve high revenue and, on the other side, progressively lowering posted prices so as to increase the probability of selling the item. 
Our assumption is that the mechanism is only aware of the range of valuations, while it does not know anything about the shape of the distribution.
%
%This is reasonable since it is usually easier to get an estimate of the range than one of the actual distribution (such as, \emph{e.g.}, in the room renting example).
This is reasonable since, differently from the actual distribution, the range of valuations can be estimated from previous data or market surveys.
%
%Indeed, learning the distribution of valuations is not viable given that the seller only observes a sequence of reject decisions, terminated by only one accept, which is clearly not sufficient.
%
%\todo{La parte sopra va spiegata meglio.}

%\textcolor{red}{
%
\citet{lavi2004competitive}~and~\citet{babaioff2017posting} provide the main state-of-the-art results on posted-price mechanisms for single-item single-unit scenarios.
However, their models do not fit to our setting, since the agents' valuations are \emph{not} discounted over time and the number of agents is known \emph{a priori}. 
As a result, these models do not embed an explicit time representation and the proposed pricing strategies are only driven by the number of agents arrived.
%}
%
%This model can only capture the case in which the time interval between two arrivals is fixed and \emph{a priori} known.

Our model encompasses many real-world scenarios, such as, \emph{e.g.}, long-term rental of rooms and apartments.
Think of a website renting rooms to students for fixed periods of one year.
The value of a room naturally decreases over time, reflecting the fact that a future tenant will benefit from the room for a period shorter than one year.
Moreover, the potential customers arrive at the renting website according to a stochastic process, which can be reasonably modeled by a Poisson process whose rate parameter can be easily estimated by looking at traffic logs of the website.

%\todo[inline, color=blue]{\textcolor{white}{Commento: Non so se si vuole inserire questo ragionamento. Anche nelle conclusions magari.}}

%\todo[inline, color=blue]{\textcolor{white}{Questo lavoro non è da considerarsi un'alternativa agli studi che cercano di stimare la demand curve attraverso tecnica di learning, \emph{e.g.} kleinberg2003value o \citep{Shah2019Semi}. Il nostro metodo segue un approccio diverso: trovare un meccanismo che sia \emph{robusto} rispetto ad una famiglia di possibili demand curves. Approccio che è stato proposto da \citep{babaioff2017posting}, dove si è osservato che considerare tutte le possibili distribuzioni sul supporto $[1,h]$ porterebbe a performance troppo basse. Tale famiglia è stata quindi definita come quella che presenta MHR. Nulla vieterebbe di trovare una sintesi tra questi due approcci, usando tecniche di learning per affinare la stima sulla famiglia di distribuzioni considerata $\mathcal{F}$ (magari dove c'è una limited supply e non unlimited) e modificare la strategia di prezzo in modo che rimanga \emph{robusta} wrt $\mathcal{F}$. Moreover, tecniche di learning come \emph{feature selection} o \emph{linear regression} possono essere combinate al presente lavoro per gestire un setting multi-parameter, dove le items differenti presentano caratteristiche comuni }}

\paragraph{Original Contributions}
%
%\textcolor{red}{
We adopt the perspective of competitive analysis~\citep{borodin2005online} and evaluate our mechanisms in terms of \emph{competitive ratio}, measuring the worst-case ratio between their revenue and that of an optimal mechanism that knows the distribution of valuations. 
As it is customary in the literature (see, \emph{e.g.}, \citep{babaioff2017posting,kleinberg2003value}), we first focus on the {identical valuation} setting in which all the agents share the same initial valuation for the item.
Then, we extend our results to the random valuation setting where the agents' valuations are drawn i.i.d.~from the same distribution satisfying the monotone hazard rate condition (when the distributions of valuations are unrestricted, \citet{lavi2004competitive} and~\citet{babaioff2017posting} show that then there is no algorithm with good performances).
In the identical valuation setting, we design a posted-price mechanism $\M_\textsc{c}$ and prove that it is optimal, \emph{i.e.}, it provides the best possible competitive ratio.
In order to derive the ratio, we first identify two crucial properties that characterize  optimal mechanisms: their undiscounted price is non-increasing in time and they always guarantee the same fraction of the expected revenue of an optimal mechanism that knows the agents' valuation, independently of its actual value.
For the specific case of linear discount, we discuss how the competitive ratio depends on the parameters.
In the random valuation setting, we first show that mechanism $\M_\textsc{c}$ still provides good performances by proving that its competitive ratio is lower bounded by a constant, which does not depend on the distribution of agents' valuations.
Then, motivated by real-world scenarios in which the seller is constrained not to change the posted prices too often, we propose a new mechanism $\M_\textsc{pc}$ defined by a piecewise constant pricing strategy and prove that its performances in terms of competitive ratio are comparable with those obtained by $\M_\textsc{c}$.
In conclusion, we empirically compare $\M_\textsc{c}$ with a natural adaption of the mechanism proposed by~\citet{babaioff2017posting} to our setting, showing that the latter is inefficient even without time discounting.
We also empirically evaluate the performances of $\M_\textsc{c}$ and $\M_\textsc{pc}$ as the frequency with which prices are allowed to change decreases, showing that, when this is not too low, then the performances of $\M_\textsc{pc}$ and $\M_\textsc{c}$ are comparable.
%}

\paragraph{Other Related Works}
%\textcolor{red}{SPOSTEREI, MODIFICANDOLO, QUESTO PARAGRAFO ALL'INTERNO DEL PARAGRAFO SUBITO PRIMA DEI CONTRIBUTI ORIGINALI. LO SCOMPORREI IN VARI PARTI (AD ESEMPIO I RIFERIMENTI SU POISSON ANDREBBERO SPOSTATI). 
%%
%They show that, when the distributions of valuations are unrestricted, then there is no algorithm with good performances. 
%%
%Instead, with monotone hazard rate distributions, good algorithms are possible.
%%\gtodo{altri lavori competitive analysis per posted price?} 
%Poisson arrivals in posted pricing are studied by \citet{rong2018dynamic} and \citet{wang1993auctions} when agents' valuations are undiscounted. 
%%
%However, these works do not provide any competitive analysis.
%}
As showed by~\citet{hajiaghayi2007automated} for single-item settings, posted pricing is strictly related to the \emph{secretary problem} and to \emph{prophet inequalities}; see also the work by~\citet{babaioff2009secretary} for single-item settings and that of~\citet{lucier2017economic} for multi-item scenarios.
Differently from our model, these works assume that the mechanism knows the probability distribution of agents' valuations and that the agents reveal their actual valuation for the item upon arrival.
When multiple units of the same item are available, learning approaches based on \emph{bandit} techniques are customarily adopted.
In particular, \citet{kleinberg2003value} study an unlimited-supply setting where the number of buyers is fixed, and derive upper bounds on the regret. Several recent works extend the results in~\citep{kleinberg2003value}.
%, we mention a few. 
%
\citet{Shah2019Semi} study a contextual setting, providing a semi-parametric model that learns from the observation of a binary outcome which stands for acceptance or rejection of the offered price.
\citet{mohri2014optimal} study revenue-maximizing learning algorithms for posted pricing with strategic buyers. They consider a repeated game in which, at each round, the seller offers the item at a certain price and a strategic buyer accepts or rejects it.
In that work, the goal is to learn the buyers' valuation for the item by minimizing the \emph{strategic-regret} of the algorithm.

\section{Preliminaries}\label{sec:prelim}

%In Section~\ref{subsec:model}, we introduce basic definitions and notation used in the paper.
%%
%Then, in Section~\ref{subsec:competitive}, we discuss how to evaluate posted-price mechanisms through the lenses of competitive analysis.

% \paragraph{Model and Notation}
%
We study a model in which a seller is interested in selling a single unit of an item within a finite time period of length $T$.
The seller implements a \emph{posted-price mechanism} by setting a take-it-or-leave-it price at each time $t \in [0,T]$.
We denote by $p : [0,T] \to \mathbb{R}_+$ the \emph{pricing strategy} adopted by the seller, with $p(t)$ being the price offered at time $t \in [0,T]$.
The agents (\emph{i.e.}, the buyers) arrive sequentially over time, according to a Poisson process with rate parameter $\l >0$.
%
% We denote by $N_t$ the cumulated number of agents arriving in the time interval $[0,t]$.

We label agents according to their order of arrival (\emph{i.e.}, agent $i$ is the $i$-th agent arriving in $[0,T]$).
Each agent $i$ has a private valuation $V_i$ for the item, drawn from a distribution $F$ with finite support $[v_\textnormal{min},v_\textnormal{max}]$, where $v_\textnormal{max} > v_\textnormal{min} > 0$ denote the maximum and minimum valuation, respectively.
In the following, for the ease of presentation, we normalize agents' valuations in the range $[1,h]$, where we define $h \coloneqq \frac{v_\textnormal{max}}{v_\textnormal{min}}$.
Accordingly, we scale the support of $F$ to $[1,h]$.
Then, we denote by $f$ the probability density function of $F$.
% its cumulative density function. 

The value of the item for sale decreases over time.
In particular, $V_i$ is agent $i$'s initial valuation at time $t=0$.
We model decreasing values by introducing a continuous non-increasing \emph{discount function} $\xi: [0,T] \to [0,1]$ such that $\xi(0)=1$ and $\xi(T)=0$.
By letting $W_i$ be the random variable representing the arrival time of agent $i$, we define the agent $i$'s \emph{discounted valuation} as $D_i \coloneqq V_i \, \xi(W_i)$, which represents how much agent $i$ is willing to pay upon her arrival.
As a result, whenever agent $i$ arrives, she buys the item if and only if $D_i \geq p(W_i)$, \emph{i.e.}, her discounted valuation is at least as large as the price offered by the mechanism.% at the time of arrival. 

We introduce the following additional notation. 
We denote by $I_{s,\t} \coloneqq [s,s+\t] \subseteq [0,T]$ the time interval of length $\t \in [0,T]$ starting from time $s \in[0,T-\t]$.
The number of agents arriving in $I_{s,\t}$ is a random variable denoted by $N_{s,s+\t}$.
Given $\t \in [0,T]$, the random variables $N_{s,s+\t}$ are equally distributed for all $s \in [0,T-\t]$, as the arrivals are generated by a Poisson process.
For the sake of presentation, we omit $s$ in $N_{s,s+\t}$, denoting by $N_\t$ the random variable of the number of agents arriving in any time interval of length $\tau$,
%
% it holds that $N_{s,s+\tau}=N_{s',s'+\tau} \coloneqq N_\t$ for all $s,s' \in [0,T-\t]$, where $N_\t$ denotes a random variable representing the number of agents arriving in any time interval of length $\tau$.
%
% Hence, in each $\t$-lengthed time interval of $[0,T]$, $N_\t$ agents arrive.
%
which follows a Poisson distribution with parameter $\l \t$.~\footnote{By definition of Poisson distribution, $\Pr \left\{ N_\t \hspace{-1mm}=\hspace{-1mm} j  \right\}  \hspace{-.5mm} = \hspace{-.5mm} \frac{(\l \t)^j e^{-\l \t}}{j !}$.}
Thus,  $N_T$ is the random variable of the total number of agents arriving in the overall time period.
% $[0,T]$ and has distribution $\mathcal{P}(\lambda T)$.
%
% Moreover, agent $i$'s arrival time (a.k.a. \emph{waiting time}) is denoted by the random variable $W_i$.
%
In the following, we sometimes focus on the \emph{linear} discount function, denoted as $\xi_\textnormal{lin}: [0,T] \to [0,1]$ with $\xi_\textnormal{lin}(t) \coloneqq  1-\frac{t}{T} $.
%
% some of the following analysis we model the discount as a linear function of time: $\xi_{lin}(t)=\big(1-\frac{t}{T}\big)$.
%
In this case, each agent $i$'s discounted valuation is $D_i \coloneqq V_i \, \left( 1-\frac{W_i}{T} \right)$.

\subsection{Performances of Posted-Price Mechanisms}

Given a deterministic posted-price mechanism $\M$ defined by a price function $p_{\M}: [0,T] \to \mathbb{R}_+$, we denote  by $\E_F[\R(\M)]$ the expected revenue that the mechanism provides to the seller.
The expectation is calculated with respect to both the Poisson arrivals and the distribution $F$ of agents' initial valuations.
We made explicit the dependence on $F$, as we will frequently refer to it along the paper.

We adopt the perspective of competitive analysis and measure the performances of a mechanism $\M$ by comparing the seller's expected revenue with that of a benchmark mechanism $\M^\star$, which is optimal having knowledge of the distribution $F$.
Notice that the benchmark has no information on the actual realizations of agents' initial valuations, but only on their distribution, whereas the mechanisms we propose operate having knowledge of their range only. 
%
%which is an optimal clairvoyant mechanism. 
%
% In the following sections, we define a suitable benchmark for each  specific setting we study.

Our goal is to bound the performances of our mechanisms w.r.t.~those of the benchmark $\M^\star$ by looking at the worst case over the set $\F$ of possible distributions $F$, \emph{i.e.}, all those with support $[1,h]$.
This is captured by the following:
%
%, we use the following metric to measure the performances of our mechanisms.
%
\begin{definition}%[Competitive ratio]
	The \emph{competitive ratio} of a deterministic posted-price mechanism $\M$ is defined as:
	%
	% , where agents' valuations are drawn i.i.d from a distribution $F$, is the worst-case ratio between its expected revenue and the expected revenue of the benchmark $\M^\ast$, defined as follows:
	\[
		\CR(\mathcal{M}) \coloneqq \min_{F \in \F} \CR_F(\M), \quad\textnormal{\emph{where} } \CR_F(\M) \coloneqq \frac{\E_F[\R(\M)]}{\E_F[\R(\M^\star)]}.
	\]
	Moreover, we say that a mechanism is \emph{optimal} when its competitive ratio is the highest possible among all the deterministic posted-price mechanisms.
\end{definition}

Notice that $\CR(\M) \in [0,1]$ and, for every possible distribution $F \in \F$, we are guaranteed that the seller's expected revenue $\E_F[\R(\M)]$ provided by mechanisms $\M$ is at least a fraction $\CR(\M)$ of that achieved by $\M^\star$, \emph{i.e.}, it holds $\E_F[\R(\M)] \geq \CR(\M) \,\E_F[\R(\M^\star)]$.
%
% In particular, we measure the performances of $\M$ through a worst-case competitive analysis

As previously showed by~\citet{babaioff2017posting} in similar settings, we can safely restrict our analysis to mechanisms maintaining the bottom price for a non-negligible period of time.
%
%As we show in the following sections, a fundamental step towards defining mechanisms with good performances in terms of competitive ratio is to identify some of their crucial properties, so as to restrict the set of mechanisms where one needs to search.
%%
%The following proposition introduces the first feature that guides our analysis: any meaningful mechanism must maintain the bottom price for a non-negligible period of time.~\footnote{Statements analogous to Proposition~\ref{prop:minimum} have already been introduced in similar settings~\cite{babaioff2017posting}.} 
%
% Our aim is to find a mechanism that is optimal according to a worst-case competitive analysis. We need to define some of its characteristics in order to find a smaller set of mechanisms where to search for it. The following proposition is a first step in this direction.
%
%
Indeed, in the case in which $F$ places all the probability mass on $1$, then any mechanism providing a non-null seller's expected revenue must set the minimum price during some time interval, otherwise no agent would buy the item.
\begin{proposition} \label{prop:minimum}
	Every deterministic posted-price mechanism $\M$ such that $\CR(\M) > 0$ must set the minimum price $p_{\M}(t) = \xi(t)$ for every $t$ in a time interval of length $\t > 0$.
\end{proposition}
%
% Also \cite{babaioff2017posting} and \cite{Zheng2016Posting} claim that the optimal mechanism maintains the bottom price for the last period of the time horizon. Despite the difference in the analyzed settings, it is possible to use the same method in order to prove this result.
%
% Indeed, consider the worst-case scenario, that is when all agents have the minimum valuation $v=1$. In this case a mechanism not setting the minimum price would have a competitive ratio equal to zero.

\section{Identical Valuation Setting}\label{sec:identical_valuation}

%\subsection{Environment and Problem Formulation}

We start studying the \emph{identical valuation} (IV) setting, where all the agents share the same initial valuation $v \in [1,h]$ for the item, \emph{i.e.}, it holds $V_i = v$ and $D_i = v \, \xi(W_i)$ for every agent $i$.
%
% have the same valuation $v \in [1,h]$ for the item.
%
% Therefore, each agent $i$'s discounted valuation is represented  by the random variable $D_i \coloneqq v \, \xi(W_i)$.
%
The IV setting is a special case of the general random valuation model where one restricts the attention to distributions $F$ placing all the probability mass on a single valuation in $[1,h]$.
In the following, we adjust notation for expected revenues and competitive ratios accordingly, writing $\E_v [\R(\M)]$ and $\CR_v(\M)$ instead of $\E_F [\R(\M)]$ and $\CR_F(\M)$.

Our main result (Theorem~\ref{thm:bestCR_gen}) is to provide a deterministic posted-price mechanism, called $\M_\textsc{c}$, which is optimal for the IV setting for every discount function $\xi$.
We also study the specific case of a linear discount function $\xi_\textnormal{lin}$, where we design an optimal mechanism $\M_{\textsc{c}, \textnormal{lin}}$ (Theorem~\ref{thm:bestCR_lin}) that enjoys an easily interpretable analytical description.
All the proofs are in the Appendix.
%
%\todo{Cambiare $\M_1$?} 
%
% Our main result is providing the best deterministic posted-price mechanism $\M_1$ in order to sell a unique item in this scenario. 

% \subsection{The benchmark for the IV setting}

First, we describe the shape of the benchmark mechanism $\M^\star$ for the IV setting.
%
%introduce an optimal clairvoyant mechanism $\B_{IV}$ that we use as a benchmark for the IV setting.
%
% \todo{Userei la $\M$ per tutti i meccanismi.} 
%
% The mechanism applies a revenue-maximizing pricing strategy, having knowledge of both the discount function $\xi$ and the initial common valuation $v$.
%
Indeed, since $\M^\star$ knows the actual initial valuation $v$, its price function $p_{\M^\star} : [0,T] \to \mathbb{R}_+$ is such that $p_{\M^\star}(t) = v \, \xi(t)$ for $t \in [0,T]$.
%
% \todo{Originariamente le pricing function sono definite senza pedice.} 
%
Therefore, we can compute the expected revenue of $\M^\star$ as follows:
%
% We compare the performances of our mechanism to the benchmark $\B_{IV}$, a clairvoyant optimal mechanism for the Identical Valuation setting. It applies the revenue-maximizing pricing strategy, knowing both the discount $\xi(t)$ and the initial valuation $v$. Indeed, it posts $p_{\B_{IV}}(t)=v\xi(t)$, for $t \in [0,T]$. We compute its expected revenue:
%
\begin{multline}
	\E_v \left[ \R( \M^\star)\right] \coloneqq \int_{0}^{T} p_{\M^\star}(t) \, \l \,e^{-\l t} \,  \dd t= \\ \int_{0}^{T} v \, \xi(t) \,  \l \,e^{-\l t} \,  \dd t= v \, k^\star, \label{eq:bench_IV}
\end{multline}
where $k^\star \coloneqq \int_{0}^{T} \xi(t) \,  \l \, e^{-\l t} \,  \dd t$ does not depend on $v$, but only on the problem parameters $T$, $\l$, and the discount function $\xi$.
%, while it is independent of the valuation $v$.
%
Let us remark that the expected revenue of the benchmark $\M^\star$ defined in Equation~\eqref{eq:bench_IV} is expressed as a \emph{linear} function of $v$.

%\begin{observation}\label{ER_lin}
%	The expected revenue of the benchmark $\B_{IV}$ is linearly dependent on the initial valuation $v$.
%\end{observation}

%\todo[inline]{Forse conviene non avere una sotto-sezione per la parte sopra.}

\subsection{Optimal Mechanism for a General Discount}

We start proving two lemmas that highlight two crucial properties which characterize optimal posted-price mechanisms for the IV setting.
Lemma~\ref{lem:decr} implies that the pricing strategy of an optimal mechanism must be such that the undiscounted price defined as $\frac{p(t)}{\xi(t)}$ is non-increasing in $t$, whereas Lemma~\ref{lem:CR_const} shows that any mechanism which always provides a constant fraction of the expected revenue of the benchmark, independently of the agents' initial valuation $v$, is an optimal mechanism. 
%
% In particular, we claim that its pricing strategy is non-increasing and that its competitive ratio is constant with respect to $v$.

\begin{restatable}{lemma}{decr}\label{lem:decr}
	In the IV setting, given any deterministic posted-price mechanism $\M$, there always exists a deterministic posted-price mechanism $\M'$ with undiscounted price $\frac{p_{\M'}(t)}{\xi(t)}$ non-increasing in $t$ such that $\E_v [ \mathcal{R}(\M) ] \le \E_v [ \mathcal{R}(\M^\prime) ]$ for every possible agents' initial valuation $v \in [1,h]$.
\end{restatable}

Notice that, since $\xi$ is continuous and non-increasing by definition, Lemma~\ref{lem:decr} also shows that there is always an optimal mechanism whose pricing strategy is non-increasing.
Moreover, by recalling Proposition~\ref{prop:minimum}, we can conclude that any optimal mechanism must set the minimum price at the end of the overall time period, \emph{i.e.}, during a time interval $[t_0, T] \subseteq [0,T]$ defined for some $t_0 \in [0,T)$.
This result is exploited to prove the following lemma.

% The following lemma state that the best mechanism has a constant competitive ratio. 
% In the proof we show that it is always possible to find a problem instance such that any other mechanism has a lower competitive ratio than that of the optimal mechanism.

\begin{restatable}{lemma}{CRconst}\label{lem:CR_const}
	In the IV setting, let $\M$ be a deterministic posted-price mechanism whose pricing strategy $p_\M$ satisfies 
	%$p_\M(0)=h$ and 
	$p_\M(t)=\xi(t)$ for $t \in [t_0,T]$ with $t_0 \in [0,T)$.
	If the ratio $\CR_v(\M) = \frac{\E_v [\R(\M)]}{\E_v [\R(\M^\star)]}$ for $\M$ does not depend on the agents' initial valuation $v$, then $\M$ is an optimal mechanism.
\end{restatable}

By Lemma~\ref{lem:CR_const}, in order to find an optimal mechanism for the IV setting, we can restrict the attention to mechanisms $\M$ whose ratios $\CR_v(\M)$ do not depend on the initial valuation $v$.
Therefore, since the expected revenue of the benchmark $\M^\star$ is a linear function of $v$ (see Equation~\eqref{eq:bench_IV}), we can search for an optimal mechanism among those having an expected revenue which linearly depends on $v$.
This crucial observation allows us to design the optimal mechanism $\M_\textsc{c}$ in Theorem~\ref{thm:bestCR_gen} 
%
%
% Now, we provide our main results for the IV setting.
%
% Consider the set of deterministic posted-price mechanisms not knowing the initial valuation of agents $v$ in the Identical Valuation setting.
%
% We describe mechanism $\M_1$ and we show that it is optimal among all the mechanisms of this set. In particular, $\M_1$ has a constant competitive ratio (Lemma \ref{CR_const}), hence, it is the optimal mechanism according to a worst-case competitive analysis.
%
% Our main result (Theorem~\ref{thm:bestCR_gen}) identifies a deterministic posted-price mechanism, called $\M_1$, which is indeed an optimal mechanism for the IV setting.
%
% Notice that, by Proposition~\ref{prop:minimum} and Lemma~\ref{lem:decr}, we already know that $\M_1$ must set price $p_{\M_1}(t) = \xi(t)$ for all the instants $t$ in a suitably defined time interval $[t_0, T] \subseteq [0,T]$.
%
% Moreover, by Lemma~\ref{lem:CR_const}, we can assume w.l.o.g. that the desired mechanism $\M_1$ is such that 
by leveraging the condition $\E_v [\R(\M_\textsc{c}) ] = k \,v$ for every $v \in [1,h]$, with $k$ being a suitably defined constant independent of $v$. 
The key insight that allows us to derive an expression for $\M_\textsc{c}$ is that we can always find the desired pricing strategy $p_{\M_\textsc{c}}$ among the continuous price functions such that $\frac{p_{\M_\textsc{c}}(t)}{\xi(t)}$ is non-increasing in $t \in [0,T)$.
Intuitively, using Lemma~\ref{lem:decr}, we can always express the expected revenue $\E_v [\R(\M_\textsc{c}) ]$ as a function of the time $t^\ast \coloneqq \sup \{ t \in [0,t_0] \mid p_{\M_\textsc{c}}(t) > v \, \xi(t) \}$, which is the first time in which $p_{\M_\textsc{c}}$ intersects $v \, \xi(t)$.
The reason is that it holds $p_{\M_\textsc{c}}(t) \leq v \, \xi(t)$ if and only if $t \geq t^\ast$, and, thus, only agents arriving after $t^\ast$ are willing to buy the item.
By using the relation among $\E_v [\R(\M_\textsc{c}) ]$ and $t^*$, we can find the desired pricing strategy $p_{\M_\textsc{c}}$ as a solution to a suitably defined differential equation.
This leads to the following theorem.
%
% The observations above allow us to prove Theorem~\ref{thm:bestCR_gen}, whose complete proof is in Appendix~\ref{app:iv}.

%We start by restricting the research of the optimal mechanism to the set of those whose price strategy $p(t)$ has a unique intersection point $t^*$ with the discounted valuation $v\xi(t)$.
%\begin{definition}
%	Given a pricing strategy $p(t)$ and a discounted valuation $v\xi(t)$ we call $t^*$ the time instant $s.t.$ $t^* \in [0,t_{0}]$ and $p(t^*)=v\xi(t^*)$. 
%\end{definition}
%\textcolor{red}{
%	Observation \ref{minimum} implies that $v\xi(t) \geq p(t)$ for $t>t^*$}, therefore, the item can be sold only after $t^*$.\\
%\textcolor{red}{1) è da spiegare meglio?}

\begin{restatable}{theorem}{bestCRgen}\label{thm:bestCR_gen}
	In the IV setting, there exists an optimal deterministic posted-price mechanism $\M_\textsc{c}$ whose pricing strategy $p_{\M_\textsc{c}}$ is defined as follows:
	\begin{equation*}
		p_{\M_\textsc{c}}(t) \coloneqq \left\{\begin{array}{ll} a \, e^{\int b(t) \dd t} & \textnormal{if } t \in\left[0, t_{0}\right) \\ \xi(t) & \textnormal{if } t \in\left[t_{0}, T\right] \end{array}\right.,
	\end{equation*}
	where $b$ is a function such that $b(t) \coloneqq \l-\frac{\l}{k \zeta(t)}-\frac{\zeta'(t)}{\zeta(t)}$ with $\zeta(t) \coloneqq \frac{1}{\xi(t)}$, whereas $a$, $k$, $t_0$ are suitably defined constants that do not depend on the agents' initial valuation $v$.
	%
	% where $A$ is constant with respect to the initial valuation $v$, $h(t)=\l-\frac{\l}{k \zeta(t)}-\frac{\zeta'(t)}{\zeta(t)}$ and $\zeta(t)=\frac{1}{\xi(t)}$.
	%
	% $k \leq 1$ is a number depending on $\l$, $h$, $T$ and ${t_0}$. Such a mechanism achieves a competitive ratio of $\frac{k}{k_{\B_{IV}}}$, where $k_{\B_{IV}}$ is the expected revenue of the benchmark when $v=1$.
	%
\end{restatable}

As a byproduct of the proof of Theorem~\ref{thm:bestCR_gen}, we also get an expression for the competitive ratio of the mechanism $\M_\textsc{c}$, as stated by the following corollary.

\begin{restatable}{corollary}{corCRgen}
	In the IV setting, mechanism $\M_\textsc{c}$ achieves:
	\[
		\CR(\M_\textsc{c}) = \frac{\int_{0}^{T-t_{0}} \xi(t)\, \l \,e^{-\l t} \,\dd t}{\int_{0}^{T} \xi(t) \,\l \,e^{-\l t} \,\dd t}.
	\]
\end{restatable}

% Notice that, as expected, $\CR(\M_1) $ only depends on the problem parameters $\l$, $T$, $h$ (through $t_0$), and the discount function $\xi$, while it is independent of the agents' initial valuation $v$.

\subsection{Optimal Mechanism for a Linear Discount}

The pricing strategy $p_{\M_\textsc{c}}$ of the optimal mechanism defined in Theorem~\ref{thm:bestCR_gen} still depends on some parameters, namely $a$, $k$, and $t_0$, which do not admit an easy analytical formula for a general discount function $\xi$.
Nevertheless, they can be expressed analytically if we restrict the attention to functions $\xi$ having a particular shape.
In the following Theorem~\ref{thm:bestCR_lin} and Corollary~\ref{cor:constantratiodue}, we analyze the case of a linear discount function $\xi_\textnormal{lin}$, defining an optimal mechanism $\M_{\textsc{c}, \textnormal{lin}}$ for such setting.

\begin{restatable}{theorem}{bestCRlin}\label{thm:bestCR_lin}
	In the IV setting with linear discount function $\xi_\textnormal{lin}$, there exists an optimal deterministic posted-price mechanism $\M_{\textsc{c}, \textnormal{lin}}$ whose pricing strategy $p_{\M_{\textsc{c}, \textnormal{lin}}}$ is defined as:
	\begin{equation*} 
		p_{\M_{\textsc{c}, \textnormal{lin}}}(t) \coloneqq \left\{\begin{array}{ll} \hspace{-1mm} h \,\left(1-\frac{t}{T}\right) \, e^{\lambda\left(1-\frac{1}{k}\right)t+\frac{\lambda}{2 k T} t^{2}} &\hspace{-1mm} \textnormal{if } t \in\left[0, t_{0}\right) \\ \hspace{-1mm} 1-\frac{t}{T} &\hspace{-1mm} \textnormal{if } t \in\left[t_{0}, T\right]\end{array}\right. \hspace{-1mm} ,
	\end{equation*}
	where $k \coloneqq \lambda \, t_{0} \, \frac{2 \,T-t_{0}}{2 \,T\left(\lambda \, t_{0}+\ln h\right)}$ and the time $t_0 \in [0,T)$ is defined as the unique positive real root of the following equation: $1-\frac{1}{\lambda T}\left(1+\lambda \, t_{0}-e^{-\lambda\left(T-t_{0}\right)}\right)=k$.
	%where $t_0 \in [0,T]$ is the time s.t. $1-\frac{1}{\lambda T}\left(1+\lambda t_{0}-e^{-\lambda\left(T-t_{0}\right)}\right)=k$.$k \leq 1$ depends on the problem parameters $\lambda$, $h$, $T$ and ${t_{0}}$, and it is constant with respect to the choice of $v$. Such a mechanism achieves a competitive ratio of $\frac{k}{1-\frac{1}{\lambda T}\left(1-e^{-\lambda T}\right)}$.
\end{restatable}

%\textcolor{red}{1)Devo mostrare che $t_0$ è l'unica radice reale positiva?}
%
%\todo[inline]{Se e ovvio no, se possono essere di piu si.}

The prices posted by $\M_{\textsc{c},\textnormal{lin}}$ decrease as a linearly discounted exponential function until $t=t_0$, starting, at time $t=0$, by setting the price equal to the maximum agents' initial valuations $h$.
Then, during the time interval $[t_0, T]$, the price function linearly decreases and equals zero in $t=T$.

\begin{restatable}{corollary}{constantratiodue}\label{cor:constantratiodue}
	In the IV setting with linear discount function $\xi_\textnormal{lin}$, $\M_{\textsc{c}, \textnormal{lin}}$ achieves a competitive ratio:
	\[
		\CR(\M_{\textsc{c}, \textnormal{lin}}) = \frac{1-\frac{1}{\lambda T}\left(1+\lambda t_{0}-e^{-\lambda\left(T-t_{0}\right)}\right)}{1-\frac{1}{\lambda T}\left(1-e^{-\lambda T}\right)}.
	\]
\end{restatable}

\begin{table}[t]
\begin{center}
%	{\renewcommand{\arraystretch}{1.0}
	\resizebox{\columnwidth}{!}{
	\begin{tabular}{r|c|c|c}
		& $T \rightarrow \infty$ & $\l \rightarrow \infty$ & $h \rightarrow \infty$ \\
		\hline
		$t_0$ & $\Theta ( \sqrt{T} ) $ & $\Theta( \sqrt{T/\l}) $ & $\Theta( T)$ \\
		$\CR(\M_{\textsc{c}, \textnormal{lin}})$ & $\Theta\left(1-\frac{1}{\sqrt{T}}\right)$ & $\Theta\left(1-\frac{1}{\sqrt{\lambda}}\right)$ & $\Theta\left(\frac{1}{\log^2(h)}\right)$ \\ 
		$\lim\CR(\M_{\textsc{c}, \textnormal{lin}})$ & 1 & 1 & 0 \\ 
	\end{tabular}
	}
\end{center}
\caption{Values of $t_0$ and $\CR(\M_{\textsc{c}, \textnormal{lin}})$ as $T,\lambda, h$ go to infinity.}
\label{table:MlinAsymptotic}
\end{table}

The asymptotic values of $t_0$ and $\CR(\M_{\textsc{c}, \textnormal{lin}})$ as $T,\lambda, h$ go to infinity are in Table~\ref{table:MlinAsymptotic} (see the Appendix for more details).
In particular, $\CR(\M_{\textsc{c}, \textnormal{lin}})$ goes asymptotically to $1$ as $\lambda$ or $T$ increases.
This corresponds to having an infinite number of agents and, thus, selling the item with certainty.
Instead, $\CR(\M_{\textsc{c}, \textnormal{lin}})$ decreases as $h$ increases, going asymptotically to $0$ as $\frac{1}{\log^2(h)}$.
The range $[1,h]$ represents the degree of uncertainty that the mechanism has on the agents' valuation.
Therefore, $\CR(\M_{\textsc{c}, \textnormal{lin}})$ decreases as the uncertainty increases and it cannot be lower bounded by any strictly positive constant if no finite upper bound on $h$ is known (\emph{i.e.}, when $h \rightarrow +\infty$). 
However, the dependency of $\CR(\M_{\textsc{c}, \textnormal{lin}})$ on the degree of uncertainty is logarithmic.
Instead, notice that a trivial mechanism setting the price equal to $\xi_\textnormal{lin}(t)$ for $t \in [0,T]$ would have a competitive ration of $\frac{1}{h}$, which depends linearly on the degree of uncertainty.
%

%
%\textcolor{red}{
%1) devo dire in teo o corollario che è il meccanismo ottimo\\
%2)notazione: per derivata uso $\partial$ o $d$ ?}
%
%
%
%
%\textcolor{red}{
%2) Problema nomi meccanismi: $\M_1$ e $\M_3$ non mi sembra vadano bene come nomi, come li chiamiamo? inoltre, nella dim del lemma 2 uso i nomi $\M_1, \M_2, \M_3, \M_4$ per fare gli esempi. \\
%3) In appendice mettiamo anche il caso senza sconto? Paper Zheng?\\
%4) Check titoli subsections}
\section{Random Valuation Setting}\label{sec:randomvaluation}

%\gitodo{fare una nota su indici variabili $\l \t$ dicendo che la numerazione parte da 1 internamente a quell'intervallo - rispetto a quello che si era detto nei preliminaries}

% \subsection{Environment and Problem Formulation}

We now switch to the \emph{random valuation} (RV) setting, where agents' initial valuations $V_i$ are i.i.d. random variables defined by a cumulative distribution function $F$ with support $[1,h]$.
%
%\textcolor{red}{
%In the following, we focus on a linear discount function such that $\xi_\textnormal{lin}(t) = 1 - \frac{t}{T}$ for $t \in [0,T]$.
%}
%
We focus on distributions $F$ satisfying the \emph{monotone hazard rate} (MHR) condition.
Formally, a distribution $F$ is MHR if the hazard rate $H(x) \coloneqq \frac{f(x)}{1-F(x)}$ is non-decreasing in $x$.
This assumption is common when studying posted-price mechanisms that operate without knowing the shape of the distribution of valuations (see~\citep{babaioff2015dynamic,babaioff2017posting}) and many distributions used in practice satisfy it (such as, \emph{e.g.}, uniform, normal, and exponential distributions).
Moreover, the MHR condition is necessary for proving our main results (Theorems~\ref{thm:bound_M1}~and~\ref{thm:LB_step}). Indeed, when the family of possible distributions is unrestricted, one cannot design posted-price mechanisms guaranteeing a constant fraction of the revenue of $\M^\star$ independently of the distribution $F$, as shown by~\citet{babaioff2017posting} for the easier setting in which agents do not arrive stochastically.
All the proofs are provided in the Appendix.

%First, in Section~\ref{subsec:aux}, we introduce some concepts that we need in the rest of the section.
%%
%Then, Section~\ref{subsec:S_M1} analyzes the competitive ratio of $\M_1$ in the RV setting, 
%%showing that it can be lower bounded by a constant which does not depend on the distribution $F$.
%%
%whereas Section~\ref{subsec:m2} introduces a new mechanism $\M_2$ with a piecewise constant pricing strategy that still provides performances in terms of competitive ratio comparable to those obtained by $\M_1$.\todo{Questo paragrafo si puo' omettere per avere maggiore spazio.}

\paragraph{Auxiliary Definitions and Results}\label{subsec:aux}

% We start with some auxiliary definitions.
%
We introduce the random variable $X_{\l \t}$ as the maximum initial valuation of agents arriving in an interval of length $\tau \in (0,T]$.
Formally:
\[
	X_{\l \t} \coloneqq \max_{i\in\{1,\dots,N_\t\}}  V_i .~\footnote{In the definition of $X_{\l \t}$ and $Y_{s, \l \t}$, overloading the notation, we assume that the agents arriving in the considered time interval of length $\t$ are labeled from $1$ to $N_\t$ according to their order. Their actual labels referred to the overall period $[0,T]$ may be different.}
\]
% where $N_\tau$ is a random variable representing the number of agents arriving in a time interval of length $\tau$.
%
$X_{\l\t}$ is the first order statistic of $N_\t$ samples drawn from $F$ and, since agents' arrivals are governed by a Poisson process, its cumulative distribution function $F_{X_{\l \t}}$ is defined as:
\begin{multline*}
	F_{X_{\l\t}}(x) \coloneqq \sum_{j=1}^\infty \Pr \left\{ N_\t=j \right\} F_{X_{\l\t}|N_\t=j} (x) =\\ \sum_{j=1}^\infty \frac{(\l\t)^j e^{-\l\t}}{j!} \left[ F(x) \right]^j = e^{-\l\t(1-F(x))} .
\end{multline*}
% where $F_{X_{\l\t}|N_\t=j}$ is the cumulative distribution function of $X_{\l\t}$ conditioned on the event $N_\t=j$.
%
%
% We define $\bar{Y}_{\lambda \t}$ as the random variable of the maximum discounted valuation of agents arriving in a $\t$-lengthed interval.
% \[\bar{Y}_{\lambda \t} =  \max_{i\in\{1,\dots,N_\t\}}V_i\xi(W_i)\]
%In the following analysis we model the discount as a linear function of time: $\xi_{lin}(t)=\big(1-\frac{t}{T}\big)$. $V_i\xi_{lin}(W_i)$ is the linearly discounted valuation.
%

We also define $Y_{s, \lambda \t}$ as the random variable representing the maximum discounted valuation of agents arriving in an interval $I_{s,\t}$ of length $\t \in (0,T]$ starting at $s \in [0, T-\t]$:
%
%\todo{Non dipende dal tempo di inizio?}
%\gitodo{Sì, infatti sarebbe da chiamare $Y_{s,\l\t}$ o $Y_{\l\t, s}$ e dire che è definita come: the random variable representing the maximum discounted valuation of agents arriving in the time interval $I_{s,\t}$ of length $\t \in (0,T]$ }
%
% when the discount is linearly dependent on the time.
%
\[
	Y_{s, \lambda \t} \coloneqq \max_{i\in\{1,\dots,N_\t\}} D_i .
\]
The cumulative distribution function $F_{Y_{s, \lambda \t}}$ of $Y_{s, \lambda \t}$ is:
\begin{multline*}
	F_{Y_{s, \lambda \t}} (x) \coloneqq \sum_{j=1}^\infty \Pr \left\{ N_\t=j \right\} F_{Y_{s, \l\t}|N_\t=j}(x) = \\ \sum_{j=1}^\infty \frac{(\l\t)^j e^{-\l\t}}{j!}F_{Y_{s, \l\t}|N_\t=j}(x), 
\end{multline*}
where $F_{Y_{s, \l\t}|N_\t=j}$ is the cumulative distribution function of $Y_{s,\l\t}$ conditioned on the event $N_\t=j$.
Let us remark that, by definition, $F_{Y_{s,\lambda \t}}$ depends on distribution $F$.
In the following, we also let $Y_{\l T} \coloneqq Y_{0,\l T}$ be the random variable representing the maximum discounted valuation of agents arriving in the overall time period $[0,T]$.
In the Supplemental Material, for the specific case of a linear discount function, we show how to exploit some useful properties of Poisson processes so as to find an analytical expression for $F_{Y_{\lambda T}}$.
In particular, by letting $Z \coloneqq V \, U$, where $V$ and $U$ are independent random variables distributed according to $F$ and $\mathcal{U}(0,1)$, respectively, we obtain:
\begin{equation*}
	F_{Y_{\lambda T}}(x) \coloneqq \sum_{j=1}^\infty \frac{(\l T)^j e^{-\l T}}{j!} \left[ F_Z(x) \right]^j,
\end{equation*}
where
\begin{equation*}
	F_{Z}(x) \coloneqq \left\{\begin{array}{ll}x \int_{1}^{h} \frac{1}{z} f(z) \dd z & \text {if } x \in[0,1) \\ F(x)+x \int_{x}^{h} \frac{1}{z} f(z) \dd z & \text {if } x \in[1, h]\end{array}\right. .
\end{equation*}

%
%
%
%\textcolor{red}{
%\subsection{The Random Valuation Benchmark} \label{bench_RV}
%In this Random-Valuation setting, the benchmark $\B_{RV}$ is the optimal mechanism that sets prices knowing the distribution of the valuations $F$. However, in this scenario, setting the optimal pricing strategy of the benchmark is not trivial as in the Identical-Valuation case. Hence, we observe that the expected revenue of $\B_{RV}$ is at most $\E[Y_{\l T}]$. This value is the expected revenue of the optimal mechanism that knows the realizations of the discounted valuations. Notice that such a mechanism would set a price equal to the realization of $Y_{\l T}$ on the whole interval $[0,T]$. Therefore, we use this value as a benchmark to compare the performances of our algorithm, recalling that
%\[\E[\R(\mathcal{B}_{RV})] \le \E[Y_{\l T}]\]
%}

\subsection{Bounding \textnormal{$\CR(\M_\textsc{c})$} in the RV Setting} \label{subsec:S_M1}

We show that mechanism $\M_\textsc{c}$ (see definition in Theorem~\ref{thm:bestCR_gen}), which is optimal in the IV setting, provides good performances also in the RV setting.
Our main result (Theorem~\ref{thm:bound_M1}) is a lower bound on the competitive ratio of the mechanism, which is obtained by showing that $\M_\textsc{c}$ always provides at least a constant fraction of the seller's expected revenue achieved by the benchmark $\M^\star$, independently of the distribution of agents' initial valuations $F$.~\footnote{To prove the lower bound, we follow an approach similar to that used by~\citet{babaioff2017posting} to bound the competitive ratio of their \emph{Equal-Sample-of-Every-Scale} mechanism. However, our setting introduces additional challenges, since the agents' arrivals are stochastic and the valuations are discounted. Thus, our proofs require different techniques w.r.t.~those of~\citet{babaioff2017posting}.} 
This is surprising since, differently from $\M^\star$, our mechanism works without having knowledge about $F$ (except for its range).
%
% This only holds if we restrict the space of possible distributions $F$ to that having the monotone hazard rate property.~\footnote{Notice that, as shown in~\citep{babaioff2017posting}, the monotone hazard rate assumption is necessary in order to obtain posted-price mechanisms with a constant competitive ratio.}

% In order to prove our main result, 
We first need some definitions and lemmas.
%
% Next, we introduce the ratio of prices posted by mechanism $\M_1$ at the endpoints of intervals of a given length $\tau$.

\begin{definition}\label{def:kappa_s}
	Let $I_{s,\t}$ be any interval of length $\t \in (0,T]$ starting at $s \in [0,T-\t]$.
	Then, the ratio between the prices posted by $\M_\textsc{c}$ at the endpoints of $I_{s,\t}$ is defined as:
	\[
		\kappa_\t (s)\coloneqq\frac{p_{\M_\textsc{c}}(s)}{p_{\M_\textsc{c}}( s + \t)} . %=\frac{T-s}{(T-\tau-s)e^{\lambda\left(1-\frac{1}{k}\right)\tau + \frac{\lambda}{2 k T}( \tau^{2}+2s\tau)}}.
	\] 
	%
	%	\[
	%		\kappa_\t (s)\coloneqq\frac{p_{\M_1}(s)}{p_{\M_1}( s + \t)}=\frac{T-s}{(T-\tau-s)e^{\lambda\left(1-\frac{1}{k}\right)\tau + \frac{\lambda}{2 k T}( \tau^{2}+2s\tau)}}.
	%	\] 
\end{definition}
Intuitively, $\kappa_\t (s)$ bounds the slope of the price function of $\M_\textsc{c}$ in the time interval $I_{s,\t}$, which depends on both the starting time $s$ and the length $\t$ of the interval.
Moreover, notice that $\kappa_\t (s) \geq 1$ since $p_{\M_\textsc{c}}$ is non-increasing by Lemma~\ref{lem:decr}.
Next, we introduce an upper bound on the price ratios of all the time intervals of length $\t$, which is useful in deriving our main result.
\begin{definition}\label{def:kappa}
	The maximum price ratio of $\M_\textsc{c}$ over intervals of length $\t \in (0,T]$ is denoted by:
	\[
		\kappa_\t \coloneqq \max_{s \in[0,T-\t]} \kappa_\t(s).
	\]
\end{definition}

The following lemma establishes a relation between the price function $p_{\M_\textsc{c}}$ of $\M_\textsc{c}$ and the expected value of the random variable $X_{\l T}$ representing the maximum initial valuation of agents arriving in the overall time period.
This is crucial to prove Theorem~\ref{thm:bound_M1}.

\begin{restatable}{lemma}{interv}\label{lem:interv}
	In the RV setting with agents' initial valuations drawn from a distribution $F$, given $\t \in (0,T]$ and $0 < \epsilon < 1$, there exists at least an interval $I_{s,\tau}$ of length $\t \in (0,T]$ starting at $s \in [0, T -\t]$ such that the prices $p_{\M_\textsc{c}}(t)$ posted by mechanism $\M_\textsc{c}$ during the time instants $t \in I_{s,\tau}$ lie in the range $\left[\frac{\E{[X_{\lambda T}] \xi(s+\t)(1-\epsilon)}}{\kappa_\t},\E{[X_{\lambda T}]\xi(s+\t)(1-\epsilon)}\right]$.
	%
	%$\left[\frac{\E{[X_{\lambda T}]\left(1-\frac{\tau+s}{T}\right)(1-\epsilon)}}{\kappa_\t},\E{[X_{\lambda T}]\left(1-\frac{\tau+s}{T}\right)(1-\epsilon)}\right]$.
\end{restatable}

% OLD COUNTEREXAMPLE
%
%
%Notice that the range of prices defined in Lemma~\ref{lem:interv} depends on the distribution $F$ through the expected value $\E [X_{\lambda T}]$.
%%
%Moreover, the existence of an interval as in Lemma~\ref{lem:interv} is no longer guaranteed for any distribution $F$ if we take $s \in [t_0, T]$.
%%
%In order to see this, we provide a counterexample, as follows. 
%%
%As explained in the proof of Lemma~\ref{lem:interv}, an interval $I_{s,\t}$ can be always found by defining its starting point $s$ such that $p_{\M_1}({s}) = \E{[X_{\lambda T}]\left(1-\frac{\tau+{s}}{T}\right)(1-\epsilon)}$.
%%
%We know that the posted price $p_{\M_1}({s})$ is $1-\frac{s}{T}$ when $s \ge t_0$.
%%
%Hence, we can rewrite the previous equation as $1-\epsilon = \frac{1-\frac{s}{T}}{\E{[X_{\lambda T}]\left(1-\frac{\tau+{s}}{T}\right)}}$, where $0 < 1-\epsilon < 1$.
%%
%Now consider a point distribution $F$ such that all the agents have the same valuation $v=1$. In this case, $\E[X_{\lambda T}] = 1$ and $1-\epsilon = \frac{1-\frac{s}{T}}{1-\frac{\tau+{s}}{T}} > 1$, that is not possible. 
%%
%Therefore, we can conclude that, for $s \ge t_0$, it is not always possible to find an interval $I_{s,\t}$ defined as in Lemma \ref{lem:interv}.

The following two lemmas are the final pieces that we need to prove Theorem~\ref{thm:bound_M1}.
Lemma~\ref{lem:F_IHR} establishes that, if the distribution $F$ is MHR, then the same holds for the distribution $F_{X_{\l \t}}$ of $X_{\l \t}$.
Lemma~\ref{lem:ex_ln}, given two intervals of length $\t$ and $\t'$ with $\t \leq \t'$, provides a lower bound on the expected value of $X_{\l \t}$ which depends on the expected value of $X_{\l \t'}$ and the logarithms of the expected number of agents' arrivals in the two intervals, respectively $\l \t $ and $\l \t'$.

\begin{restatable}{lemma}{FIHR}\label{lem:F_IHR}
	$F_{X_{\l \t}}$ has non-decreasing monotone hazard rate.
\end{restatable}

\begin{restatable}{lemma}{Zheng}\label{lem:ex_ln}
	For every $\t, \t' \in (0, T]$ with $\t \leq \t'$, it holds:
	\[
		\frac{\E[X_{\l\t}]}{\E[X_{\l \t'}]} \ge \frac{\ln \left( {\l\t} \right)}{\ln \left({\l \t'} \right)}.
	\]
\end{restatable}

\begin{restatable}{theorem}{boundM}\label{thm:bound_M1}
	Consider the RV setting with $\l \t = (\lambda T)^{1-\epsilon} \ge 1- \ln (e-1)$ for some $\t \in (0,T]$ and $0 < \epsilon < 1$.
	Then, restricted to the set $\F$ of distributions $F$ satisfying the MHR condition, mechanism $\M_\textsc{c}$ has a competitive ratio that can be lower bounded as follows:
    \[
    \CR(\mathcal{M}_\textsc{c}) \geq \frac{\xi(t_0 + T^{1-\epsilon}\l^{-\epsilon})(1-\epsilon)}{\kappa_\t e}.
    \]
	%	
	%	In the RV setting with a linear discount function $\xi_\textnormal{lin}$, let $0 < \epsilon < 1$ be a constant such that $(\l T)^\epsilon \ge \log_{\kappa}h$.
	%	%
	%	Then, for any MHR distribution $F$ with support $[1,h]$, it holds:
	%	%
	%	\[
	%		\CR_F(\M_1) \geq \frac{\big(1-(\l T)^{-\epsilon}-\frac{t_0}{T}\big)(1-\epsilon)}{\kappa e},
	%	\]
	%	%
	%	where $\kappa$ is a constant that only depends on the problem parameters $\l$, $T$, and $h$.
	%
	% This lower bound is constant with respect to the choice of $F$. 
\end{restatable}

The idea of the proof is to use $\CR_F(\M_\textsc{c}) \geq \frac{\E_F [\R(\M_\textsc{c})] }{ \E [Y_{\l T}]}$, following from the fact that $\E_F [ \R(\M^\star) ]$ cannot be larger than $\E [Y_{\l T}]$, which is the expected revenue achieved by an optimal mechanism that knows the realization of agents' initial valuations and arrivals.
Then, $\E_F [\R(\M_\textsc{c})]$ is lower bounded by the revenue that $\M_\textsc{c}$ achieves in a suitably defined interval $I_{s,\t}$, whose existence is guaranteed by Lemma~\ref{lem:interv}.
Moreover, Lemmas~\ref{lem:interv},~\ref{lem:F_IHR},~and~\ref{lem:ex_ln}, together with the properties of MHR distributions, allow us to write $\E_F [\R(\M_\textsc{c})] \geq \frac{\E [X_{\l T}]\xi(t_0 + \t) (1-\epsilon)}{\kappa_\t e}$, giving the result as $\E[Y_{\l T}] \leq \E[X_{\l T}]$.
%
%Noticing that $\E[Y_{\l T}] \geq \E[X_{\l T}]$ gives the final result.

%The idea of the proof is that the expected revenue of the mechanism in the overall time period is greater than or equal to the one achieved in a shorter time interval of length $\tau$. In particular, we consider a time interval $I_{s,\t}$ satisfying the condition stated in Lemma \ref{lem:interv}, if there exists one with $s < t_0$. Otherwise, we consider the time interval $I_{t_0,\t}$. At this point, all the previous lemmas (\ref{lem:interv}, \ref{lem:F_IHR} and \ref{lem:ex_ln}) are needed in order to derive the lower bound provided by Theorem \ref{thm:bound_M1} for the case of $s < t_0$, while, in the other case, it is easier to get the same result.  

%\textcolor{red}{
%Note that the monotone-hazard rate assumption is only required for $F$, the distribution of valuations.	No such requirement is needed for the distribution of the random variables $V_i(1-\frac{W_i}{T})$, for $i \in \{1,\dots,N_T\}$, that are the linearly discounted valuations of the agents arriving in $[0,T]$.
%}

\subsection{A Mechanism with a Piecewise Constant Price}\label{subsec:m2}

We introduce a new mechanism $\M_\textsc{pc}$ whose pricing strategy $p_{\M_\textsc{pc}}$ is a piecewise constant function.
This turns out to be useful in all the situations in which the seller is constrained not to change the posted price too often, \emph{e.g.}, when the mechanism is required to set prices for time intervals having a given minimum length.
Our main result (Theorem~\ref{thm:LB_step}) is a lower bound on the competitive ratio of $\M_\textsc{pc}$ in the RV setting, which is comparable to that obtained for $\M_\textsc{c}$ in Theorem~\ref{thm:bound_M1}.
Thus, we show that, even in presence of constraints on the allowed pricing strategies, we are still able to design mechanisms with good performances in terms of competitive ratio.
Clearly, $\M_\textsc{pc}$ depends on the minimum length requirement, which influences the resulting lower bound.
In particular, $\M_\textsc{pc}$ is tuned by a parameter $\delta$ related to the number of time intervals in which the price must be constant.

%\textcolor{red}{Va motivata bene l'utilità di avere un meccanismo che offra prezzo costante in intervalli temporali di lunghezza prestabilita. Nella realtà potrebero esserci dei vincoli per cui non è possibile variare il prezzo ad ogni istante. Il vincolo impone dunque di mantenere il prezzo per un certo periodo di tempo. In base a questo scegliamo il parametro $\d$ che determina la lunghezza degli intervalli in cui è suddiviso $[0,T]$.\\
%Altra motivazione: vogliamo generalizzare il meccanismo di babaioff al nostro scenario.\\
%Queste motivazioni va messa qui o nell'Introduzione? }

Mechanism $\M_\textsc{pc}$ works by evenly partitioning the time interval $[0,t_0]$ into $\lceil \log_{\d}h \rceil$ sub-intervals of length $\t$, where $\delta \in (1,h]$ and $t_0 \in [0,T]$ are suitably defined parameters.
Then, the remaining time $[t_0, T]$ is organized in other sub-intervals of length $\t$.
%
% We evenly partition $[0,t_0]$ in $\lceil \log_{\d}h\rceil$ $\t$-lengthed subintervals, where $t_0 \in [0,T]$, hence $\delta \geq 1$. \gtodo{fatto i conti per dire che $\delta \geq 1$, anche intuitivamente è così}
%
As a result, $[0,T]$ is partitioned into $\left\lceil \frac{T}{\t} \right\rceil$ sub-intervals, which, overloading notation, we denote by $I_i \coloneqq \left[ (i-1) \t, \min \{i\t, T \} \right]$ for $i = 1,\ldots,\left\lceil \frac{T}{\t} \right\rceil$.
Notice that $\t=\frac{t_0}{\lceil \log_{\d}h\rceil}$, and, thus, parameters $t_0$ and $\delta$ can be tuned to match the required minimum length $\t$.
The pricing strategy $p_{\M_\textsc{pc}}$ of $\M_\textsc{pc}$ is defined in such a way that the price is constant in each interval $I_i$.
By letting $p_{\M_\textsc{pc}}(I_i)$ be the price posted during $I_i$, we define the function $p_{\M_\textsc{pc}}$ as follows:~\footnote{Whenever $T$ is not divisible by $\t$, then the last time interval is shorter than $\t$. Thus, in order to satisfy the minimum length constraint, we set its price equal to the one in the preceeding interval.}
\begin{equation*}
	p_{\M_\textsc{pc}}(I_i) \coloneqq \left\{\begin{array}{ll} \frac{h}{\d^i} \xi(i\t) & \textnormal{if } i = 1,\ldots,\lfloor \log_{\d}h\rfloor \\ 
	\xi(i\t) & \textnormal{if }  i = \lceil \log_{\d}h\rceil, \ldots , \left\lceil \frac{T}{\t} \right\rceil -1  .\\
	\xi((i-1)\t) & \textnormal{if } i = \left\lceil \frac{T}{\t} \right\rceil \end{array}\right. 
\end{equation*}
%
%The constant price posted in interval $I_i$ is denoted by $p(I_i)$ and $\d$ is a parameter in the range $(1,h]$ that can be optimized.

%\gitodo{parentesi intervalli [) altrimenti prezzi si sovrappongono in un punto} 
%
%Each agent arriving in the time interval $I_i$ is offered a price by mechanism $\M_\textsc{pc}$ that is:
%\begin{itemize}
% \item $\frac{h}{\d^i}(1-\frac{i\t}{T})$, for $i\in\{1,\dots,\lfloor \log_{\d}h\rfloor \}$
% \item $(1-\frac{i\t}{T})$, for $i\in\{\lceil \log_{\d}h\rceil, \dots , \lceil \frac{T}{\t} \rceil -1 \}$
% \item $(1-\frac{(i-1)\t}{T})$ for $i = \lceil \frac{T}{\t} \rceil$
%\end{itemize}

%Notice that $\frac{h}{\d^i}(1-\frac{i\t}{T}) = (1-\frac{i\t}{T})$ when $\lfloor \log_{\d}h\rfloor = \lceil \log_{\d}h\rceil$.

We compare in Figure~\ref{fig:M1and2lin} the prices of $\M_{\textsc{c},\textnormal{lin}}$ and $\M_{\textsc{pc},\textnormal{lin}}$ (\emph{i.e.}, $\M_\textsc{pc}$ with a linear discount) in a specific setting for two values of $\tau$. Notice that $\M_\textsc{pc}$ can be thought of as an extension of the \emph{Equal-Sample-of-Every-Scale} (ESoES) mechanism by~\citet{babaioff2017posting} to the more general setting in which agents arrive stochastically according to a Poisson process and agents' valuations are discounted over time. 

\begin{figure}[t]
\begin{subfigure}{0.2\textwidth}
\begin{tikzpicture}[domain=0:2, scale=0.5]
\draw[->] (-0.2,0) -- (6.5,0) node at (6.8,-0.3) {$t$};
\draw[->] (0,-0.2) -- (0,4.5)node[above=0.1cm] {};%{$p_{\M_{2,\textnormal{lin}}}(t)$};
\draw[gray, thin, dashed] (0,4) -- (6,0);
\draw[gray, thin, dashed] (0,1) -- (6,0);
\draw[gray, thin, dashed] (6,0) -- (6,4);
\node at (-0.3,4) {$h$};
\node at (-0.3,1) {$1$};
\node at (6,-0.3) {$T$};
\node at (1.5,-0.3) {$t_{0}$};
\draw[color=blue, thin, domain=0:0.688] plot (\x,{4*(1-\x/6)*exp(-3.85437997725*\x+2.65453166477*\x^2)}); 
\draw[blue, thin] (0.688,0.877390) -- (6,0);
\draw[black, thick] (0,2.31) -- (0.5,2.31);
\draw[black, thick] (0.5,2.31) -- (0.5,1.323);
\draw[black, thick] (0.5,1.323) -- (1,1.323);
\draw[black, thick] (1,1.323) -- (1,0.75);
\draw[black, thick] (1,0.75) -- (1.5,0.75);
\draw[black, thick] (1.5,0.75) -- (1.5,0.6666);
\draw[black, thick] (1.5,0.6666) -- (2,0.6666);
\draw[black, thick] (2,0.6666) -- (2,0.5833);
\draw[black, thick] (2,0.5833) -- (2.5,0.5833);
\draw[black, thick] (2.5,0.5833) -- (2.5,0.5);
\draw[black, thick] (2.5,0.5) -- (3,0.5);
\draw[black, thick] (3,0.5) -- (3,0.417);
\draw[black, thick] (3,0.417) -- (3.5,0.417);
\draw[black, thick] (3.5,0.417) -- (3.5,0.3333);
\draw[black, thick] (3.5,0.3333) -- (4,0.3333);
\draw[black, thick] (4,0.3333) -- (4,0.25);
\draw[black, thick] (4,0.25) -- (4.5,0.25);
\draw[black, thick] (4.5,0.25) -- (4.5,0.167);
\draw[black, thick] (4.5,0.167) -- (5,0.167);
\draw[black, thick] (5,0.167) -- (5,0.0833);
\draw[black, thick] (5,0.0833) -- (6,0.0833);
\end{tikzpicture}
\centering
%\caption{$\M_{2,\textnormal{lin}}$: $\t=1, h=2.8, \l=10, T=12$}
\caption{$\t=1$}
\label{fig:M2linA}
\centering
\end{subfigure} \qquad 
\begin{subfigure}{0.2\textwidth}
\begin{tikzpicture}[domain=0:2, scale=0.5]
\draw[->] (-0.2,0) -- (6.5,0) node at (6.8,-0.3) {$t$};
\draw[->] (0,-0.2) -- (0,4.5)node[above=0.1cm] {};%{$p_{\M_{2,\textnormal{lin}}}(t)$};
\draw[gray, thin, dashed] (0,4) -- (6,0);
\draw[gray, thin, dashed] (0,1) -- (6,0);
\draw[gray, thin, dashed] (6,0) -- (6,4);
\node at (-0.3,4) {$h$};
\node at (-0.3,1) {$1$};
\node at (6,-0.3) {$T$};
\node at (0.875,-0.3) {$t_{0}$};
\draw[color=blue, thin, domain=0:0.688] plot (\x,{4*(1-\x/6)*exp(-3.85437997725*\x+2.65453166477*\x^2)});
\draw[blue, thin] (0.688,0.877390) -- (6,0);
\draw[black, thick] (0,3.213) -- (0.125,3.213);
\draw[black, thick] (0.125,3.213) -- (0.125,2.58);
\draw[black, thick] (0.125,2.58) -- (0.25,2.58);
\draw[black, thick] (0.25,2.58) -- (0.25,2.07);
\draw[black, thick] (0.25,2.07) -- (0.375,2.07);
\draw[black, thick] (0.375,2.07) -- (0.375,1.66);
\draw[black, thick] (0.375,1.66) -- (0.5,1.66);
\draw[black, thick] (0.5,1.66) -- (0.5,1.331);
\draw[black, thick] (0.5,1.331) -- (0.625,1.331);
\draw[black, thick] (0.625,1.331) -- (0.625,1.067);
\draw[black, thick] (0.625,1.067) -- (0.75,1.067);
\draw[black, thick] (0.75,1.067) -- (0.75,0.854);
\draw[black, thick] (0.75,0.854) -- (0.875,0.854);
\draw[black, thick] (0.875,0.854) -- (0.875,0.8333);
\draw[black, thick] (0.875,0.8333) -- (1,0.8333);
\draw[black, thick] (1,0.8333) -- (1,0.8125);
\draw[black, thick] (1,0.8125) -- (1.125,0.8125);
\draw[black, thick] (1.125,0.8125) -- (1.125,0.792);
\draw[black, thick] (1.125,0.792) -- (1.25,0.792);
\draw[black, thick] (1.25,0.792) -- (1.25,0.771);
\draw[black, thick] (1.25,0.771) -- (1.375,0.771);
\draw[black, thick] (1.375,0.771) -- (1.375,0.75);
\draw[black, thick] (1.375,0.75) -- (1.5,0.75);
\draw[black, thick] (1.5,0.75) -- (1.5,0.729);
\draw[black, thick] (1.5,0.729) -- (1.625,0.729);
\draw[black, thick] (1.625,0.729) -- (1.625,0.708);
\draw[black, thick] (1.625,0.708) -- (1.75,0.708);
\draw[black, thick] (1.75,0.708) -- (1.75,0.6875);
\draw[black, thick] (1.75,0.6875) -- (1.875,0.6875);
\draw[black, thick] (1.875,0.6875) -- (1.875,0.6667);
\draw[black, thick] (1.875,0.6667) -- (2,0.6667);
\draw[black, thick] (2,0.6667) -- (2,0.6458);
\draw[black, thick] (2,0.6458) -- (2.125,0.6458);
\draw[black, thick] (2.125,0.6458) -- (2.125,0.625);
\draw[black, thick] (2.125,0.625) -- (2.25,0.625);
\draw[black, thick] (2.25,0.625) -- (2.25,0.6042);
\draw[black, thick] (2.25,0.6042) -- (2.375,0.6042);
\draw[black, thick] (2.375,0.6042) -- (2.375,0.5833);
\draw[black, thick] (2.375,0.5833) -- (2.5,0.5833);
\draw[black, thick] (2.5,0.5833) -- (2.5,0.5625);
\draw[black, thick] (2.5,0.5625) -- (2.625,0.5625);
\draw[black, thick] (2.625,0.5625) -- (2.625,0.5417);
\draw[black, thick] (2.625,0.5417) -- (2.75,0.5417);
\draw[black, thick] (2.75,0.5417) -- (2.75,0.5208);
\draw[black, thick] (2.75,0.5208) -- (2.875,0.5208);
\draw[black, thick] (2.875,0.5208) -- (2.875,0.5);
\draw[black, thick] (2.875,0.5) -- (3,0.5);
\draw[black, thick] (3,0.5) -- (3,0.4792);
\draw[black, thick] (3,0.4792) -- (3.125,0.4792);
\draw[black, thick] (3.125,0.4792) -- (3.125,0.4583);
\draw[black, thick] (3.125,0.4583) -- (3.25,0.4583);
\draw[black, thick] (3.25,0.4583) -- (3.25,0.4375);
\draw[black, thick] (3.25,0.4375) -- (3.375,0.4375);
\draw[black, thick] (3.375,0.4375) -- (3.375,0.4167);
\draw[black, thick] (3.375,0.4167) -- (3.5,0.4167);
\draw[black, thick] (3.5,0.4167) -- (3.5,0.3958);
\draw[black, thick] (3.5,0.3958) -- (3.625,0.3958);
\draw[black, thick] (3.625,0.3958) -- (3.625,0.375);
\draw[black, thick] (3.625,0.375) -- (3.75,0.375);
\draw[black, thick] (3.75,0.375) -- (3.75,0.3542);
\draw[black, thick] (3.75,0.3542) -- (3.875,0.3542);
\draw[black, thick] (3.875,0.3542) -- (3.875,0.3333);
\draw[black, thick] (3.875,0.3333) -- (4,0.3333);
\draw[black, thick] (4,0.3333) -- (4,0.3125);
\draw[black, thick] (4,0.3125) -- (4.125,0.3125);
\draw[black, thick] (4.125,0.3125) -- (4.125,0.2917);
\draw[black, thick] (4.125,0.2917) -- (4.25,0.2917);
\draw[black, thick] (4.25,0.2917) -- (4.25,0.2708);
\draw[black, thick] (4.25,0.2708) -- (4.375,0.2708);
\draw[black, thick] (4.375,0.2708) -- (4.375,0.25);
\draw[black, thick] (4.375,0.25) -- (4.5,0.25);
\draw[black, thick] (4.5,0.25) -- (4.5,0.2292);
\draw[black, thick] (4.5,0.2292) -- (4.625,0.2292);
\draw[black, thick] (4.625,0.2292) -- (4.625,0.2083);
\draw[black, thick] (4.625,0.2083) -- (4.75,0.2083);
\draw[black, thick] (4.75,0.2083) -- (4.75,0.1875);
\draw[black, thick] (4.75,0.1875) -- (4.875,0.1875);
\draw[black, thick] (4.875,0.1875) -- (4.875,0.1667);
\draw[black, thick] (4.875,0.1667) -- (5,0.1667);
\draw[black, thick] (5,0.1667) -- (5,0.1458);
\draw[black, thick] (5,0.1458) -- (5.125,0.1458);
\draw[black, thick] (5.125,0.1458) -- (5.125,0.125);
\draw[black, thick] (5.125,0.125) -- (5.25,0.125);
\draw[black, thick] (5.25,0.125) -- (5.25,0.1042);
\draw[black, thick] (5.25,0.1042) -- (5.375,0.1042);
\draw[black, thick] (5.375,0.1042) -- (5.375,0.0833);
\draw[black, thick] (5.375,0.0833) -- (5.5,0.0833);
\draw[black, thick] (5.5,0.0833) -- (5.5,0.0625);
\draw[black, thick] (5.5,0.0625) -- (5.625,0.0625);
\draw[black, thick] (5.625,0.0625) -- (5.625,0.0417);
\draw[black, thick] (5.625,0.0417) -- (5.75,0.0417);
\draw[black, thick] (5.75,0.0417) -- (5.75,0.0208);
\draw[black, thick] (5.75,0.0208) -- (6,0.0208);
\end{tikzpicture}
\caption{$\t=0.25$}
\label{fig:M2linB}
\end{subfigure}
\caption{Prices of mechanisms $\M_{\textsc{c},\textnormal{lin}}$ (blue) and $\M_{\textsc{pc},\textnormal{lin}}$ (black) when $h=2.8, \l=10,$ and $T=12$.}
\label{fig:M1and2lin}
\end{figure}
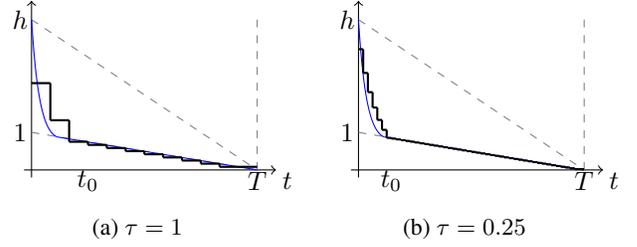

% LASCEREI La discussione piu vaga nell footnote
%
%Recall that the length $\t$ has been calibrated by the seller according to some constraint. 
%%
%In particular, the posted price has to be constant for a time interval with length greater that or equal to $\t$. 
%%
%When $\lceil \frac{T}{\t} \rceil \ne \lfloor \frac{T}{\t} \rfloor$, the last interval has a length lower than the one imposed by the constraint, therefore, the mechanism has to post the same price in the last two intervals.
%%
% Notice that, choosing the price of $(1-\frac{\lfloor \frac{T}{\t} \rfloor \t}{T})$ for the last two intervals, is equivalent to give away the item for free. This would imply a null revenue that is certainly lower that the one expected by maintaining in the last interval the price posted in the previous one, that is $(1-\frac{(i-1)\t}{T})$.

Before proving our main result, we need the following lemma, which is the analogous of Lemma~\ref{lem:interv} working for mechanism $\M_\textsc{pc}$ instead of $\M_\textsc{c}$.

%
%Theorem \ref{LB_step} provides a lower bound for this mechanism that is constant with respect to the choice of the monotone non-decreasing hazard rate distribution of agents' valuations $F$.

\begin{restatable}{lemma}{ex}\label{lem:ex}
	In the RV setting with agents' initial valuations drawn from a distribution $F$, given $0 < \epsilon < 1$, there exists $i = 1,\ldots,\lceil \log_{\d}h\rceil $ such that the price $ p_{\M_\textsc{pc}}(I_i)$ posted by $\M_\textsc{pc}$ during the interval $I_i$ lies in the range $\left[\frac{\nu}{\delta}\xi(i\t),\nu\xi(i\t)\right]$, where $\nu \coloneqq \max\{1,\E[X_{\lambda T}](1-\epsilon)\}$.
\end{restatable}

Now, we provide our main result.
The idea behind its proof is similar to the one used for Theorem \ref{thm:bound_M1}.

\begin{restatable}{theorem}{LBstep}\label{thm:LB_step}
	%
	%$(\l T)^\epsilon \ge \log_{\delta}h$ condizione eliminata 
	%
	Consider the RV setting with $\lambda\tau = (\lambda T)^{1-\epsilon} \ge 1- \ln (e-1) $ for some $\t \in (0,T]$ and $0 < \epsilon < 1$.
	Then, restricted to the set $\F$ of distributions $F$ satisfying the MHR condition, mechanism $\M_\textsc{pc}$ has a competitive ratio that can be lower bounded as follows:
	\[
	\rho(\mathcal{M}_\textsc{pc}) \geq \frac{\xi((\lceil \log_{\d}h\rceil + 1) T^{1-\epsilon}\lambda^{-\epsilon}) (1-\epsilon)}{\delta e}.
	\]
\end{restatable}

%\textcolor{cerise}{The assumption $(\l T)^\epsilon \ge \lceil\log_{\delta}h\rceil$ is always satisfied. Given $t_0 = \t \lceil\log_{\delta}h\rceil$ we have: $(\l T)^\epsilon \ge \lceil\log_{\delta}h\rceil$ <=> $(\l T)^\epsilon \ge \frac{t_0}{\t}$ <=> $(\l T)^\epsilon \ge \frac{\l t_0}{(\l T)^{1-\epsilon}}$ <=> $(\l T)^{\epsilon + 1 - \epsilon} \ge \l t_0$ <=> $T \ge t_0$ that is always true.}

\section{Empirical Evaluation}

% \paragraph{Experimental Setting} 
%
We evaluate mechanisms $\M_{\textsc{c}}$, $\M_{\textsc{pc}}$, and a natural adaption of the ESoES mechanism by~\citet{babaioff2017posting} to \emph{stochastic settings} with no time discounting (called ESoES-SS).
The pricing strategy of ESoES-SS is defined as follows. 
First, we compute the prices of ESoES by setting the number of agents equal to the expected number $\l T$ of agents arriving in $[0,T]$ according to a Poisson process of parameter $\l$. 
Then, ESoES-SS proposes the price that ESoES would propose to the $i$-th agent arrived if $i \leq \l T$ and $1$ otherwise.

We use the following parameters values for the experiments: $\l \in \{1,\ldots, 20\}, T \in \{10,20,50,100\}$, and $h\in \{2,\ldots, 20\}$. 
The following results do not consider time discounting so as to have a fair comparison between our mechanisms and ESoES-SS.
% and to show how the source of stochasticity affects the performance of ESoES-SS. 
%
Further results with a linear discount function are provided in the Appendix.
%

%We adopt Monte Carlo Sampling to estimate the empiric revenue provided by $\M_1$, $\M_2$, and ESoES-SS using 1,000 samples for every combination of the parameters. 

\paragraph{Result \#1}
We study a RV setting with a uniform probability distribution over $[1,h]$.
For every combination of values of $\lambda, T, h$, we run $1000$ Monte Carlo simulations, evaluating the revenue provided by mechanisms ESoES-SS, $\M_{\textsc{C}}$, and $\M_{\textsc{PC}}$.
In particular, we analyze some variants of mechanisms $\M_{\textsc{PC}}$ differing for the number of subintervals (\emph{i.e.}, $Nsub$) in which $[0, t_0]$ is partitioned.
Furthermore, we normalize the revenue provided by the mechanisms in each simulation with respect to $h$.
We report the results in Figure~\ref{figure:empiricalevaluationbabaioff1} for $T=10$ and $T=50$, when $h = 10$.
The results obtained for different values of $h$ are similar.
$\M_{\textsc{C}}$ and $\M_{\textsc{PC}}$ with $Nsub = 232$ have overlapping performances that beat those of the other mechanisms. 
$\M_{\textsc{PC}}$ with $Nsub = 13$ has a performance close to that of the previous two mechanisms, showing that mechanism $\M_{\textsc{PC}}$ provides good performances even with few subintervals.
$\M_{\textsc{PC}}$ with $Nsub = 4$ and ESoES-SS have almost overlapping performances, showing that very few subintervals are sufficient to $\M_{\textsc{PC}}$ to match the performances of ESoES-SS.
The worst mechanism is $\M_{\textsc{PC}}$ with $Nsub = 2$.
The loss of ESoES-SS w.r.t.~$\M_{\textsc{C}}$ averaged over the values of $\lambda$ is about $0.3\, h$ when $T=10$, and $0.4\, h$ when $T=50$.
Surprisingly, the performances of ESoES-SS seem to do not strictly depend on $\lambda$ and $T$.

\begin{figure}[t]
	\centering
	\begin{subfigure}{0.232\textwidth}
		\includegraphics[width=\textwidth]{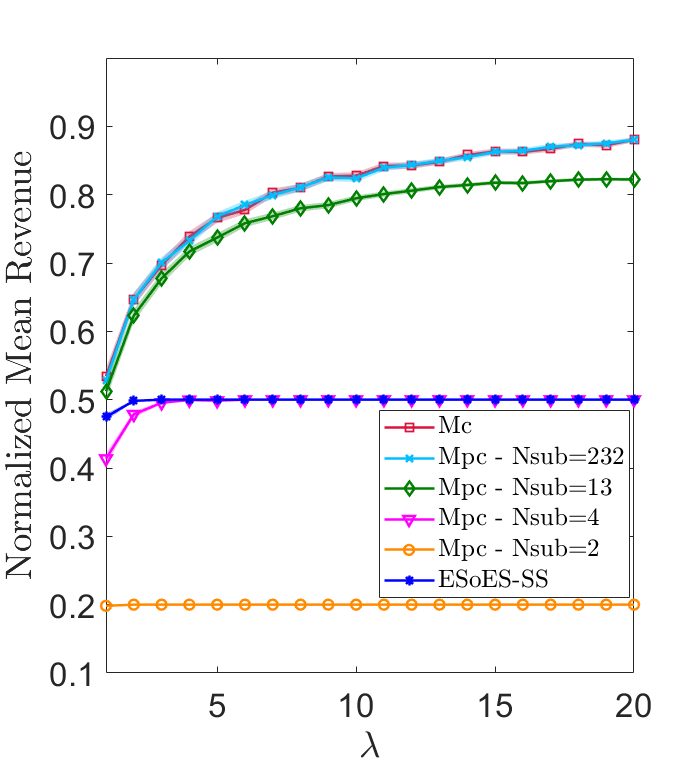}
		\caption{$T=10$} 
	\end{subfigure}
	\begin{subfigure}{0.232\textwidth}
		\includegraphics[width=\textwidth]{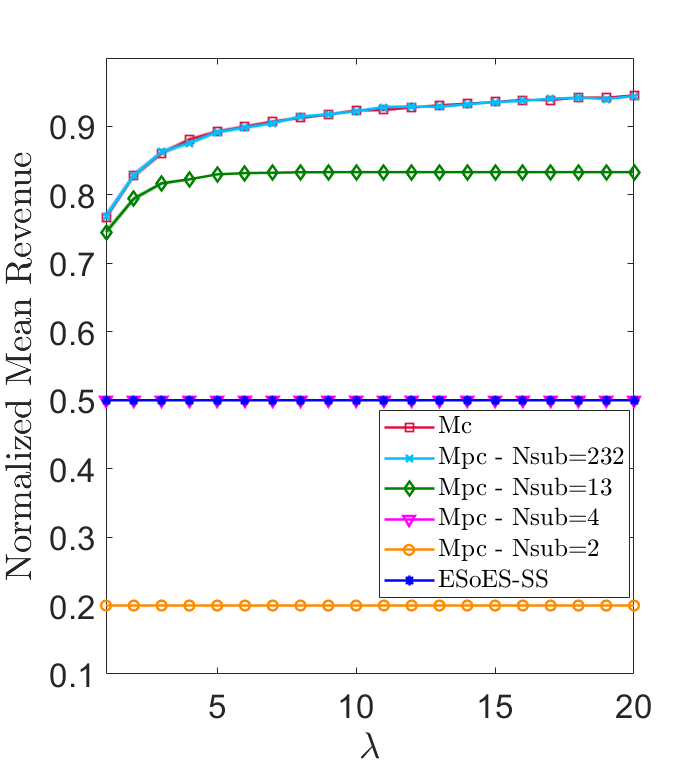}
		\caption{$T=50$}
	\end{subfigure}
	\caption{Average normalized revenue of $\M_{\textsc{C}}$, $\M_{\textsc{PC}}$, ESoES-SS in a RV setting w. uniform distribution ($h=10$).}
	\label{figure:empiricalevaluationbabaioff1}
\end{figure}

\paragraph{Result \#2} 
We study an IV setting.
For every combination of values of $\lambda, T, h$, and for every $v \in \{1.0, 1.5, 2.0,\ldots,h\}$, we run $1000$ Monte Carlo simulations, evaluating the normalized revenue provided by mechanisms ESoES-SS and $\M_{\textsc{C}}$.
%
% Furthermore, we normalize the revenue provided by the mechanisms in every simulation with respect to $h$.
%
For every combination of values of $\lambda, T, h$, we calculate $\max_v \frac{\E_v[\R(\M_{\textsc{c}})]-\E_v[\R(\textnormal{ESoES-SS})]}{h}$, corresponding to the maximum normalized loss of ESoES-SS w.r.t.~$\M_{\textsc{c}}$ over all valuations $v$, and $\max_v \frac{\E_v[\R(\textnormal{ESoES-SS})]-\E_v[\R(\M_{\textsc{c}})]}{h}$, corresponding to the maximum normalized loss of $\M_{\textsc{c}}$ w.r.t.~ESoES-SS over all valuations $v$.
These two indexes are shown in Figure~\ref{figure:empiricalevaluationbabaioff2} for $T=10$ and $T=50$, when $h = 10$.
The results obtained for different values of $h$ are similar.
The loss of ESoES-SS w.r.t.~$\M_{\textsc{c}}$ is always larger than $0.5\, h$ except when both $\lambda$ and $T$ assume small values, while the loss of $\M_{\textsc{c}}$ w.r.t.~ESoES-SS is negligible. 
Furthermore, the two losses converge to two constants as $\lambda$ and $T$ increase.
This shows that, even if there are some special settings where ESoES-SS performs better than $\M_{\textsc{c}}$, the improvement is negligible. 
Instead, mechanism $\M_{\textsc{c}}$, which is designed to deal with stochastic arrivals, provides a very significant improvement.
In particular, we observe that the difference between the revenue provided by ESoES-SS and that provided by $ \M_{\textsc{c}}$ is maximized for small values of $v$ close to 1, while between $ \M_{\textsc{c}}$ and ESoES-SS for large values of $v$ close to $h$.

\begin{figure}[!htp]
	\centering
	\begin{subfigure}{0.232\textwidth}
		\includegraphics[width=\textwidth]{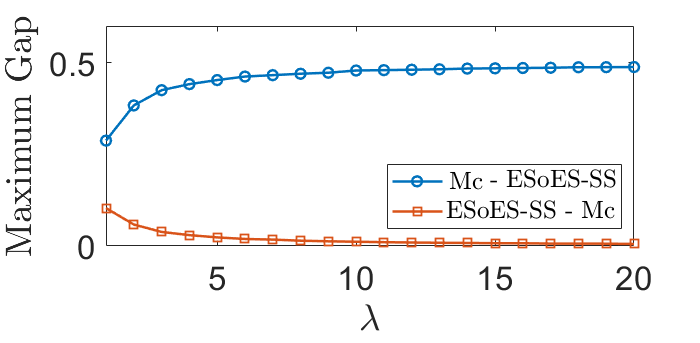}
		\caption{$T=10$} 
	\end{subfigure}
	\begin{subfigure}{0.232\textwidth}
		\includegraphics[width=\textwidth]{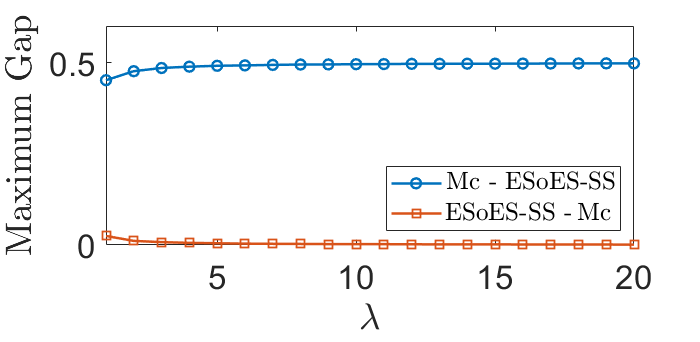}
		\caption{$T=50$}
	\end{subfigure}
	\caption{Maximum difference between the normalized revenues of $\M_{\textsc{C}}$ and ESoES-SS in an IV setting ($h=10$).}
	\label{figure:empiricalevaluationbabaioff2}
\end{figure}

\section{Conclusion and Future Works}

We study distribution-free posted-price mechanisms in order to sell a unique item within a finite time period.
In our model, the agents arrive online according to a Poisson process, and their valuations for the item are discounted over time.
Following a worst-case competitive analysis, we design a mechanism $\M_{\textsc{c}}$ providing an optimal competitive ratio in the identical valuation setting.
Then, as for the random valuation setting, we analyze the performances of $\M_{\textsc{c}}$ and of a new mechanism $\M_\textsc{pc}$ that is constrained to set constant prices during time intervals having a given minimum length.
We prove that both mechanisms achieve a competitive ratio that is constant with respect to the actual valuation when the distribution of the valuations has a monotone hazard rate.
This shows that our mechanisms are robust even in non-stationary markets subject to arbitrary distribution changes preserving the same support. 
%
%We empirically evaluate our mechanisms in different scenarios and compare them with a benchmark.

%An interesting research direction for future works is to explore the combination of methods capable of learning probability distributions in a given class from data with our techniques used to design mechanisms that are  robust with respect to that class.

%An interesting future research direction is to study our mechanisms in different settings that allow to learn the distribution. We could exploit their robustness in a first phase where parameters are still unknown and then analyze how to progressively combine them with methods suited for the specific distribution.

In future, we will investigate hybrid settings in which our robust mechanisms can be combined with machine learning tools.
For instance, data could be used to learn a class of distributions, and we could design a mechanism robust with respect to all the distributions of that class.

\section{Ethical Impact}

Posted-price mechanisms are widely adopted in real-world economic transactions, thanks to their simplicity: a seller posts prices and buyers arrive sequentially, deciding whether to accept the offer or not.
Nowadays, most e-commerce websites implement this form of interaction with their users.
Our mechanisms apply to concrete scenarios where the probability distribution of buyers' valuations is unknown, the value of item for sale may decrease over time, and buyers' arrivals are stochastic.
In these settings, our mechanisms can make economic transactions more efficient and robust, allowing agents (buyers and sellers) to find better economic agreements. 
As we argued in the paper, our mechanisms provide theoretical guarantees in terms of online worst-case performance.
This could have an arguably positive societal impact when applied to real-world economic problems.
However, further research in this direction is required to prevent scenarios with an unbalanced reward structure, where agreements may just award one side (buyers or sellers) with the largest utilities at the expense of the others. 

\bibliography{bibliography}

\begin{thebibliography}{27}
\providecommand{\natexlab}[1]{#1}
\providecommand{\url}[1]{\texttt{#1}}
\providecommand{\urlprefix}{URL }
\expandafter\ifx\csname urlstyle\endcsname\relax
  \providecommand{\doi}[1]{doi:\discretionary{}{}{}#1}\else
  \providecommand{\doi}{doi:\discretionary{}{}{}\begingroup
  \urlstyle{rm}\Url}\fi

\bibitem[{Adamczyk et~al.(2017)Adamczyk, Borodin, Ferraioli, Keijzer, and
  Leonardi}]{adamczyk2017sequential}
Adamczyk, M.; Borodin, A.; Ferraioli, D.; Keijzer, B.~D.; and Leonardi, S.
  2017.
\newblock Sequential posted-price mechanisms with correlated valuations.
\newblock \emph{ACM Transactions on Economics and Computation (TEAC)} 5(4):
  1--39.

\bibitem[{Babaioff et~al.(2017)Babaioff, Blumrosen, Dughmi, and
  Singer}]{babaioff2017posting}
Babaioff, M.; Blumrosen, L.; Dughmi, S.; and Singer, Y. 2017.
\newblock Posting prices with unknown distributions.
\newblock \emph{ACM Transactions on Economics and Computation (TEAC)} 5(2):
  1--20.

\bibitem[{Babaioff et~al.(2009)Babaioff, Dinitz, Gupta, Immorlica, and
  Talwar}]{babaioff2009secretary}
Babaioff, M.; Dinitz, M.; Gupta, A.; Immorlica, N.; and Talwar, K. 2009.
\newblock Secretary problems: weights and discounts.
\newblock In \emph{Proceedings of the twentieth annual ACM-SIAM symposium on
  Discrete algorithms}, 1245--1254. SIAM.

\bibitem[{Babaioff et~al.(2015)Babaioff, Dughmi, Kleinberg, and
  Slivkins}]{babaioff2015dynamic}
Babaioff, M.; Dughmi, S.; Kleinberg, R.; and Slivkins, A. 2015.
\newblock Dynamic pricing with limited supply.
\newblock \emph{ACM Transactions on Economics and Computation (TEAC)} 3(1):
  1--26.

\bibitem[{Barlow and Marshall(1964)}]{barlow1964}
Barlow, R.~E.; and Marshall, A.~W. 1964.
\newblock Bounds for Distributions with Monotone Hazard Rate, II.
\newblock \emph{Annals of Mathematical Statistics} 35(3): 1258--1274.

\bibitem[{Borodin and El-Yaniv(2005)}]{borodin2005online}
Borodin, A.; and El-Yaniv, R. 2005.
\newblock \emph{Online computation and competitive analysis}.
\newblock Cambridge University Press.

\bibitem[{Chawla et~al.(2010)Chawla, Hartline, Malec, and
  Sivan}]{chawla2010multi}
Chawla, S.; Hartline, J.~D.; Malec, D.~L.; and Sivan, B. 2010.
\newblock Multi-parameter mechanism design and sequential posted pricing.
\newblock In \emph{Proceedings of the forty-second ACM symposium on Theory of
  computing}, 311--320.

\bibitem[{Correa et~al.(2017)Correa, Foncea, Hoeksma, Oosterwijk, and
  Vredeveld}]{correa2017posted}
Correa, J.; Foncea, P.; Hoeksma, R.; Oosterwijk, T.; and Vredeveld, T. 2017.
\newblock Posted price mechanisms for a random stream of customers.
\newblock In \emph{Proceedings of the 2017 ACM Conference on Economics and
  Computation}, 169--186.

\bibitem[{Einav et~al.(2018)Einav, Farronato, Levin, and
  Sundaresan}]{einav2018auctions}
Einav, L.; Farronato, C.; Levin, J.; and Sundaresan, N. 2018.
\newblock Auctions versus posted prices in online markets.
\newblock \emph{Journal of Political Economy} 126(1): 178--215.

\bibitem[{Gatti, Di~Giunta, and Marino(2008)}]{DBLP:journals/ai/GattiGM08}
Gatti, N.; Di~Giunta, F.; and Marino, S. 2008.
\newblock Alternating-offers bargaining with one-sided uncertain deadlines: an
  efficient algorithm.
\newblock \emph{Artificial Intelligence} 172(8-9): 1119--1157.

\bibitem[{Hajiaghayi, Kleinberg, and Sandholm(2007)}]{hajiaghayi2007automated}
Hajiaghayi, M.~T.; Kleinberg, R.; and Sandholm, T. 2007.
\newblock Automated online mechanism design and prophet inequalities.
\newblock In \emph{Proceedings of the Twenty-Second {AAAI} Conference on
  Artificial Intelligence}, 58--65.

\bibitem[{Kleinberg and Leighton(2003)}]{kleinberg2003value}
Kleinberg, R.; and Leighton, T. 2003.
\newblock The value of knowing a demand curve: Bounds on regret for online
  posted-price auctions.
\newblock In \emph{44th Annual IEEE Symposium on Foundations of Computer
  Science, 2003. Proceedings.}, 594--605. IEEE.

\bibitem[{Lavi and Nisan(2004)}]{lavi2004competitive}
Lavi, R.; and Nisan, N. 2004.
\newblock Competitive analysis of incentive compatible on-line auctions.
\newblock \emph{Theoretical Computer Science} 310(1-3): 159--180.

\bibitem[{Lucier(2017)}]{lucier2017economic}
Lucier, B. 2017.
\newblock An economic view of prophet inequalities.
\newblock \emph{ACM SIGecom Exchanges} 16(1): 24--47.

\bibitem[{Mao et~al.(2018)Mao, Zheng, Wu, and Chen}]{DBLP:conf/ijcai/MaoZWC18}
Mao, W.; Zheng, Z.; Wu, F.; and Chen, G. 2018.
\newblock Online Pricing for Revenue Maximization with Unknown Time Discounting
  Valuations.
\newblock In \emph{Proceedings of the Twenty-Seventh International Joint
  Conference on Artificial Intelligence}, 440--446.

\bibitem[{Mason and Välimäki(2011)}]{MASON20111699}
Mason, R.; and Välimäki, J. 2011.
\newblock Learning about the arrival of sales.
\newblock \emph{Journal of Economic Theory} 146(4): 1699 -- 1711.

\bibitem[{Mitrinovic, Pecaric, and Fink(2013)}]{mitrinovic2013classical}
Mitrinovic, D.~S.; Pecaric, J.; and Fink, A.~M. 2013.
\newblock \emph{Classical and new inequalities in analysis}, volume~61.
\newblock Springer Science \& Business Media.

\bibitem[{Mohri and Munoz(2014)}]{mohri2014optimal}
Mohri, M.; and Munoz, A. 2014.
\newblock Optimal regret minimization in posted-price auctions with strategic
  buyers.
\newblock In \emph{Advances in Neural Information Processing Systems},
  1871--1879.

\bibitem[{Parkes(2007)}]{parkes2007online}
Parkes, D.~C. 2007.
\newblock \emph{Online mechanisms}.
\newblock Cambridge University Press.

\bibitem[{Pinsky and Karlin(2010)}]{pinsky2010introduction}
Pinsky, M.; and Karlin, S. 2010.
\newblock \emph{An introduction to stochastic modeling}.
\newblock Academic press.

\bibitem[{Rong, Qin, and An(2018)}]{rong2018dynamic}
Rong, J.; Qin, T.; and An, B. 2018.
\newblock Dynamic pricing for reusable resources in competitive market with
  stochastic demand.
\newblock In \emph{Proceedings of the Thirty-Second {AAAI} Conference on
  Artificial Intelligence}, 4718--4726.

\bibitem[{Rosenthal(2011)}]{RosenthalRealEstate}
Rosenthal, E.~C. 2011.
\newblock A Pricing Model for Residential Homes with {Poisson} Arrivals and a
  Sales Deadline.
\newblock \emph{The Journal of Real Estate Finance and Economics} 42(2):
  143--161.

\bibitem[{Ross et~al.(1996)Ross, Kelly, Sullivan, Perry, Mercer, Davis,
  Washburn, Sager, Boyce, and Bristow}]{ross1996stochastic}
Ross, S.~M.; Kelly, J.~J.; Sullivan, R.~J.; Perry, W.~J.; Mercer, D.; Davis,
  R.~M.; Washburn, T.~D.; Sager, E.~V.; Boyce, J.~B.; and Bristow, V.~L. 1996.
\newblock \emph{Stochastic processes}, volume~2.
\newblock Wiley New York.

\bibitem[{Rubinstein(1982)}]{10.2307/1912531}
Rubinstein, A. 1982.
\newblock Perfect Equilibrium in a Bargaining Model.
\newblock \emph{Econometrica} 50(1): 97--109.

\bibitem[{Seifert(2006)}]{seifert2006posted}
Seifert, S. 2006.
\newblock \emph{Posted price offers in internet auction markets}, volume 580.
\newblock Springer Science \& Business Media.

\bibitem[{Shah, Johari, and Blanchet(2019)}]{Shah2019Semi}
Shah, V.; Johari, R.; and Blanchet, J. 2019.
\newblock Semi-Parametric Dynamic Contextual Pricing.
\newblock In \emph{Advances in Neural Information Processing Systems},
  2363--2373.

\bibitem[{Wang(1993)}]{wang1993auctions}
Wang, R. 1993.
\newblock Auctions versus posted-price selling.
\newblock \emph{The American Economic Review} 838--851.

\end{thebibliography}

\clearpage
\onecolumn
\appendix

\begin{center}
	\LARGE\bf Appendix
\end{center}

\section{Omitted Proofs for the IV Setting}\label{app:iv}

\decr*

\begin{proof}
	We only need to prove the result for mechanisms $\M$ whose undiscounted price $\frac{p_{\M}(t)}{\xi(t)}$ is \emph{not} non-increasing in $t$, otherwise the statement of the lemma is trivially true.
	The main idea of the proof is to let the time period $[0,T]$ be evenly partitioned into time intervals of length $\t$ such that the undiscounted price function of $\M$ is constant in each interval.
	This is w.lo.g. if we take $\t \rightarrow 0$.
	Then, there must be two consecutive time intervals, namely  $I_1 \coloneqq I_{s,\t}$ and $I_2 \coloneqq I_{s+\t,\t}$ for some starting time $s \in [0, T- \t]$, such that there exist $p_1 < p_2 \in [1,h]$ with $\frac{p_{\M}(t)}{\xi(t)} = p_1$ and $\frac{p_{\M}(t)}{\xi(t)} = p_2$ during $I_1$ and $I_2$, respectively (otherwise the undiscounted price would be non-increasing).
	Now, let us define a mechanism $\M^\prime$ whose undiscounted price function is the same as that of $\M$, except for the fact that $\frac{p_{\M^\prime}(t)}{\xi(t)} = p_2$ during $I_1$ and $\frac{p_{\M^\prime}(t)}{\xi(t)} = p_1$ during $I_2$ (\emph{i.e.}, intuitively, we exchange the values in the two intervals so as to make the undiscounted price non-increasing in that window of time).

	We show that the expected revenue provided by $\M^\prime$ is always greater than or equal to that achieved by $\M$, as long as $\tau \rightarrow 0$.
	In order to compare the expected revenues of the two mechanisms, it is sufficient to focus on the window of time $I_1 \cup I_2$, where their price functions differ.
	Given $p_1$ and $p_2$, we can partition the agents' valuations $v \in [1,h]$ into three different subsets, as follows:
	%, depending on the values of the corresponding discounted valuations over $I_1 \cup I_2$, as follows:
	%
	\begin{itemize}
		\item $v < p_1$, implying that $v \, \xi(t) < p_{\M}(t)$ and $v \, \xi(t) < p_{\M^\prime}(t)$ for every time instant $t \in I_1 \cup I_2$;
		\item $p_1 \leq v \leq p_2$, implying that $p_{\M}(t) \leq v \, \xi(t) \leq p_{\M^\prime}(t)$ for every time instant $t \in I_1$ and $p_{\M^\prime}(t) \leq v \, \xi(t) \leq p_{\M}(t)$ for every time instant $t \in I_2$;
		\item $v > p_2$, implying that $v \, \xi(t) > p_{\M}(t)$ and $v \, \xi(t) > p_{\M^\prime}(t)$ for every time instant $t \in I_1 \cup I_2$.
	\end{itemize}
	In the first case, $\E_v [ \mathcal{R}(\M) ] - \E_v [ \mathcal{R}(\M^\prime) ] = 0$, since both $\M$ and $\M^\prime$ achieve an expected revenue equal to $0$ during the time window $I_1 \cup I_2$, given that the item is never sold in that window (as both $p_{\M}(t)$ and $p_{\M^\prime}(t)$ are always higher than the agents' discounted valuation $v \, \xi(t)$).
	As for the second case, let us assume $p_1 < v < p_2$ (since the cases $v=p_1$ and $v=p_2$ are analogous).
	Then, $\M$ can sell the item only during the interval $I_1$, while $\M^\prime$ can sell the item only during the other interval $I_2$.
	Thus, we have the following:
	\[
	\E_v [ \mathcal{R}(\M) ] - \E_v [ \mathcal{R}(\M^\prime) ] = \int_{s}^{s+\t} p_1 \xi(t) \, \l e^{-\l t} dt - \int_{s+\t}^{s+2\t} p_1 \xi(t) \, \l e^{-\l t} dt,
	\]
	which goes to $0$ as long as $\t \rightarrow 0$, given that $\xi$ is continuous. 
	%
	% $p_2$, and, thus, both mechanisms achieve an expected revenue of $p_2 \left( 1-e^{\l d\t} \right)$ over  $I_1 \cup I_2$, where we recall that $\left( 1-e^{\l d\t} \right)$ is the probability that at least one agent arrives in a time interval of length $d\t$.
	%
	% As a result, $\E_v [ \mathcal{R}(\M) ] - \E_v [ \mathcal{R}(\M^\prime) ] = 0$ also in the second case.
	%
	Finally, in the third case, we can compute the difference between the expected revenues of the two mechanisms as follows:
	\begin{align*}
	\E_v [\mathcal{R}(\M)-\E_v [\mathcal{R}(\M^\prime)] & = \int_{s}^{s+\t} p_1 \xi(t) \, \l e^{-\l t} dt + \int_{s+\t}^{s+2\t} p_2 \xi(t) \, \l e^{-\l t} dt \\
	& \quad\quad -  \int_{s}^{s+\t} p_2 \xi(t) \, \l e^{-\l t} dt - \int_{s+\t}^{s+2\t} p_1 \xi(t) \, \l e^{-\l t} dt \\
	& = \left( p_1-p_2 \right) \int_{s}^{s+\t} \xi(t) \, \l e^{-\l t} dt  - (p_1 - p_2) \int_{s+\t}^{s+2\t} \xi(t) \, \l e^{-\l t} dt\\
	& = ( p_1 - p_2 ) \left[ \int_{s}^{s+\t} \xi(t) \, \l e^{-\l t} dt  - \int_{s+\t}^{s+2\t} \xi(t) \, \l e^{-\l t} dt  \right],
	\end{align*}
	which is less than or equal to $0$ as $\t \rightarrow 0$, by continuity of $\xi$.
	%
	%	\begin{align*}
	%		\E_v [\mathcal{R}(\M)-\E_v [\mathcal{R}(\M^\prime)] & = p_1 \left( 1-e^{\l \t} \right) + p_2 \left( 1-e^{\l \t} \right) e^{\l \t} \\
	%		& \quad\quad -  p_2 \left( 1-e^{\l \t} \right) - p_1 \left( 1-e^{\l \t} \right) e^{\l \t} \\
	%		& = \left( p_1-p_2 \right) \left( 1-e^{\l \t} \right)^2 \\
	%		& \le 0,
	%	\end{align*}
	%
	% where we remark that $\left( 1-e^{\l d\t} \right) e^{\l d\t} $ is the probability that at least one agent arrives in the second interval $I_2$ when no agent has arrived in the first interval $I_1$.
	%

	By re-iterating the procedure on all the pairs of consecutive infinitesimal intervals (since $\t \rightarrow 0$) defined as $I_1$ and $I_2$ (each time using the last mechanism $\M^\prime$ as the new $\M$), we can render the undiscounted price function non-increasing, obtaining a final mechanism $\M^\prime$ such that $\E_v [ \mathcal{R}(\M) ] \le \E_v [ \mathcal{R}(\M^\prime) ]$ for every possible agents' valuation $v \in [1,h]$.
\end{proof}

\CRconst*

\begin{proof}
	By contradiction, suppose that $\M$ is \emph{not} optimal, \emph{i.e.}, there exists another deterministic posted-price mechanism $\M'$ such that $\CR(\M')>\CR(\M)$.
	According to Proposition~\ref{prop:minimum} and Lemma~\ref{lem:decr}, $\M^\prime$ must be defined by a pricing strategy $p_{\M^\prime}$ such that the undiscounted price $\frac{p_{\M^\prime}(t)}{\xi(t)}$ is non-increasing in $t$ and the minimum price is selected for a time interval $[t_0',T] \subseteq [0,T]$ having non-zero length (recall that $\CR_v(\M) > 0$ does not depend on $v$ and $\CR(\M) = \min_{v \in [1,h]} \CR_v(\M)$).
	
	\textbf{Case $t_0' \ge t_0$.}
	Let us consider the valuation $v=1$.
	Then, we have that the expected revenue of mechanism $\M$ is $\E_v [\R(\M)]= \int_{t_0}^T \xi(t) \, \l e^{- \l t} \dd t $ (accounting for the case in which an agent arrives at $t \geq t_0$ and buys the item at price $\xi(t)$), which is greater than or equal to the expected revenue of mechanism $\M^\prime$, defined as $\E_v [\R(\M')]= \int_{t_0^\prime}^T \xi(t) \, \l e^{- \l t} \dd t $.
	Intuitively, $\E_v [\R(\M)] \geq \E_v [\R(\M')]$ since $\M'$ posts the minimum price for a period of time shorter than that of $\M$.
	Therefore, it holds $\CR(\M^\prime) \leq \CR_v(\M') \leq \CR_v(\M)  \leq \CR(\M)$, which is a contradiction.
	
	\textbf{Case $t_0' < t_0$.}
	First, suppose that there exists a time instant $t' \in [0,t_0^\prime]$ defined as $t^\prime \coloneqq \sup \{ t \in [0, t_0^\prime] \mid p_\M(t) < p_{\M^\prime}(t) \}$, \emph{i.e.}, the last time instant in which $p_\M(t)$ changes from being less than $p_{\M^\prime}(t)$ to being larger than or equal to $p_{\M^\prime}(t)$.
	Clearly, it holds $p_\M(t) \ge p_{\M^\prime}(t)$ for every $t \in [0,T] : t>t'$.
	Moreover, let us consider the agents' valuation $v \in [1,h]$ such that $v \, \xi(t') = p_{\M}(t')$ and focus on the case in which $p_\M(t) = p_{\M^\prime}(t)$ (as the other cases are analogous).
	Notice that, for every time instant $t \leq t^\prime$, mechanism $\M^\prime$ cannot sell the item, since, by using Lemma~\ref{lem:decr}, we get:
	\[
	v \, \xi(t) < v p_{\M'}(t) \frac{\xi(t')}{p_{\M'}(t')} = v p_{\M'}(t) \frac{\xi(t')}{p_{\M}(t')} \leq v p_{\M'}(t) \frac{\xi(t')}{v \, \xi(t')} \leq p_{\M'}(t) .
	\] 
	%	
	%	\textcolor{red}{Nota che questo funziona sia se ci sono più intersezioni tra $p(t)$ e $v \xi(t)$ perchè in questo caso considero l'ultima, sia se la price strategy $p(t)$ è a scalini. Specificarlo?}
	%	
	%	Moreover, let us consider the valuation $v$ such that $v \, \xi(t') = p(t')$.
	%	%
	%	\textcolor{red}{Mechanism $\M'$ cannot sell the item before $t'$.}
	%	
	%	
	Additionally, with an analogous reasoning we can show that, for all the times $t \in [0,T] : t>t'$, both mechanisms may sell the item, but the price posted by $\M'$ is always less than or equal to that chosen by $\M$, with a non-empty time interval in which the former is strictly less than the latter (as $t_0' < t_0$).
	Thus, in this case, it holds $\CR_v(\M)>\CR_v(\M')$, which implies that $\CR(\M') <\CR(\M)$, a contradiction. 
	Finally, it remains to analyze the case in which a time instant $t'$ defined above does \emph{not} exist.
	Since the undiscounted price functions are non-increasing by Lemma~\ref{lem:CR_const} and $t_0' < t_0$, it must be the case that there is no intersection point between the two functions.
	Hence, it must be $p_\M(t)>p_{\M'}(t)$ for all $ t \in [0,t_0]$, which implies that $\CR(\M')<\CR(\M)$ by taking $v=h$.
	This leads to a contradiction.
	%
	% We have shown the contradiction proving that $\M$ achieves a higher worst case $CR$ in all the possible scenarios.
\end{proof}

\bestCRgen*

\begin{proof}
	By Lemma~\ref{lem:CR_const} and using $\CR_v(\M_{\textsc{c}}) = \frac{\E_v[\R(\M_{\textsc{c}})] }{\E_v [\R(\M^\star)]}$, it is sufficient to search for an optimal mechanism $\M_{\textsc{c}}$ whose pricing strategy $p_{\M_{\textsc{c}}}$ is such that the expected revenue of the mechanism is linearly dependent in $v$, \emph{i.e.}, for every valuation $v \in [1,h]$, it must be the case that:
	\[
	\E_v \left[ \R(\mathcal{M}_{\textsc{c}}) \right] = kv,
	\]
	where $k > 0$ is a suitably defined constant that does depend on $v$.
	%
	% Our aim is to define a pricing strategy $p(t):[0, T] \rightarrow \mathbb{R}^{+}$ that achieves a constant competitive ratio. Therefore, the expected revenue of the mechanism has to be expressed as:
	% \[\E[\R(\mathcal{M}_{\textsc{c}})]=kv\]
	% where $k$ depends only on the parameters of the problem.
	%
	In the following, for the ease of presentation, we omit the index $\M_{\textsc{c}}$ from $p_{\M_{\textsc{c}}}$ as the mechanism is clear from the context.
	
	From Proposition~\ref{prop:minimum}, there must be a $t_0 \in [0,T)$ such that $p(t) = \xi(t)$ for every $t \in [t_0, T]$, otherwise $\CR_v(\M_{\textsc{c}}) = 0$ for the valuation $v = 1$.
	Thus, it remains to define $p(t)$ for $t \in [0,t_0)$.
	%
	% We know that the minimum price is offered in $[t_0,T]$, hence, we have to define a price strategy $p(t):[0, t_{0}] \rightarrow \mathbb{R}^{+}$ such that the total expected revenue in is equal to $kv$.
	
	For any valuation $v \in [1,h]$, by letting $t^\ast \coloneqq \sup \{ t \in [0,t_0] \mid p(t) > v \, \xi(t) \}$, we can express the expected revenue of the mechanism $\M_{\textsc{c}}$ as a function of $t^\ast$.
	First, notice that, it holds $p(t^\ast) = v \, \xi(t^\ast)$.
	Moreover, by using Lemma~\ref{lem:decr}, it must be the case that $p(t) > v \, \xi(t)$ for every $t < t^\ast$, since:
	\[
	v \, \xi(t) < v \, p(t) \frac{\xi(t^\ast)}{p(t^\ast)} = v \, p(t) \frac{\xi(t^\ast)}{v \, \xi(t^\ast)} = p(t).
	\]
	As a result, the item is never sold before time $t^\ast$, which allows us to write the following:
	\begin{equation*}
	\E_v \left[ \R(\M_{\textsc{c}}) \right] = e^{\l t^{*}} \int_{t^{*}}^{t_{0}} p(t) \, \l e^{-\l t} \dd t+e^{\l t^{*}} \int_{t_{0}}^{T}\xi(t) \, \l e^{-\l t} \dd t.
	\end{equation*}
	Thus, since we want $\E_v \left[ \R(\mathcal{M}_{\textsc{c}}) \right] = kv$, by using $v=\frac{p(t^*)}{\xi(t^*)} $ and letting $\zeta(t) \coloneqq\frac{1}{\xi(t)}$, we get:
	\begin{equation}\label{eq:const_rev_gen} 
	e^{\l t^{*}} \int_{t^{*}}^{t_{0}} p(t) \, \l e^{-\l t} \dd t+e^{\l t^{*}} \int_{t_{0}}^{T}\xi(t) \, \l e^{-\l t} \dd t = k  \zeta(t^*) p\left(t^{*}\right) .
	\end{equation}
	%		
	%	We can express $v$ as a function of $t^*$, recalling that $p(t^*)=v\xi(t^*)$. Hence, $v=\frac{1}{\xi(t^*)} p(t^*)$. We call $\zeta(t)=\frac{1}{\xi(t)}$. We have to solve the following equation in order to find $p(t)$.
	%	\begin{align} 
	%	e^{\l t^{*}} \int_{t^{*}}^{t_{0}} p(t) \cdot \l e^{-\l t} d t+e^{\l t^{*}} \int_{t_{0}}^{T}\xi(t) \l e^{-\l t} d t=k  \zeta(t^*) p\left(t^{*}\right) \label{const_rev_gen}
	%	\end{align}
	%	where the left side is $\E[\R(\mathcal{M}_{\textsc{c}})]$ and the right side is $kv$.
	%
	%
	By deriving the left-hand side of Equation~\eqref{eq:const_rev_gen} with respect to $t^*$,  we get:
	\begin{align*}
	\frac{\dd \E_v [\R(\mathcal{M}_{\textsc{c}})]}{\dd t^*}&=e^{\l t^*} \frac{\dd G(t^*)}{\dd t^*}+\l\left[e^{\l t^*} \int_{t^*}^{t_0} p(t) \, \l e^{-\l t} \dd t + \l e^{\l t^*} \int_{t_0}^{T}\xi(t) \, \l e^{-\l t} \dd t\right]\\
	&= -\l p(t^*) +\l k\zeta(t^*)  p(t^*)
	\end{align*}
	where $G\left(t^*\right) \coloneqq \int_{t^*}^{t_0} p(t) \l e^{-\l t} \dd t=\int_{t_0}^{t^*} - p(t) \l e^{-\l t} \dd t=\int_{t_0}^{t^*} g(t) \dd t$, with $g(t) \coloneqq - p(t) \l e^{-\l t}$.
	By applying the fundamental theorem of calculus, we have that $\frac{\dd G(t^*)}{\dd t^*}=g(t^*)=-p(t^*)\l e^{-\l t^*}$.
	Thus, the last equality is readily obtained by noticing that the term in the squared brackets is exactly equal to the expected revenue $\E_v [\R(\mathcal{M}_{\textsc{c}})]$, which, in turn, must be equal to $k\zeta(t^*)p(t^*)$.
	Furthermore, by deriving the right-hand side of Equation~\eqref{eq:const_rev_gen} with respect to $t^\ast$, we get:
	\[ 
	\frac{\dd}{\dd t^*}\left[ k\zeta(t^*)p(t^*) \right] = k \zeta'(t^*)p(t^*)+k \zeta(t^*) p'(t^*).
	\]
	By equating the derivatives of the two sides of Equation~\eqref{eq:const_rev_gen}, we get the following differential equation:
	\begin{equation} \label{eq_dif_gen}
	p^{\prime}\left(t^{*}\right)=\left[\lambda-\frac{\lambda}{k \zeta\left(t^{*}\right)}-\frac{\zeta^{\prime}\left(t^{*}\right)}{\zeta\left(t^{*}\right)}\right] p\left(t^{*}\right)
	\end{equation}
	By solving Equation~\eqref{eq_dif_gen} for $p(t)$, we obtain the function:
	\[
	p(t)=a \, e^{\int \left[\l-\frac{\l}{k \zeta(t)} - \frac{\zeta'(t)}{\zeta(t)}\right] dt}, 
	\]
	and, from the boundaries conditions $p(0)=h$ and $p(t_0)=\xi(t_0)$, we can derive constants $a$ and $k$.
	Notice that the condition $p(0)=h$ can be derived from the fact that, if $p(0) < h$, then the expected revenue $\E_v [ \R(\M_{\textsc{c}}) ]$ is the same for all the valuations $v \in [1,h]$ such that $p(0) \leq v \leq h$, which is not possible since we want that $\E_v [ \R(\M_{\textsc{c}}) ]$ linearly depends on $v$.
	
	We recall that $\E_v [\R(\M_{\textsc{c}})]=kv$ for all $v \in [1,h]$.
	Thus, we can use this in order to find $t_0$ as a function of the problem parameters $\l$, $T$, $h$, and function $\xi$.
	Using $v=1$, we get:
	\begin{equation}\label{eq_k_gen}
	\int_{0}^{T-t_{0}}\xi(t) \lambda e^{-\lambda t} d t= k, 
	\end{equation}
	which gives $t_0$ after replacing $k$ with the expression we got from the boundaries conditions.
	%	
	%	We can finally provide the expression of the deterministic posted-price mechanism that achieves the best competitive ratio:
	%	\[ p(t)=\left\{\begin{array}{ll}A \cdot e^{\int\left[\l'-\frac{\lambda}{k \zeta'(t)}- \frac{\zeta'(t)}{\zeta(t)}\right] d t}& t \in\left[0, t_{0}\right) \\ \xi(t) & t \in\left[t_{0}, T\right]\end{array}\right. \]
\end{proof}

\corCRgen*

\begin{proof}
	Let us recall that, from the proof of Theorem~\ref{thm:bestCR_gen}, $\M_{\textsc{c}}$ is characterized by the same ratio $\CR_v(\M_{\textsc{c}})$ for all $v \in [1,h]$.
	Hence, we can calculate the competitive ratio by taking $v=1$:
	\begin{equation}\label{eq:diff_k_gen}
	\CR(\M_{\textsc{c}}) = \frac{\E_v [\R(\mathcal{M}_{\textsc{c}})]}{\E_v [\R(\M^\star)]} = \frac{k}{k^\star}= \frac{\int_{0}^{T-t_{0}} \xi(t) \l e^{-\l t} \dd t}{\int_{0}^{T} \xi(t) \l e^{-\l t} \dd t} ,
	\end{equation}
	where we used Equation~\eqref{eq:bench_IV} and Equation~\eqref{eq_k_gen} from the proof of Theorem~\ref{thm:bestCR_gen}. 	
\end{proof}

\bestCRlin*

\begin{proof}
	We follow the line of the proof of Theorem~\ref{thm:bestCR_gen}, \emph{i.e.}, we look for a mechanism $\M_{\textsc{c}, \textnormal{lin}}$ such that $\E_v [\R(\M_{\textsc{c}, \textnormal{lin}})] = k v$ for every $v \in [1,h]$, where $k > 0$ is suitably defined constant that does not depend on $v$.
	For the ease of presentation, we omit the subscript $\M_{\textsc{c}, \textnormal{lin}}$ from the pricing strategy $p_{\M_{\textsc{c}, \textnormal{lin}}}$.
	%	
	%	Our aim is to define a pricing strategy $p(t):[0, T] \rightarrow \mathbb{R}^{+}$ that achieves a constant competitive ratio. Therefore, the expected revenue of the mechanism has to be expressed as: 
	%	\[\E[\R(\mathcal{M}_{\textsc{c}})]=kv\]
	%	where $k$ depends only on the parameters of the problem.\\
	%	We know that the minimum price is offered in $[t_0,T]$, hence, we have to define a price strategy $p(t):[0, t_{0}] \rightarrow \mathbb{R}^{+}$ such that the total expected revenue in is equal to $kv$. We can express $v$ as a function of $t^*$, recalling that $p(t^*)=v\big(1-\frac{t^*}{T}\big)$. We have to solve the following equation in order to find $p(t)$.
	
	Let us fix $v \in [1,h]$.
	By defining $t^\ast$ as in the proof of Theorem~\ref{thm:bestCR_gen}, since in this case the discount is $\xi_\textnormal{lin} (t) = 1 - \frac{t}{T}$ for $t \in [0,T]$, we have $p(t^*)=v \, \left( 1-\frac{t^*}{T} \right) $, which allows us to write the following:
	\begin{equation}\label{eq:const_rev}
	e^{\l t^{*}} \int_{t^{*}}^{t_{0}} p(t) \, \l e^{-\l t} \dd t+e^{\l t^{*}} \int_{t_{0}}^{T}\left(1-\frac{t}{T}\right) \l e^{-\l t} \dd t=k  \frac{T}{T-t^{*}}  p\left(t^{*}\right) ,
	\end{equation}
	where the left-hand side is the expected revenue $\E_v [\R(\M_{\textsc{c}, \textnormal{lin}})]$ and the right-hand side is $kv$.
	By deriving with respect to $t^*$ the left-hand side of the Equation~\eqref{eq:const_rev}, we get:
	\begin{align*}
	\frac{\dd \E_v[\R(\M_{\textsc{c}, \textnormal{lin}})]}{\dd t^*}&=e^{\l t^*} \frac{\dd G(t^*)}{\dd t^*}+\l\left[e^{\l t^*} \int_{t^*}^{t_0} p(t) \l e^{-\l t} \dd t + \l e^{\l t^*} \int_{t_0}^{T}\left(1-\frac{t}{T}\right)\l e^{-\l t} \dd t\right]\\
	&= -\l p(t^*) +\l k\frac{T}{T-t^*}  p(t^*), 
	\end{align*}
	where $G\left(t^*\right)$ is defined as in the proof of Theorem~\ref{thm:bestCR_gen}.
	%	
	%	where $G\left(t^*\right)=\int_{t^*}^{t_0} p(t) \l e^{-\l t} d t=\int_{t_0}^{t^*} - p(t) \l e^{-\l t} d t=\int_{t_0}^{t^*} g(t) d t$. Applying the fundamental theorem of calculus, we have that $\frac{d G(t^*)}{dt^*}=g(t^*)=-p(t^*)\l e^{-\l t^*}$. Then, note that the term in the squared brackets is equal to $\E[\R(\mathcal{M}_{\textsc{c}})]=k\frac{T}{T-t^*}  p(t^*)$.
	%
	%
	Now, we derive the right-hand side of Equation~\eqref{eq:const_rev} with respect to $t^\ast$:
	\[ 
	\frac{\dd}{\dd t^*}\left(\frac{k T}{T-t^*} p\left(t^*\right)\right)=\frac{k T}{T-t^*} p'\left(t^*\right)+\frac{k T}{\left(T-t^*\right)^{2}} p\left(t^*\right) .
	\]
	By equating the derivatives of the two sides of Equation~\eqref{eq:const_rev}, we get the following differential equation:
	\begin{equation} \label{eq:eq_dif}
	p^{\prime}\left(t^{*}\right)=\left[\lambda-\frac{\lambda\left(T-t^{*}\right)}{k T}-\frac{1}{T-t^{*}}\right] p\left(t^{*}\right).
	\end{equation}
	After solving Equation~\eqref{eq:eq_dif} for $p(t)$, we obtain the general solution:
	\[
	p(t)=a \, e^{\int\left[\lambda-\frac{\lambda\left(T-t\right)}{k T}-\frac{1}{T-t}\right] \dd t}=a \, e^{\lambda\left(1-\frac{1}{k}\right) t+\frac{\lambda}{2 k T} t^{2}+\ln (T-t)}, 
	\]
	where, using boundary conditions $p(0)=h$ and $p(t_{0})=1-\frac{t_{0}}{T}$, we can derive the expressions $a \coloneqq \frac{h}{T}$ and $k \coloneqq \lambda t_{0} \frac{2 T-t_{0}}{2 T\left(\lambda t_{0}+\ln (h)\right)}$.

	Since $\E_v[\R(\M_{\textsc{c}, \textnormal{lin}})]=kv$ for all $v \in [1,h]$, we can use the equation in order to define $t_0$ with respect to the problem parameters $\l$, $T$ and $h$.
	For $v=1$:
	\begin{equation}\label{eq:eq_k}
	\int_{0}^{T-t_{0}}\left(1-\frac{t}{T}\right) \lambda e^{-\lambda t} \dd t=1-\frac{1}{\lambda T}\left(1+\lambda t_{0}-e^{-\lambda\left(T-t_{0}\right)}\right)=k,
	\end{equation}
	and, by replacing $k$ with the expression we got from the boundary conditions, we obtain:
	\begin{equation}\label{eq:eq_t0} 
	1-\frac{1}{\lambda\, T}\left(1+\lambda\, t_{0}-e^{-\lambda\,\left(T-t_{0}\right)}\right)=\lambda \,t_{0} \frac{2 \,T-t_{0}}{2 \,T\left(\lambda \,t_{0}+\ln (h)\right)} 
	\end{equation}
	Finally, we can define $t_0$ as the {unique positive real root} of Equation~\eqref{eq:eq_t0}.
	In particular, it is easy to show that Equation~\eqref{eq:eq_t0} always admits a positive real root in the range $(0,T)$. Indeed,
	we call $q(x)=\lambda x \frac{2 T-x}{2 T\left(\lambda x+\ln (h)\right)} - 1 +\frac{1}{\lambda T}\left(1+\lambda x-e^{-\lambda\left(T-x\right)}\right)$. We observe that $q(x)$ is continuous on the interval $[0,T]$ and that $q(T)>0$ and $q(0)<0$, therefore, for Bolzano's theorem, there exists al least a $t_0 \in (0,T)$ such that $q(t_0)=0$. The uniqueness can be derived as consequence of Lemma \ref{lem:CR_const}.
\end{proof}

	%	We can finally provide the expression of the deterministic posted-price mechanism that achieves the best competitive ratio:
	%	\[ p(t)=\left\{\begin{array}{ll}h \cdot\left(1-\frac{t}{T}\right) \cdot e^{\lambda\left(1-\frac{1}{k}\right)t+\frac{\lambda}{2 k T} t^{2}} & t \in\left[0, t_{0}\right) \\ 1-\frac{t}{T} & t \in\left[t_{0}, T\right]\end{array}\right. \]
	%	
	%	
	%\begin{figure}[h]
	%	\begin{tikzpicture}[domain=0:2]
	%	\draw[->] (-0.2,0) -- (5.5,0) node at (5.8,-0.3) {$t$};
	%	\draw[->] (0,-0.2) -- (0,4.5)node[left] {$p$};
	%	\draw[gray, thin] (0,4) -- (5,0);
	%	\draw[gray, thin] (0,1) -- (5,0);
	%	\draw[gray, thin] (5,0) -- (5,4);
	%	\node at (-0.3,4) {$h$};
	%	\node at (-0.3,1) {$1$};
	%	\node at (5,-0.3) {$T$};
	%	\node at (1.714,-0.3) {$t_{0}$};
	%	\draw[color=black, thick, domain=0:1.714] plot (\x,{4*(1-\x/5)*exp(-1.38983050847*\x+0.338983*\x^2)}); %    %node[right] {$p(t)$};
	%	\draw[black, thick] (1.714,0.657) -- (5,0);
	%	\end{tikzpicture}
	%	\centering
	%	\caption{The Upper Bound (\emph{II}): mechanism $\M_{\textsc{c}}$}
	%	\label{upperbounddue}
	%	\centering
	%\end{figure}	
	%

\constantratiodue*

\begin{proof}
	We can calculate it by taking $v=1$:
	\begin{align*}
	\CR(\M_{\textsc{c}, \textnormal{lin}}) = \frac{\E_v [\R(\M_{\textsc{c}, \textnormal{lin}})]}{\E_v [\R(\M^\star)]}&= \frac{k}{k^\star}= \frac{\int_{0}^{T-t_{0}}\left(1-\frac{t}{T}\right) \lambda e^{-\lambda t} \dd t}{\int_{0}^{T}	\left(1-\frac{t}{T}\right) \lambda e^{-\lambda t} \dd t} \\
	&=\frac{1-\frac{1}{\lambda T}\left(1+\lambda t_{0}-e^{-\lambda\left(T-t_{0}\right)}\right)}{1-\frac{1}{\lambda T}\left(1-e^{-\lambda T}\right)} ,
	\end{align*}
	where we used Equation~\eqref{eq:bench_IV} and Equation~\eqref{eq:eq_k} from the proof of Theorem~\ref{thm:bestCR_lin}. 
\end{proof}

\section{Examples of Mechanisms \textnormal{$\M_{\textsc{c}, \textnormal{lin}}$ and $\M_{\textsc{pc}, \textnormal{lin}}$} and Competitive Ratio Analysis}\label{app:M1lin}

In order to ease the reader in the understanding of our mechanisms, we provide their graphical representation for the case of a linear discount function $\xi_\textnormal{lin}(t) \coloneqq 1-\frac{t}{T} $.
In particular, we focus on mechanisms $\M_{\textsc{c},\textnormal{lin}}$ and $\M_{\textsc{pc},\textnormal{lin}}$, where the latter is the linear-discount version of the general-discount mechanism $\M_\textsc{pc}$.
The price function $p_{\M_{\textsc{pc},\textnormal{lin}}}$ of $\M_{\textsc{pc},\textnormal{lin}}$ can be easily obtained from that of $\M_\textsc{pc}$ by using the specific definition of the discount function.
We report it below for completeness.
\begin{equation*}
p_{\M_{\textsc{pc},\textnormal{lin}}}(I_i) \coloneqq \left\{\begin{array}{ll} \frac{h}{\d^i} \left( 1-\frac{i\t}{T} \right) & \textnormal{if } i = 1,\ldots,\lfloor \log_{\d}h\rfloor \\ 
 1-\frac{i\t}{T}  & \textnormal{if }  i = \lceil \log_{\d}h\rceil, \ldots , \left\lceil \frac{T}{\t} \right\rceil -1  \\
 1-\frac{(i-1)\t}{T}  & \textnormal{if } i = \left\lceil \frac{T}{\t} \right\rceil \end{array}\right. .
%~\footnote{Notice that, whenever $T$ is not divisible by $\t$, then the last time interval is shorter than $\t$. Thus, in order to satisfy the minimum length constraint, we set its price equal to the one in the preceeding interval.}
\end{equation*}

We tune the parameters $h$, $\l$, and $T$ so as to simulate real-world scenarios representing the long-term rental of a single room.
In particular, we fix the parameter values by analyzing data from a real-world co-living company operating on the web, {counting over 7000 rooms}.~\footnote{We cannot disclose the name of the company for privacy reasons.}
In this scenario, the goal is to rent a single room to students for a fixed period of one year.
We set $T=12$, assuming that each time interval of length $1$ corresponds to a period of one month, and we fix the starting time $t=0$ as the time in which the contract of the previous tenant ends.
Therefore, the room value is discounted over time as an effect of the ever shorter period of stay of the future tenant.
We also set $h=2.8$, which means that the highest valuation for the room is around three times the lowest one.

Figure~\ref{fig:M1lin} shows how the shape of mechanism $\M_{\textsc{c},\textnormal{lin}}$ changes by varying the arrival rate $\lambda$, which is the expected number of agents arriving in a time interval of one month.
We observe that the price function decreases as a linearly discounted exponential function in the time interval $[0,t_0]$, and, then, as a linear function in $[t_0,T]$. 
Notice that, by comparing Figure~\ref{fig:M1linA} and Figure~\ref{fig:M1linB}, it is easy to see that the time instant $t_0$ gets closer to zero as the arrival rate $\l$ increases. 
This can be explained by recalling that the mechanism has to deal with the trade-off between setting high prices so as to achieve high revenues and posting lower prices in order to increase the probability of selling the item.
In the first period of time, the seller posts high prices hoping for the arrival of an agent having an high valuation.
This phase cannot be too long, otherwise the item risks to remain unsold, and, on the other hand, it cannot even be too short, otherwise the probability of encountering such an high-valuation agent becomes too small.
Therefore, when the arrival rate decreases, the high-price phase must be enlarged in order to still have some chance of concluding the purchase for an high price (Figure~\ref{fig:M1linB}), while, if $\l$ increases, it suffices to post high prices for a shorter time period (Figure~\ref{fig:M1linA}).

Figure~\ref{fig:M2lin} represents the behavior of mechanism $\M_{\textsc{pc},\textnormal{lin}}$ when we impose different constraints on the minimum time in which the price must be constant. 
In particular, Figure~\ref{fig:M2linA_app} and Figure~\ref{fig:M2linB_app} show the shape of $\M_{\textsc{pc},\textnormal{lin}}$ when the posted price does not change for time intervals of length $\t$ equal to one month (\emph{i.e.}, $\t=1$) and one week (\emph{i.e.}, $\t=0.25$), respectively.

\begin{figure}[h]
\begin{subfigure}{0.5\textwidth}
\begin{tikzpicture}[domain=0:2, scale=0.8]
\draw[->] (-0.2,0) -- (6.5,0) node at (6.8,-0.3) {$t$};
\draw[->] (0,-0.2) -- (0,4.5)node[above=0.1cm] {$p_{\M_{\textsc{c},\textnormal{lin}}}(t)$};
\draw[gray, thin] (0,4) -- (6,0);
\draw[gray, thin] (0,1) -- (6,0);
\draw[gray, thin] (6,0) -- (6,4);
\node at (-0.3,4) {$h$};
\node at (-0.3,1) {$1$};
\node at (6,-0.3) {$T$};
\node at (0.688,-0.3) {$t_{0}$};
\draw[color=black, thick, domain=0:0.688] plot (\x,{4*(1-\x/6)*exp(-3.85437997725*\x+2.65453166477*\x^2)}); %node[right] {$p(t)$};
\draw[black, thick] (0.688,0.88533333) -- (6,0);
\end{tikzpicture}
\centering
\caption{$\M_{\textsc{c},\textnormal{lin}}$: $h=2.8, \l=10, T=12$}
\label{fig:M1linA}
\centering
\end{subfigure} \qquad 
\begin{subfigure}{0.5\textwidth}
\begin{tikzpicture}[domain=0:2, scale=0.8]
\draw[->] (-0.2,0) -- (6.5,0) node at (6.8,-0.3) {$t$};
\draw[->] (0,-0.2) -- (0,4.5)node[above=0.1cm] {$p_{\M_{\textsc{c},\textnormal{lin}}}(t)$};
\draw[gray, thin] (0,4) -- (6,0);
\draw[gray, thin] (0,1) -- (6,0);
\draw[gray, thin] (6,0) -- (6,4);
\node at (-0.3,4) {$h$};
\node at (-0.3,1) {$1$};
\node at (6,-0.3) {$T$};
\node at (1.286,-0.3) {$t_{0}$};
\draw[color=black, thick, domain=0:1.286] plot (\x,{4*(1-\x/6)*exp(-1.93791556728*\x+0.667326297274*\x^2)}); %node[right] {$p(t)$};
\draw[black, thick] (1.286,0.785666) -- (6,0);
\end{tikzpicture}
\centering
\caption{$\M_{\textsc{c},\textnormal{lin}}$: $h=2.8, \l=2, T=12$}
\label{fig:M1linB}
\centering
\end{subfigure}
\caption{Mechanism $\M_{\textsc{c},\textnormal{lin}}$ with different rate parameters $\l$.}
\label{fig:M1lin}
\end{figure}
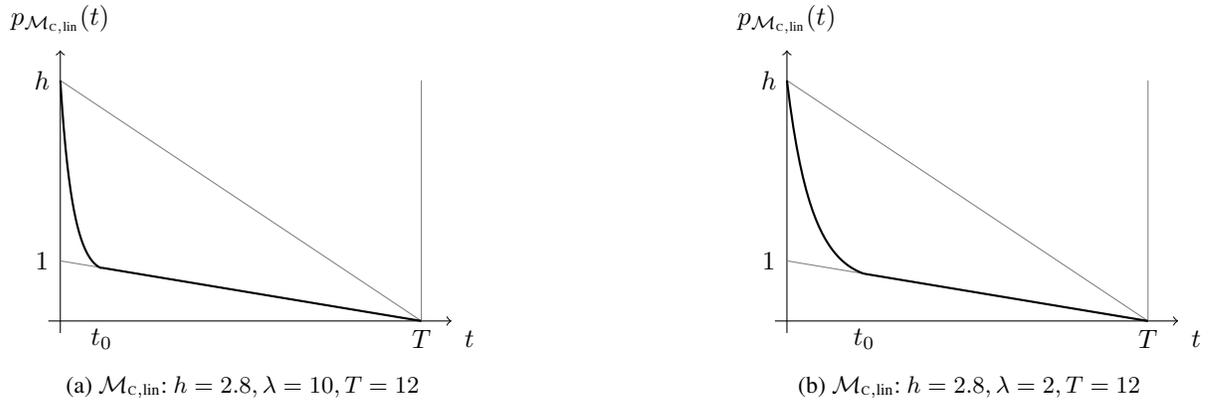

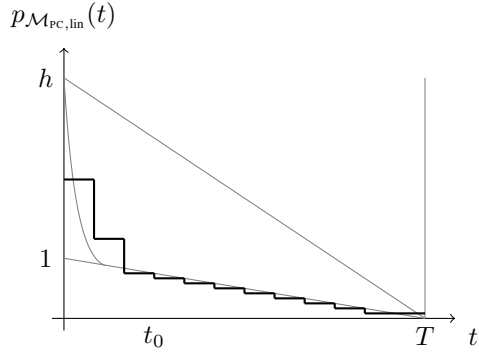
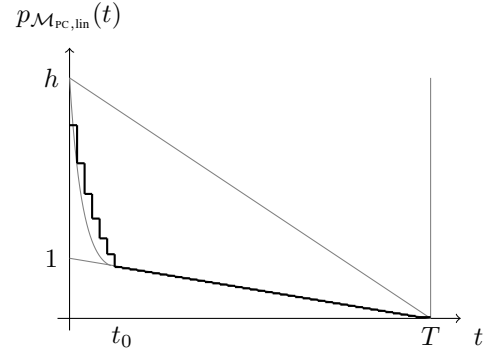
\begin{figure}[h]
\begin{subfigure}{0.5\textwidth}
\begin{tikzpicture}[domain=0:2, scale=0.8]
\draw[->] (-0.2,0) -- (6.5,0) node at (6.8,-0.3) {$t$};
\draw[->] (0,-0.2) -- (0,4.5)node[above=0.1cm] {$p_{\M_{\textsc{pc},\textnormal{lin}}}(t)$};
\draw[gray, thin] (0,4) -- (6,0);
\draw[gray, thin] (0,1) -- (6,0);
\draw[gray, thin] (6,0) -- (6,4);
\node at (-0.3,4) {$h$};
\node at (-0.3,1) {$1$};
\node at (6,-0.3) {$T$};
\node at (1.5,-0.3) {$t_{0}$};
\draw[color=gray, thin, domain=0:0.688] plot (\x,{4*(1-\x/6)*exp(-3.85437997725*\x+2.65453166477*\x^2)}); 
\draw[black, thick] (0,2.31) -- (0.5,2.31);
\draw[black, thick] (0.5,2.31) -- (0.5,1.323);
\draw[black, thick] (0.5,1.323) -- (1,1.323);
\draw[black, thick] (1,1.323) -- (1,0.75);
\draw[black, thick] (1,0.75) -- (1.5,0.75);
\draw[black, thick] (1.5,0.75) -- (1.5,0.6666);
\draw[black, thick] (1.5,0.6666) -- (2,0.6666);
\draw[black, thick] (2,0.6666) -- (2,0.5833);
\draw[black, thick] (2,0.5833) -- (2.5,0.5833);
\draw[black, thick] (2.5,0.5833) -- (2.5,0.5);
\draw[black, thick] (2.5,0.5) -- (3,0.5);
\draw[black, thick] (3,0.5) -- (3,0.417);
\draw[black, thick] (3,0.417) -- (3.5,0.417);
\draw[black, thick] (3.5,0.417) -- (3.5,0.3333);
\draw[black, thick] (3.5,0.3333) -- (4,0.3333);
\draw[black, thick] (4,0.3333) -- (4,0.25);
\draw[black, thick] (4,0.25) -- (4.5,0.25);
\draw[black, thick] (4.5,0.25) -- (4.5,0.167);
\draw[black, thick] (4.5,0.167) -- (5,0.167);
\draw[black, thick] (5,0.167) -- (5,0.0833);
\draw[black, thick] (5,0.0833) -- (6,0.0833);
\end{tikzpicture}
\centering
\caption{$\M_{\textsc{pc},\textnormal{lin}}$: $\t=1, h=2.8, \l=10, T=12$}
\label{fig:M2linA_app}
\centering
\end{subfigure} \qquad 
\begin{subfigure}{0.5\textwidth}
\begin{tikzpicture}[domain=0:2, scale=0.8]
\draw[->] (-0.2,0) -- (6.5,0) node at (6.8,-0.3) {$t$};
\draw[->] (0,-0.2) -- (0,4.5)node[above=0.1cm] {$p_{\M_{\textsc{pc},\textnormal{lin}}}(t)$};
\draw[gray, thin] (0,4) -- (6,0);
\draw[gray, thin] (0,1) -- (6,0);
\draw[gray, thin] (6,0) -- (6,4);
\node at (-0.3,4) {$h$};
\node at (-0.3,1) {$1$};
\node at (6,-0.3) {$T$};
\node at (0.875,-0.3) {$t_{0}$};
\draw[color=gray, thin, domain=0:0.688] plot (\x,{4*(1-\x/6)*exp(-3.85437997725*\x+2.65453166477*\x^2)});
\draw[black, thick] (0,3.213) -- (0.125,3.213);
\draw[black, thick] (0.125,3.213) -- (0.125,2.58);
\draw[black, thick] (0.125,2.58) -- (0.25,2.58);
\draw[black, thick] (0.25,2.58) -- (0.25,2.07);
\draw[black, thick] (0.25,2.07) -- (0.375,2.07);
\draw[black, thick] (0.375,2.07) -- (0.375,1.66);
\draw[black, thick] (0.375,1.66) -- (0.5,1.66);
\draw[black, thick] (0.5,1.66) -- (0.5,1.331);
\draw[black, thick] (0.5,1.331) -- (0.625,1.331);
\draw[black, thick] (0.625,1.331) -- (0.625,1.067);
\draw[black, thick] (0.625,1.067) -- (0.75,1.067);
\draw[black, thick] (0.75,1.067) -- (0.75,0.854);
\draw[black, thick] (0.75,0.854) -- (0.875,0.854);
\draw[black, thick] (0.875,0.854) -- (0.875,0.8333);
\draw[black, thick] (0.875,0.8333) -- (1,0.8333);
\draw[black, thick] (1,0.8333) -- (1,0.8125);
\draw[black, thick] (1,0.8125) -- (1.125,0.8125);
\draw[black, thick] (1.125,0.8125) -- (1.125,0.792);
\draw[black, thick] (1.125,0.792) -- (1.25,0.792);
\draw[black, thick] (1.25,0.792) -- (1.25,0.771);
\draw[black, thick] (1.25,0.771) -- (1.375,0.771);
\draw[black, thick] (1.375,0.771) -- (1.375,0.75);
\draw[black, thick] (1.375,0.75) -- (1.5,0.75);
\draw[black, thick] (1.5,0.75) -- (1.5,0.729);
\draw[black, thick] (1.5,0.729) -- (1.625,0.729);
\draw[black, thick] (1.625,0.729) -- (1.625,0.708);
\draw[black, thick] (1.625,0.708) -- (1.75,0.708);
\draw[black, thick] (1.75,0.708) -- (1.75,0.6875);
\draw[black, thick] (1.75,0.6875) -- (1.875,0.6875);
\draw[black, thick] (1.875,0.6875) -- (1.875,0.6667);
\draw[black, thick] (1.875,0.6667) -- (2,0.6667);
\draw[black, thick] (2,0.6667) -- (2,0.6458);
\draw[black, thick] (2,0.6458) -- (2.125,0.6458);
\draw[black, thick] (2.125,0.6458) -- (2.125,0.625);
\draw[black, thick] (2.125,0.625) -- (2.25,0.625);
\draw[black, thick] (2.25,0.625) -- (2.25,0.6042);
\draw[black, thick] (2.25,0.6042) -- (2.375,0.6042);
\draw[black, thick] (2.375,0.6042) -- (2.375,0.5833);
\draw[black, thick] (2.375,0.5833) -- (2.5,0.5833);
\draw[black, thick] (2.5,0.5833) -- (2.5,0.5625);
\draw[black, thick] (2.5,0.5625) -- (2.625,0.5625);
\draw[black, thick] (2.625,0.5625) -- (2.625,0.5417);
\draw[black, thick] (2.625,0.5417) -- (2.75,0.5417);
\draw[black, thick] (2.75,0.5417) -- (2.75,0.5208);
\draw[black, thick] (2.75,0.5208) -- (2.875,0.5208);
\draw[black, thick] (2.875,0.5208) -- (2.875,0.5);
\draw[black, thick] (2.875,0.5) -- (3,0.5);
\draw[black, thick] (3,0.5) -- (3,0.4792);
\draw[black, thick] (3,0.4792) -- (3.125,0.4792);
\draw[black, thick] (3.125,0.4792) -- (3.125,0.4583);
\draw[black, thick] (3.125,0.4583) -- (3.25,0.4583);
\draw[black, thick] (3.25,0.4583) -- (3.25,0.4375);
\draw[black, thick] (3.25,0.4375) -- (3.375,0.4375);
\draw[black, thick] (3.375,0.4375) -- (3.375,0.4167);
\draw[black, thick] (3.375,0.4167) -- (3.5,0.4167);
\draw[black, thick] (3.5,0.4167) -- (3.5,0.3958);
\draw[black, thick] (3.5,0.3958) -- (3.625,0.3958);
\draw[black, thick] (3.625,0.3958) -- (3.625,0.375);
\draw[black, thick] (3.625,0.375) -- (3.75,0.375);
\draw[black, thick] (3.75,0.375) -- (3.75,0.3542);
\draw[black, thick] (3.75,0.3542) -- (3.875,0.3542);
\draw[black, thick] (3.875,0.3542) -- (3.875,0.3333);
\draw[black, thick] (3.875,0.3333) -- (4,0.3333);
\draw[black, thick] (4,0.3333) -- (4,0.3125);
\draw[black, thick] (4,0.3125) -- (4.125,0.3125);
\draw[black, thick] (4.125,0.3125) -- (4.125,0.2917);
\draw[black, thick] (4.125,0.2917) -- (4.25,0.2917);
\draw[black, thick] (4.25,0.2917) -- (4.25,0.2708);
\draw[black, thick] (4.25,0.2708) -- (4.375,0.2708);
\draw[black, thick] (4.375,0.2708) -- (4.375,0.25);
\draw[black, thick] (4.375,0.25) -- (4.5,0.25);
\draw[black, thick] (4.5,0.25) -- (4.5,0.2292);
\draw[black, thick] (4.5,0.2292) -- (4.625,0.2292);
\draw[black, thick] (4.625,0.2292) -- (4.625,0.2083);
\draw[black, thick] (4.625,0.2083) -- (4.75,0.2083);
\draw[black, thick] (4.75,0.2083) -- (4.75,0.1875);
\draw[black, thick] (4.75,0.1875) -- (4.875,0.1875);
\draw[black, thick] (4.875,0.1875) -- (4.875,0.1667);
\draw[black, thick] (4.875,0.1667) -- (5,0.1667);
\draw[black, thick] (5,0.1667) -- (5,0.1458);
\draw[black, thick] (5,0.1458) -- (5.125,0.1458);
\draw[black, thick] (5.125,0.1458) -- (5.125,0.125);
\draw[black, thick] (5.125,0.125) -- (5.25,0.125);
\draw[black, thick] (5.25,0.125) -- (5.25,0.1042);
\draw[black, thick] (5.25,0.1042) -- (5.375,0.1042);
\draw[black, thick] (5.375,0.1042) -- (5.375,0.0833);
\draw[black, thick] (5.375,0.0833) -- (5.5,0.0833);
\draw[black, thick] (5.5,0.0833) -- (5.5,0.0625);
\draw[black, thick] (5.5,0.0625) -- (5.625,0.0625);
\draw[black, thick] (5.625,0.0625) -- (5.625,0.0417);
\draw[black, thick] (5.625,0.0417) -- (5.75,0.0417);
\draw[black, thick] (5.75,0.0417) -- (5.75,0.0208);
\draw[black, thick] (5.75,0.0208) -- (6,0.0208);
\end{tikzpicture}
\centering
\caption{$\M_{\textsc{pc},\textnormal{lin}}$: $\t=0.25, h=2.8, \l=10, T=12$}
\label{fig:M2linB_app}
\centering
\end{subfigure}
\caption{Mechanism $\M_{\textsc{pc},\textnormal{lin}}$ with different constraints on the minimum time in which the price must be constant.}
\label{fig:M2lin}
\end{figure}

\paragraph{Analysis of the Competitive Ratio $\CR(\M_{\textsc{c}, \textnormal{lin}})$}
From Corollary~\ref{cor:constantratiodue}, we know that, in the IV setting with linear discount function $\xi_\textnormal{lin}$, mechanism $\M_{\textsc{c}, \textnormal{lin}}$ achieves a competitive ratio:
\[
\CR(\M_{\textsc{c}, \textnormal{lin}}) = \frac{1-\frac{1}{\lambda T}\left(1+\lambda t_{0}-e^{-\lambda\left(T-t_{0}\right)}\right)}{1-\frac{1}{\lambda T}\left(1-e^{-\lambda T}\right)},
\]
where $t_0 \in [0,T)$ is defined in Theorem~\ref{thm:bestCR_lin} as the unique positive real root of the equation:
\begin{equation}\label{eq:equation_t0}
	\lambda \, t_{0} \, \frac{2 \,T-t_{0}}{2 \,T\left(\lambda \, t_{0}+\ln h\right)} = 1-\frac{1}{\lambda T}\left(1+\lambda \, t_{0}-e^{-\lambda\left(T-t_{0}\right)}\right).
\end{equation}
We analyze the behavior of the competitive ratio $\CR(\M_{\textsc{c}, \textnormal{lin}})$ when the problem parameters $h,\l$ and $T$ vary.
We summarize our results in the following table:
%
%\begin{center}
%	{\renewcommand{\arraystretch}{1.3}
%	\begin{tabular}{ c|c|c|c| } 
%		%\hline
%		      & $T \rightarrow \infty$ & $\l \rightarrow \infty$ & $h \rightarrow \infty$ \\ 
%		\hline
%		$t_0$ & $t_0 \sim \sqrt{T} $ & $t_0 \sim \frac{1}{\sqrt{\l}} $ & $t_0 \rightarrow T$ \\ 
%		\hline
%		$\lim\CR(\M_{\textsc{c}, \textnormal{lin}})$ & 1 & 1 & 0 \\ 
%		\hline
%	\end{tabular}
%	}
%\end{center}
%\begin{table}[h]
	\begin{center}
		{\renewcommand{\arraystretch}{1.3}
			\begin{tabular}{r|c|c|c}
				& $T \rightarrow \infty$ & $\l \rightarrow \infty$ & $h \rightarrow \infty$ \\
				\hline
				$t_0$ & $\Theta ( \sqrt{T} ) $ & $\Theta( \sqrt{T/\l}) $ & $\Theta( T)$ \\
				$\CR(\M_{\textsc{c}, \textnormal{lin}})$ & $\Theta\left(1-\frac{1}{\sqrt{T}}\right)$ & $\Theta\left(1-\frac{1}{\sqrt{\lambda}}\right)$ & $\Theta\left(\frac{1}{\log^2(h)}\right)$ \\ 
				$\lim\CR(\M_{\textsc{c}, \textnormal{lin}})$ & 1 & 1 & 0 \\ 
			\end{tabular}
		}
	\end{center}
%\end{table}
%
By using Equation~\eqref{eq:equation_t0}, we can conclude that $t_0$ is asymptotically equivalent to $\sqrt{T}$ when $T \to \infty$. Then, the limit of the competitive ratio is:
\[\lim_{T \rightarrow \infty}\CR(\M_{\textsc{c}, \textnormal{lin}}) = \lim_{T \to \infty} 1 - \frac{t_0}{T} = \lim_{T \to \infty} 1 - \frac{1}{\sqrt{T}} = 1.\]
Similarly, $t_0$ is asymptotically equivalent to $\sqrt{\frac{T}{\l}}$ when $\l \to \infty$. Thus:
\[\lim_{\l \rightarrow \infty}\CR(\M_{\textsc{c}, \textnormal{lin}}) = \lim_{\l \to \infty} 1 - \frac{t_0}{T} = \lim_{\l \to \infty} 1 - \frac{1}{\sqrt{T\l}} = 1.\]
Moreover, it is easy to see that $t_0 \rightarrow T$ when $h \to \infty$. Indeed, in this case we have that $t_0$ is the unique positive real root of the equation:
\[1-\frac{1}{\lambda T}\left(1+\lambda \, t_{0}-e^{-\lambda\left(T-t_{0}\right)}\right) = 0. \] 
Since $t_0 \rightarrow T$, the limit of the competitive ratio is:
\[ \lim_{h \rightarrow \infty} \CR(\M_{\textsc{c}, \textnormal{lin}}) = \frac{1-\frac{1}{\lambda T}\left(1+\lambda T-e^{-\lambda\left(T-T\right)}\right)}{1-\frac{1}{\lambda T}\left(1-e^{-\lambda T}\right)} = 0 .\]
Therefore, having a valuation function with finite support is fundamental in order to achieve a certain fraction of the expected revenue of an optimal mechanism. There are no guarantees for valuation functions with unbounded support.
In the following we analyze how $\CR(\M_{\textsc{c},\textnormal{lin}})$ goes to $0$.
%In particular, $\M_{\textsc{c},\textnormal{lin}}$ goes asymptotically to $0$ as $\frac{1}{\log^2(h)}$. We show the calculations below.
%
We first observe that $\CR(\M_{\textsc{c},\textnormal{lin}})$ is proportional to $\frac{\lambda (T-t_0)^2}{2T} + o(T-t_0)^2$ when $h \to \infty$. Indeed, the numerator of $\CR(\M_{\textsc{c},\textnormal{lin}})$ depends on $h$ through $t_0$: 
%
%the numerator of $\CR(\M_{\textsc{c},\textnormal{lin}})$ when $h \to \infty$ and $z=(T-t_0) \to 0$ is:
%
\[ 1-\frac{1}{\lambda T}\left(1+\lambda t_{0}-e^{-\lambda\left(T-t_{0}\right)}\right) = \frac{z}{T} - \frac{1}{\lambda T}(1-e^{-\lambda z}) = \frac{\lambda (T-t_0)^2}{2T} + o(T-t_0)^2, \]
where $z\coloneqq T-t_0 \to 0$ as $h \to \infty$, and the Taylor series $e^{-\lambda z} = 1 - \lambda z + \frac{\lambda^2 z^2}{2} + o(z^2)$ is used to expand function $e^{-\lambda z}$ at $t_0=T$.
Notice that the competitive ratio is decreasing in $t_0$.
We compute $\bar{t}_0$, which is an upper bound for $t_0$, by solving Equation~\eqref{eq:equation_t0} with the exponential term $e^{-\lambda(T-t_0)}$ substituted by parameter a $\varepsilon$. We impose $\varepsilon$ and $e^{-\lambda(T-t_0)}$ to have the same domain, hence $\varepsilon \in (0,1)$. Thus, we obtain the following equation:
\begin{equation}\label{eq:equation_t0_bound}
\lambda \, x \, \frac{2 \,T-x}{2 \,T\left(\lambda \, x+\ln h\right)} = 1-\frac{1}{\lambda T}\left(1+\lambda \, x- \varepsilon \right).
\end{equation}
The solution
\[ x = \frac{\sqrt{2 \lambda T \ln(h) + \ln^2 h + \varepsilon^2 - 2\varepsilon + 1} -\ln h + \varepsilon -1}{\lambda}\]
is increasing in $\varepsilon$, being its first partial derivative positive for all $\varepsilon \in (0,1)$:
\[ \frac{\partial x}{\partial \varepsilon} = \frac{\varepsilon -1}{\lambda\sqrt{2\lambda T \ln(h) + \ln^2 h + (\varepsilon - 1)^2}} + \frac{1}{\lambda}.\] 
Hence, by setting $\varepsilon=1$, we get the following upper bound on $t_0$:
\[\bar{t}_0 \coloneqq \frac{-\ln h + \sqrt{2 \lambda T \ln(h) + \ln^2 h} }{\lambda} = \frac{2 T \ln(h)}{\sqrt{2 \lambda T \ln(h) + \ln^2 h} + \ln h}.\]
Notice that $T-\bar{t}_0$ is a lower bound for $T-t_0$. By asymptotic analysis, as $h \to \infty$ we have:
\[ T-\bar{t}_0 = \frac{ 2 \lambda T^2 \ln h} {2 \ln^2 h + 2\lambda T \ln h + 2 \ln h \sqrt{2\lambda T \ln h + \ln^2 h}} \sim C_1 \frac{1}{\ln(h)} \]
where $C_1$ is constant with respect to $h$ and depends on parameters $\lambda$ and $T$.
Hence, as $h \to \infty$:
\[\CR(\M_{\textsc{c},\textnormal{lin}}) \sim \frac{\lambda (T-t_0)^2}{2T-\frac{2}{\lambda}(1-e^{-\lambda T})} \ge  \frac{\lambda (T-\bar{t}_0)^2}{2T-\frac{2}{\lambda}(1-e^{-\lambda T})} \sim C_2 \frac{1}{\ln^2(h)}, \]
where $C_2$ is a constant with respect to $h$ and depends on parameters $\lambda$ and $T$. We conclude that, as $h \to \infty$, the competitive ratio $\CR(\M_{\textsc{c},\textnormal{lin}})$ converges to $0$ slower than or the same as the function $\frac{1}{\log^2(h)}$.
%:\[\lim_{h \rightarrow \infty} \CR(\M_{\textsc{c},\textnormal{lin}})= \Theta\left(\frac{1}{\log^2(h)}\right)\]

\section{Omitted Proofs for the RV Setting}\label{app:random_valuation}

\interv*

\begin{proof}
	Given how the function $p_{\M_{\textsc{c}}}$ is defined, we can always define a time interval $I_{{s},\tau}$ as desired by selecting its starting time ${s} \in [0,T - \t]$ in such a way that $p_{\M_{\textsc{c}}}({s}) = \E{[X_{\lambda T}]\xi(s+\t)(1-\epsilon)}$.
	From Definition~\ref{def:kappa} we know that $\kappa_\t({s}) \le \kappa_\t$.
	Hence,
	\[
	p_{\M_{\textsc{c}}}({s}+\t)=\frac{p_{\M_{\textsc{c}}}({s})}{\kappa_\t({s})}=\frac{\E{[X_{\lambda T}]\xi(s+\t)(1-\epsilon)}}{\kappa_\t({s})}\ge\frac{\E{[X_{\lambda T}]\xi(s+\t)(1-\epsilon)}}{\kappa_\t}.
	\]
	Since $p_{\M_{\textsc{c}}}$ is non-increasing by Lemma~\ref{lem:decr}, for every $t \in I_{s,\t}$ we have:
	\[
	p_{\M_{\textsc{c}}}(t) \in \left[ p_{\M_{\textsc{c}}}({s}),p_{\M_{\textsc{c}}}({s}+\t) \right] \subseteq \left[\frac{\E{[X_{\lambda T}]\xi(s+\t)(1-\epsilon)}}{\kappa_\t},\E{[X_{\lambda T}]\xi(s+\t)(1-\epsilon)}\right].
	\] 
	Notice that the inequality involving $p_{\M_{\textsc{c}}}({s}+\t)$ holds with equality if ${s} \in \arg\max_{s \in[0,T - \t]}\kappa_\t(s)$.
	If this is the case, then there exists a unique interval verifying the statement.
\end{proof}

\FIHR*

\begin{proof}
	Let us recall that the cumulative distribution function of $X_{\l \t}$ is such that:
	%	\[
	%		F_{X_{\l \t}}(x) = \sum_{i=1}^\infty \Pr(N_\t=i)F_{X_{\l\t}|(N_\t=i)} =\sum_{i=0}^\infty \frac{(\l \t)^i e^{-\l \t}}{i!}(F(x))^i = %e^{-\l\t}\sum_{i=0}^\infty \frac{(\l \t F(x))^i}{i!} = 
	%		e^{-\l\t}e^{\l \t F(x)} = e^{-\l\t(1-F(x))}
	%	\]
	%
	\[
	F_{X_{\l \t}}(x) = e^{-\l\t(1-F(x))}.
	\]
	We compute the hazard rate of $F_{X_{\l \t}}$ and show it is non-decreasing, as follows:
	\begin{align*}
	H_{X_{\l \t}}(x) & = \frac{f_{X_{\l \t}}(x)}{1-F_{X_{\l \t}}(x)} = \frac{\frac{\dd}{\dd x}F_{X_{\l \t}}(x)}{1-F_{X_{\l \t}}(x)} = \frac{\l\t f(x) e^{-\l\t(1-F(x))}}{1-e^{-\l\t(1-F(x))}} \\
	& =	\frac{\l\t f(x)}{e^{\l\t(1-F(x))}-1} 
	= \l\t \frac{f(x)}{1-F(x)} \frac{1-F(x)}{e^{\l\t(1-F(x))}-1} \\
	& = \l\t H(x) \frac{1-F(x)}{e^{\l\t(1-F(x))}-1}.
	\end{align*}
	Since $F$ is MHR, the hazard rate $H(x)$ is non-decreasing.
	Notice that $F(x)$ is non-decreasing, and, thus, $1-F(x)$ in non-increasing.
	As a result, proving that $\frac{1-F(x)}{e^{\l\t(1-F(x))}-1}$ is non-decreasing in $x$ is equivalent to show that $g(y) \coloneqq\frac{y}{e^{\l\t y}-1}$ is non-increasing in $y$.
	We study the first derivative of $g(y)$:
	\[
	\frac{\dd}{\dd y}g(y)= \frac{e^{\l\t y}(1-\l\t y)-1}{(e^{\l\t y}-1)^2} \leq 0 \quad \textnormal{for all } y \in [0,1].
	\] 
	This implies that $g(y)$ is non-increasing in $y$; hence, $\frac{1-F(x)}{e^{\l\t(1-F(x))}-1}$ is non-decreasing in $x$.
	We conclude that $H_{X_{\l \t}}(x)$ is monotone non-decreasing. 
\end{proof}

In order to prove Lemma~\ref{lem:ex_ln}, we first state the following variant of the Chebyshev inequality~\citet{mitrinovic2013classical}, where the adopted notation is specific for the proposition.
%. We need it for the proof of Lemma~\ref{ex_ln}.

\begin{proposition}[\cite{mitrinovic2013classical}]\label{cheby}
	Suppose function $h(x)$ is positive and non-decreasing on $[a,b]$, function $g(x)$ is non-decreasing on $[a,b]$, and function $f(x)$ is continuous on $[a,b]$, then the following inequality holds:
	\[\frac{\int_{a}^{b} h(x)f(x)g(x) \dd x}{\int_{a}^{b}h(x)f(x) \dd x}\ge \frac{\int_{a}^{b}f(x)g(x) \dd x}{\int_{a}^{b}f(x) \dd x}.\]
\end{proposition}
% Lemma~\ref{cheby} is a variant of Chebyshev Inequality \citep{mitrinovic2013classical}.

\Zheng*
\begin{proof}
	Recall that $F_{X_{\l \t}}(x)=e^{-\l\t(1-F(x))}$.
	%
	% from proof of Lemma~\ref{F_IHR}.
	%
	% We need to express the expected value of the maximum valuation of agents arriving in a $\t$-lengthed time interval, with $\t \le \t' \le T$, as follows:
	%
	Then, we can write the following:
	\begin{align*}
		\E[X_{\l \t}] &=\int_{0}^{\infty}x f_{X_{\l \t}}(x) \dd x = \int_{0}^{\infty} 1 - F_{X_{\l \t}}(x) \dd x =  \int_{0}^{\infty} 1 - e^{-\l\t(1-F(x))} \dd x \\
		& = \int_{0}^{\infty} \frac{1-F(x)}{f(x)} \frac{1- e^{-\l\t(1-F(x))}}{1-F(x)} \dd F(x) \\
		&= \int_{0}^{1} \frac{1}{H(F^{-1}(1-\eta))} \frac{1- e^{-\l\t \eta}}{\eta} \dd \eta.
	\end{align*}
	Now, we apply Lemma~\ref{cheby}.
	$F$ having non-decreasing monotone hazard rate implies that $h(\eta)\coloneqq\frac{1}{H(F^{-1}(1-\eta))}$ is a non-decreasing function of $\eta$. Hence, $h(k)$ is non-decreasing and positive on $[0,1]$. $g(\eta)\coloneqq\frac{1- e^{-\l\t \eta}}{1- e^{-\l \t' \eta}}$ is non-decreasing on $[0,1]$ and $f(\eta)\coloneqq\frac{1- e^{-\l \t' \eta}}{\eta}$ is continuous on $[0,1]$. Thus,
	\begin{align*}
	\frac{\E[X_{\l \t}]}{\E[X_{\l \t'}]} &= \frac{\int_{0}^{1} \frac{1}{H(F^{-1}(1-\eta))} \frac{1- e^{-\l\t \eta}}{\eta} \dd \eta}{\int_{0}^{1} \frac{1}{H(F^{-1}(1-\eta))} \frac{1- e^{-\l \t' \eta}}{\eta} \dd \eta} 
	= \frac{\int_{0}^{1} \frac{1}{H(F^{-1}(1-\eta))} \frac{1- e^{-\l \t' \eta}}{\eta}\frac{1- e^{-\l\t \eta}}{1- e^{-\l \t' \eta}} \dd \eta}{\int_{0}^{1} \frac{1}{H(F^{-1}(1-\eta))} \frac{1- e^{-\l \t' \eta}}{\eta} \dd \eta} \\
	&\ge \frac{\int_{0}^{1}\frac{1- e^{-\l\t \eta}}{\eta}\dd \eta}{\int_{0}^{1} \frac{1- e^{-\l \t' \eta}}{\eta} \dd \eta}
	= \frac{\int_{0}^{\l \t}\frac{1- e^{-t}}{t} \dd t}{\int_{0}^{\l \t'} \frac{1- e^{-t}}{t}\dd t} = \frac{Ein(\l \t)}{Ein(\l \t')} \\
	&= \frac{\gamma- Ei(-\l \t) + \ln(\l \t)}{\gamma- Ei(-\l \t') + \ln(\l \t')} \ge \frac{\ln(\l \t)}{\ln(\l \t')},
	\end{align*}
	where $Ein(x)\coloneqq \int_{0}^{x}\frac{1-e^{-t}}{t}\dd t$ is the entire exponential integral function, $Ei(x)\coloneqq\int_{-\infty}^{x} \frac{e^t}{t} \dd t$ is the exponential integral function, and $\gamma\approx0.577$ is the Euler's constant.
\end{proof}

\boundM*

\begin{proof}
	%Consider the prices posted by mechanism $\M_{\textsc{c}}$ and a $\tau$-lengthed time interval $I_{s,\tau}=[s,\tau+s] \subseteq [0,T]$. From Definition~\ref{kappa}, the price posted at $s$ is at most $\kappa$ times the price at the endpoint of the interval. \\
	%
	% Let $(\lambda T)^\epsilon \ge \log_{\kappa}h$
	%
	% Let us define $\t \in (0,T]$ such that $\lambda\tau = (\lambda T)^{1-\epsilon}$, where we recall that $0 <  \epsilon < 1$ satisfies $(\l T)^\epsilon \ge \log_{\kappa}h$. 
	%
	By hypothesis, we have $\lambda\tau = (\lambda T)^{1-\epsilon}$ for some $\t \in (0,T]$ and $0 < \epsilon < 1$, which implies that $1-\epsilon = \frac{\ln(\lambda\tau)}{\ln(\lambda T)}$.
	Moreover, let us fix a distribution $F$ satisfying the MHR condition.
	From Lemma~\ref{lem:interv}, there exists a time interval $I_{s,\tau} $ with starting time $s \in [0,T -\t]$ such that $p_{\M_{\textsc{c}}}(t) \in \left[\frac{\E{[X_{\lambda T}] \xi(s+\t) (1-\epsilon)}}{\kappa},\E[X_{\lambda T}]\xi(s+\t)(1-\epsilon)\right]$ for every $t \in I_{s,\tau}$.
	We distinguish two cases, depending on whether the starting time of the interval is before or after the time $t_0$ characterizing mechanism $\M_{\textsc{c}}$ (as defined in Theorem~\ref{thm:bestCR_gen}).

	\textbf{Case $s < t_0$.}
	By using the fact that the seller's expected revenue for the overall time period is at least that achieved during the interval $I_{s,\t}$, we have:
	\begin{align}
	\E_F [\R(\mathcal{M}_{\textsc{c}})] & \ge p_{\M_{\textsc{c}}}(s+\t)\Pr \left\{  Y_{\lambda\tau}\ge p_{\M_{\textsc{c}}}(s) \right\} \nonumber\\ 
	& \ge p_{\M_{\textsc{c}}}(s+\t)\Pr  \left\{  X_{\lambda\tau} \xi(s+\t) \ge \E{[X_{\lambda T}] \xi(s+\t) (1-\epsilon)} \right\} \label{mag_x2} \\ 
	& = p_{\M_{\textsc{c}}}(s+\t)\Pr \left\{ X_{\lambda\tau}\ge \E{[X_{\lambda T}](1-\epsilon)} \right\} \nonumber\\
	& = p_{\M_{\textsc{c}}}(s+\t)\Pr \left\{  X_{\lambda\tau}\ge \E{[X_{\lambda T}]} \frac{\ln(\lambda\tau)}{\ln(\lambda T)} \right\} \nonumber\\
	& \ge p_{\M_{\textsc{c}}}(s+\t)\Pr \left\{ X_{\lambda\tau}\ge \E{[X_{\lambda\tau}]} \right\} \label{lem} \\
	& \ge \frac{p_{\M_{\textsc{c}}}(s+\t)}{e} \label{1/e}\\ 
	& \ge \frac{\E{[X_{\lambda T}] \xi(s+\t) (1-\epsilon)}}{\kappa_\t e}\nonumber \\
	& \ge \frac{\E{[X_{\lambda T}] \xi(t_0+\t) (1-\epsilon)}}{\kappa_\t e} \nonumber\label{t0}.
	\end{align}
	%
	% where $\E[\R(\mathcal{M}_{\textsc{c}})]_{[a,b]}$ is the expected revenue of the mechanism in an interval $[a,b]$. 
	%
	Equation~\eqref{mag_x2} holds since $X_{\lambda\tau}\xi(s+\t)$ is a random variable representing the maximum initial valuation of agents arriving in a time interval of length $\t$ weighted by the maximum possible discount, and, thus, it is always smaller than or equal to $Y_{\lambda\tau}$.
	%
	% , that is the minimum weight in that interval. $Y_{\lambda\tau}$ is the maximum discounted valuation which means that is some valuation wighted by a discount $\big(1-\frac{t}{T}\big) \ge \big(1-\frac{\tau+s}{T}\big)$.
	%
	Equation~\eqref{lem} follows from Lemma~\ref{lem:ex_ln}.
	Equation~\eqref{1/e} follows from a result by~\citet{barlow1964}, which implies that, for any MHR distribution, the probability of exceeding its expectation is at least $\frac{1}{e}$.
	% for a non-decreasing monotone hazard rate distribution. We use this result in Inequality~(\ref{1/e}).
	%
	% The last inequality is motivated by the fact that $s \in [0,t_0)$ from Lemma \ref{lem:interv}.
	%	
	%	Notice that in order to have a bound that is greater then zero, we should fix parameter $\t$ such that $\t < T - t_0$. 
	%	
	%	\textcolor{cp}{basta che dica così oppure potrebbe essere utile spiegare che questo vincolo e' ragionevole / non troppo restrittivo e far vedere come è fatto il nostro meccanismo fissando dei parametri verosimili di lambda h e T (ad esempio dati simil DoveVivo o comunque sensati)? in quel caso, si dovrebbe vedere che $T-t_0$ è tipo 3/4 dell'intervalle [0,T], quindi la condizione $\t < T-t_0$ non costringe $\t$ ad essere troppo piccolo}
	
	\textbf{Case $s \geq t_0$.}
	In this case, we can lower bound the seller's expected revenue for the overall time period with that obtained during the the interval $I_{t_0, \t}$, as follows:
	\begin{align}
	\E_F [\R(\mathcal{M}_{\textsc{c}})] & \ge p_{\M_{\textsc{c}}}(t_0 + \t) \left( 1 - e^{- \l \t} \right) \nonumber \\
	& \geq \xi(t_0 + \t) \frac{1}{e} \nonumber \\
	& \ge \frac{\E{[X_{\lambda T}] \xi(t_0+\t) (1-\epsilon)}}{\kappa_\t e} \nonumber,
	\end{align}
	where for the first inequality we used the fact that the expected revenue in $I_{t_0, \t}$ is at least the lowest price posted during the interval times the probability that at least one agent arrives in $I_{t_0, \t}$, the second inequality holds since $\left( 1 - e^{- \l \t} \right) \geq \frac{1}{e}$ when $\l \t \ge 1- \ln (e-1) \simeq 0,46$, while the last inequality follows from the fact that $s \geq t_0$.
	Indeed, by Lemma~\ref{lem:interv}, we can write the following:
	\begin{equation*}
	p_{\M_{\textsc{c}}}(s+\t) = \xi(s + \t) \geq \frac{\E{[X_{\lambda T}] \xi(s+\t) (1-\epsilon)}}{\kappa_\t } ,
	\end{equation*} 
	which implies that $\frac{\E{[X_{\lambda T}]  (1-\epsilon)}}{\kappa_\t } \leq 1$.
	
	We can now compute a lower bound on the ratio $\CR_F(\M_{\textsc{c}})$ of mechanism $\M_{\textsc{c}}$, as follows:
	\begin{align}
	\CR_F(\mathcal{M}_{\textsc{c}}) & = \frac{\E_F[\R(\M_{\textsc{c}})]}{\E_F [\R(\M^\star)]} \nonumber\\
	& \ge \frac{\E_F [\R(\M_{\textsc{c}})]}{\E[Y_{\lambda T}]}\nonumber\\
	& \geq\frac{\E[X_{\lambda T}]}{\E[Y_{\lambda T}]} \frac{\xi(t_0 + \t)(1-\epsilon)}{\kappa_\t e}\nonumber\\
	&\ge \frac{\xi(t_0 + \t)(1-\epsilon)}{\kappa_\t e} \nonumber
	\end{align}
	where the first inequality holds since $\E[Y_{\lambda T}]$ is the expected revenue of a mechanism that knows the actual realization of agents' initial valuations and arrival times, \emph{i.e.}, the realization of variable $Y_{\l T}$.-
	This mechanism achieves an expected revenue greater than or equal to that obtained by the benchmark $\M^\star$, since the latter only knows the distribution of valuations.
	As for the second inequality, it is easy to see that $\frac{\E[X_{\lambda T}]}{\E[Y_{\lambda T}]} \ge 1$.
	% and $\lim_{\lambda T \rightarrow \infty} \frac{\E[X_{\lambda T}]}{\E[Y_{\lambda T}]} = 1$, hence we can write Inequality~\ref{beta_1}.\\
	%
	Finally, by recalling the condition $\l\t=(\l T)^{1-\epsilon}$,
	we have $\t=T^{1-\epsilon}\l^{-\epsilon}$,
	which allows us to write the following bound:
	\[
	\CR(\mathcal{M}_{\textsc{c}}) \geq \frac{\xi(t_0 + T^{1-\epsilon}\l^{-\epsilon})(1-\epsilon)}{\kappa_\t e}.
	\]
	% where $t_0$ and $\kappa$ depends on the parameters of the problem, which are $T, \l, h$, and are constant with respect to the distribution $F$.
	%
	This concludes the proof.
\end{proof}

\ex*

\begin{proof}
	In the following, for the ease of presentation, we let $ \tilde{I}_i \coloneqq \left[\frac{\nu}{\delta}\xi(i\t),\nu\xi(i\t)\right]$ for any $i = 1,\ldots,\lceil \log_{\d}h\rceil $.
	By contradiction, suppose that there is no $i = 1,\ldots,\lceil \log_{\d}h\rceil $ such that $ p_{\M_\textsc{pc}}(I_i) \in \tilde{I}_i$.
	Notice that $\nu$ is a lower bound on $\E[X_{\lambda T}]$ and belongs to the range $\in [1,h)$.
	
	We reach a contradiction by employing an iterated reasoning.
	As a first step, we observe that either $\nu \in \left (\frac{h}{\d},h \right)$ or $\nu \in \left[1,\frac{h}{\d} \right]$.
	If $\nu \in \left(\frac{h}{\d},h \right)$, then $p_{\M_\textsc{pc}}(I_1) =  \frac{h}{\d}\xi(\t)$ is in the range $\tilde{I}_1 = \left[\frac{\nu}{\delta}\xi(\t),\nu \xi(\t) \right]$.
	Hence, it must hold $\nu \in [1,\frac{h}{\d}]$. 
	Then, as a second step, we can conclude that either $\nu \in \left(\frac{h}{\d^2},\frac{h}{\d} \right]$ or $\nu \in \left[ 1,\frac{h}{\d^2} \right]$. 
	If $\nu \in \left(\frac{h}{\d^2},\frac{h}{\d} \right]$, then $p_{\M_\textsc{pc}}(I_2) = \frac{h}{\d^2}\xi(2\t)$ is in the range $\tilde{I}_2 = \left[\frac{\nu}{\delta}\xi(2\t),\nu \xi(2\t) \right]$. 
	Hence, it must hold $\nu \in [1,\frac{h}{\d^2}]	$.
	By iterating the reasoning until the $\lfloor \log_{\d}h\rfloor$-th step, we obtain that either $\nu \in \left(\frac{h}{\d^{\lfloor \log_{\d}h\rfloor}},\frac{h}{\d^{\lfloor \log_{\d}h\rfloor - 1}} \right]$ or $\nu \in \left[1,\frac{h}{\d^{\lfloor \log_{\d}h\rfloor}} \right]$.  
	
	Let us first consider the case in which it holds $\lfloor \log_{\d}h\rfloor \ne \lceil \log_{\d}h\rceil$.
	If $\nu \in \left(\frac{h}{\d^{\lfloor \log_{\d}h\rfloor}},\frac{h}{\d^{\lfloor \log_{\d}h\rfloor - 1}} \right]$, then $p_{\M_\textsc{pc}} \left( I_{\lfloor \log_{\d}h\rfloor} \right) \in \tilde{I}_{\lfloor \log_{\d}h\rfloor} $ since:
	\[
	\frac{h}{\d^{\lfloor \log_{\d}h\rfloor}} \xi(\lfloor \log_{\d}h\rfloor \t) \in  \left[\frac{\nu}{\delta}\xi(\lfloor \log_{\d}h\rfloor \t),\nu \xi(\lfloor \log_{\d}h\rfloor \t) \right]. 
	\]
	Hence, it must hold $\nu \in \left[ 1,\frac{h}{\d^{\lfloor \log_{\d}h\rfloor}} \right]$.
	Then, $p_{\M_\textsc{pc}} \left( I_{\lceil \log_{\d}h\rceil} \right)= \xi(\lceil \log_{\d}h\rceil \t)$ belongs to the range $\tilde{I}_{\lceil \log_{\d}h\rceil} = \left[\frac{\nu}{\delta}\xi(\lceil \log_{\d}h\rceil \t),\nu \xi(\lceil \log_{\d}h\rceil \t) \right]$, which leads to a contradiction. 
	%
	%	We cannot go on dividing progressively the interval $[1,\frac{h}{\d^{\lfloor \log_{\d}h\rfloor}}]$ because the next step would be saying that either $\nu \in (\frac{h}{\d^{\lfloor \log_{\d}h\rfloor + 1}},\frac{h}{\d^{\lfloor \log_{\d}h\rfloor}}]$ or $\nu \in [1,\frac{h}{\d^{\lfloor \log_{\d}h\rfloor +1 }}]$. However, notice that $\frac{h}{\d^{\lfloor \log_{\d}h\rfloor +1 }} < 1$, hence, we study the only case left, that is when $\nu \in [1,\frac{h}{\d^{\lfloor \log_{\d}h\rfloor}}]$.
	%	%  
	%	In this last interval there exists a price $p^*_i = p^*_{\lceil \log_{\d}h\rceil} = 1-\frac{\lceil \log_{\d}h\rceil \t}{T}$ that lies in $\tilde{I}_{\lceil \log_{\d}h\rceil} = [\frac{\nu}{\delta}\big(1-\frac{\lceil \log_{\d}h\rceil \tau}{T} \big),\nu \big(1-\frac{\lceil \log_{\d}h\rceil \tau}{T} \big)]$. 
	
	Now, suppose that $\lfloor \log_{\d}h\rfloor = \lceil \log_{\d}h\rceil = \log_{\d}h$.
	Then, in the $\lfloor \log_{\d}h\rfloor$-th step of the iterated reasoning, we can conclude that $\nu \in \left[1,\frac{h}{\d^{\log_{\d}h - 1}} \right]$ and $p_{\M_\textsc{pc}} \left( I_{\log_{\d}h} \right) =\xi((\log_{\d}h) \t)$ is in the range $\tilde{I}_{\log_{\d}h}  = \left[\frac{\nu}{\delta}\xi((\log_{\d}h) \t),\nu \xi((\log_{\d}h) \t) \right]$, which leads to the final contradiction.
	%	
	%	We showed that there is always an interval $I_i$ for $i \in \{1,\dots,\lceil \log_{\d}h\rceil \}$, meaning that $I_i$ is before $t_0$, such that its price $p(I_i)$ lies in $\tilde{I}_i$.
\end{proof}

\LBstep*

\begin{proof}
	% \textcolor{red}{with $(\l T)^\epsilon \ge \log_{\delta}h$ and} omitted condition
	By hypothesis we have $\l \t = (\l T)^{\-\epsilon}$ for $\t=\frac{t_0}{\lceil \log_{\d}h\rceil}$ $\in (0,T]$. 
	We distinguish two cases, depending on whether $\E[X_{\l T}](1-\epsilon)$ is greater or lower than one. Note that $\E[X_{\l T}](1-\epsilon)$ is a lower bound for $\E[X_{\l T}]$ and that one is the minimum value that $\E[X_{\l T}]$ can assume. In particular $\E[X_{\l T}] = 1$ when $F$ is the point distribution such that $P(V_i \le 1) = P(V_i = 1) = 1$.
	
	\textbf{Case $\E[X_{\l T}](1-\epsilon) >= 1 $.}
	For Lemma \ref{lem:ex}, there exists an $i \in \{1,\dots,\lceil \log_{\d}h\rceil \}$ such that the price $p^*_i= p_{\M_\textsc{pc}}(I_i)$ lies in the range $\tilde{I_i} = \Big[\frac{\E{[X_{\lambda T}]\xi(i\tau)(1-\epsilon)}}{\delta},\E[X_{\lambda T}]\xi(i\tau)(1-\epsilon)\Big]$.
	By using the fact that the seller's expected revenue for the overall time period is at least that achieved during the interval $I_i$, we have:
	\begin{align}
	\E[\R(\mathcal{M}_{\textsc{pc}})] &\ge p^*_i\Pr(Y_{\lambda\tau,i}\ge p^*_i) \nonumber \\ 
	& \ge p^*_i\Pr\Big(X_{\lambda\tau}\xi(i\tau)\ge \E{[X_{\lambda T}]\xi(i\tau)(1-\epsilon)} \Big) \label{mag_x2_2} \\ 
	& =
	p^*_i\Pr(X_{\lambda\tau}\ge \E{[X_{\lambda T}](1-\epsilon)}) \nonumber\\
	& = p^*_i\Pr\bigg(X_{\lambda\tau}\ge \E{[X_{\lambda T}]} \frac{\ln(\lambda\tau)}{\ln(\lambda T)}\bigg)\nonumber\\
	& \ge p^*_i\Pr(X_{\lambda\tau}\ge \E{[X_{\lambda\tau}]}) \label{lem_2} \\
	& \ge \frac{p^*_i}{e} \label{1/e_2}\\
	& \ge \frac{\E{[X_{\lambda T}]\xi(i\tau)(1-\epsilon)}}{\delta e} \nonumber\\
	& \ge \frac{\E{[X_{\lambda T}]\xi(\lceil\log_{\d}h \rceil \t)(1-\epsilon)}}{\delta e} \label{t0_2} \\
	& \ge \frac{\E{[X_{\lambda T}] \xi((\lceil\log_{\d}h \rceil +1) \t) (1-\epsilon)}}{\delta e} \nonumber
	\end{align}
	
	Equation~(\ref{mag_x2_2}) holds since $X_{\lambda\tau}\xi(i\tau)$ is a random variable representing the maximum initial valuation of agents arriving in a time interval of length $\t$ weighted by the maximum possible discount, thus it is always smaller than or equal to $Y_{\lambda\tau,i}$. 
	Equation~(\ref{lem_2}) follows from Lemma~\ref{lem:ex_ln}.
	Equation~(\ref{1/e_2}) follows from a result by~\citet{barlow1964}, which implies that, for any MHR distribution, the probability of exceeding its expectation is at least $\frac{1}{e}$.
	
	\textbf{Case $\E[X_{\l T}](1-\epsilon) < 1 $.}
	In this case we can lower bound the seller's expected revenue for the overall time period with that obtained during the interval $I_{\lceil \log_{\d}h\rceil + 1}$, as follows:
	\begin{align}
	\E_F [\R(\mathcal{M}_{\textsc{pc}})] & \ge p^*_{\lceil \log_{\d}h\rceil + 1} \left( 1 - e^{- \l \t} \right) \nonumber \\
	& \geq \xi((\lceil\log_{\d}h \rceil +1) \t) \frac{1}{e} \nonumber \\
	& \ge \frac{\E{[X_{\lambda T}] \xi((\lceil\log_{\d}h \rceil +1) \t) (1-\epsilon)}}{\delta e} \nonumber,
	\end{align}
	where for the first inequality we used the fact that the expected revenue in $I_{\lceil \log_{\d}h\rceil + 1}$ is the price posted during the interval times the probability that at least one agent arrives in $I_{\lceil \log_{\d}h\rceil + 1}$, the second inequality holds since $\left( 1 - e^{- \l \t} \right) \geq \frac{1}{e}$ when $\l \t \ge 1- \ln (e-1) \simeq 0,46$, while the last inequality follows from the fact that $\E[X_{\l T}](1-\epsilon) < 1$ and $\delta \geq 1$.
	
	We can now compute a lower bound on the ratio of the mechanism $\rho_F(\M_\textsc{pc})$, as follows:

	\begin{align}
	\rho_F(\mathcal{M}_{\textsc{pc}}) & = \frac{\E_F[\R(\M_\textsc{pc})]}{\E_F [\R(\M^\star)]} \nonumber\\
	& \ge \frac{\E_F [\R(\M_\textsc{pc})]}{\E[Y_{\lambda T}]}\nonumber\\
	& \geq\frac{\E[X_{\lambda T}]}{\E[Y_{\lambda T}]} \frac{ \xi((\lceil\log_{\d}h \rceil +1) \t) (1-\epsilon)}{\delta e}\nonumber\\
	&\ge \frac{\xi((\lceil\log_{\d}h \rceil +1) \t) (1-\epsilon)}{\delta e} \nonumber
	\end{align}
	where it is easy to see that $\frac{\E[X_{\lambda T}]}{\E[Y_{\lambda T}]} \ge 1$.
	% and $\lim_{\lambda T \rightarrow \infty} \frac{\E[X_{\lambda T}]}{\E[Y_{\lambda T}]} = 1$, hence we can write Inequality~\ref{beta_1}.\\
	%
	By recalling the condition $\l\t=(\l T)^{1-\epsilon}$,
	we have $\t=T^{1-\epsilon}\l^{-\epsilon}$,
	which allows us to write the following bound:
	\[
	\rho_F(\mathcal{M}_{\textsc{pc}}) \geq \frac{\xi((\lceil\log_{\d}h \rceil +1)T^{1-\epsilon}\l^{-\epsilon}) (1-\epsilon)}{\delta e}.
	\]
	% where $t_0$ and $\kappa$ depends on the parameters of the problem, which are $T, \l, h$, and are constant with respect to the distribution $F$.
	%
	This concludes the proof.
\end{proof}

\section{An Analytical Expression of $F_{Y_{\l T}}$ for the RV Setting with Linear Discount}\label{app:properties}

We study the cumulative distribution function of the random variable $Y_{ \l T}$ so as to unveil its dependence on $F$. 
We perform our analysis for the specific case of a linear discount function; thus:
\[
	Y_{\l T} = \max_{ i \in \{ 1, \ldots, N_T \}  } V_i \, \left(  1 - \frac{W_i}{T} \right).
\]
The results presented in the following crucially rely on some properties of Poisson processes. 

First, we introduce some auxiliary definitions and results.
%
%\begin{property}
%	A function $g : \mathbb{R}^n \to \mathbb{R}$ is said to be \emph{symmetric} if $g(x_1,\dots,x_n) = g(x_{i_1},\dots,x_{i_n})$ for any permutation $(i_1,\dots,i_n)$ of the tuple $(1,\dots,n)$ of the indexes of variables.
%\end{property}

\begin{proposition}[\citet{ross1996stochastic}]\label{gamma}
	The random variable $W_i$ representing the arrival time of agent $i$ has a Gamma distribution $\Gamma(i,\l)$, with shape parameter $i>0$ and rate parameter $\l>0$, whose probability density function is defined as follows:
	\[
		f_{W_i}(w) \coloneqq \frac{\l^i w^{i-1}}{(i-1)!} \, e^{-\l w}, \quad \textnormal{for every } w \in [0,T].
	\]
\end{proposition}

%\begin{proposition}\label{sym}
%	The function $\max_{i\in\{1,\dots,N_T\}}V_i\left(1-\frac{W_i}{T}\right)$ is symmetric {with respect to variables $W_i$.}
%	%both to the pair $(V_i,W_i)$ and to the single variables $V_i$ and $W_i$.
%\end{proposition}
%\begin{proof}
%	Variables $V_i$ are i.i.d. with distribution $F$, hence $\{V_1,\dots,V_{N_\t}\} = \{V_{i_1},\dots,V_{i_{N_\t}}\}$ for any permutation $(i_1,\dots,i_{N_\t})$ of the \textcolor{red}{sequence} $(1,\dots,{N_\t})$.  
%	Variables $W_i$ are independent and drawn from different Gamma distributions (Property \ref{gamma}). We observe that $\max_{i\in\{1,\dots,N_\t\}}V_i\big(1-\frac{W_i}{\t}\big) = \max_{i\in\{i_1,\dots,i_{N_\t}\}}V_i\big(1-\frac{W_i}{\t}\big)$ for any permutation $(i_1,\dots,i_{N_\t})$ of the \textcolor{red}{sequence} $(1,\dots,{N_\t})$. 
%\end{proof}
%\textcolor{red}{
%1) check proof}

\begin{theorem}[\citet{pinsky2010introduction}]\label{pp_unif}
	Let $W_1,W_2,\ldots$ be random variables representing the arrival times in a Poisson process with rate parameter $\l>0$. Conditioned on the event $N_T=n$, the variables $W_1,\ldots,W_n$ have a joint probability density function defined as follows:
	\[
		f_{W_1,\ldots,W_n \mid N_T=n}(w_1,\dots,w_n)=n! \, T^{-n}, \quad \textnormal{for } 0 < w_1 < \ldots < w_n \le T.
	\]
\end{theorem}
Intuitively, as discussed in~\citep{ross1996stochastic}, a consequence of Theorem~\ref{pp_unif} is that, conditioned on the event $N_T=n$, the times $W_1,\ldots, W_n$ at which the $n$ arrivals occur, considered as unordered random variables, are distributed uniformly and independently in the interval $[0,T]$.
This is the crucial observation that allows to derive the following theorem.
%
%
%\begin{remark} \citet{ross1996stochastic}
%	Under condition that $n$ events have occurred in $[0,T]$, the times $W_1,\dots, W_n$ in which events occur, considered as unordered random variables, are distributed uniformly and independently in the interval $[0,T]$.
%	% uguale scrivere (0,T) o [0,T] ? sì perché tempo continuo
%\end{remark}

\begin{theorem}\label{product}
	The random variable representing the maximum discounted valuation of agents arriving in the overall time period $[0,T]$ conditioned on the event that $N_T = n$ is defined as follows:
	\[
		Y_{\l T \mid N_T=n} \coloneqq \max_{i\in\{1,\dots,n\}}V_i \, U_i, \quad \textnormal{where } U_i \sim \mathcal{U}(0,1).~\footnote{We denote by $\mathcal{U}(a,b)$ a continuous uniform distribution over the interval $[a,b]$.}
	\]
	%where $V_i$ has distribution $F$ and $U_i \sim \mathcal{U}([0,1])$.
\end{theorem}

\begin{proof}
	% For the symmetry of the function $\max_{i\in\{1,\dots,N_\t\}}V_i\big(1-\frac{W_i}{T}\big)$ (Proposition \ref{sym}) and for Theorem \ref{pp_unif}, we can express the law of variable $Y_{\l_T}$ conditioned on $N_T$ as follows:
	%
	Given the symmetry of the functional $\max_{i\in\{1,\dots,N_T\}}V_i\, \left(1-\frac{W_i}{T}\right)$ and Theorem~\ref{pp_unif}, we can write the following:
	\begin{align*}
		\Pr \left\{ Y_{\l T}=y \mid N_T=n \right\} &=\Pr \left\{ \max_{i\in\{1,\dots,N_\t\}}V_i\bigg(1-\frac{W_i}{T}\bigg)=y \mid N_\t=n \right\}  \\
		& = \Pr \left\{  \max_{i\in\{1,\dots,n\}}V_i\bigg(1-\frac{\tilde{U}_i}{T}\bigg)=y \right\}
	\end{align*}
	where $\tilde{U}_i$ is a random variable distributed according to $\mathcal{U}(0,T)$, which is a continuous uniform distribution with support $[0,T]$.
	Letting $U_i \coloneqq \left(1-\frac{\tilde{U}_i}{T}\right)$, it is easy to show that $U_i \sim \mathcal{U}(0,1)$.
	Formally, for every $x \in [0,1]$, the cumulative distribution function $F_{U_i}$ of $U_i$ is defined as follows:
	\begin{align*}
	F_{U_i}(x)&\coloneqq\Pr \left\{ U_i \le x \right\} = \Pr\left\{ \left(1-\frac{\tilde{U}_i}{T}\right) \le x \right\} = \Pr \left\{ T(1-x)\le \tilde{U}_i \right\} \\
	& = 1-\Pr \left\{ \tilde{U}_i \le T(1-x) \right\} = 1-\frac{T(1-x)}{T} = x 
	\end{align*}
	Moreover, for $x<0$ it holds $F_{U_i}(x)=0$, while for $x>1$ it holds $F_{U_i}(x)=1$.
	Thus, $F_{U_i}$ is the cumulative distribution function of a random variable drawn from a uniform with support $[0,1]$.
\end{proof}

In the following, we denote by $Z$ a product variable $V \, U$, where $V$ and $U$ are random variables distributed according to $F$ and $\mathcal{U}(0,1)$, respectively.
% variable $VU$, where the law of $U$ is $\mathcal{U}([0,1])$.
%
Moreover, we let $Z_i \coloneqq V_i \, U_i$ be the variable $Z$ referred to a specific agent $i$.
Theorem~\ref{product} allows us to express $F_{Y_{\l T \mid N_T=j}}$ as follows:
\begin{align*}
F_{Y_{\l T \mid N_T=j}}(x) &=F_{\max_{i\in\{1,\dots,j\}}Z_i}(x) 
= \Pr \left\{ \bigcap_{i=1}^j Z_i \le x \right\} = \prod_{i=1}^j \Pr \left\{ Z_i \le x \right\} = \left[ F_Z(x) \right]^j .
\end{align*}

Hence, we can write $F_{Y_{\lambda T}}$ as:
\begin{align*}
F_{Y_{\lambda T}}(x)& = \sum_{j=1}^\infty \frac{(\l T)^j e^{-\l T}}{j!} \left[ F_Z(x) \right]^j ,
\end{align*}
where
\begin{equation}\label{eq:F_Z}
F_{Z}(x) = \left\{\begin{array}{ll}x \int_{1}^{h} \frac{1}{v} f(v) \dd v & \text { if } x \in[0,1) \\ F(x)+x \int_{x}^{h} \frac{1}{v} f(v) \dd v & \text { if } x \in[1, h]\end{array}\right. .
\end{equation}
Thus, it is easy to see that $F_{Y_{\l T}}$ depends on $F$ and $f$, which are the cumulative distribution function and the probability density function of agents' initial valuations, respectively.

%All the calculations needed to obtain $F_Z$ can be found in the Appendix.

It remains to show how to derive the expression of $F_Z$ in Equation~\eqref{eq:F_Z}.
%
% $F_{Z}$ is the conditional cumulative distribution function of the product variable $Z=VU$, where $V$ is the agent valuation. Its probability density function is $f$ and its cumulative distribution function is $F$. 
%
% In the following, for the sake of clarity, we denote such functions by $f_V$ and $F_V$.
%
Notice that, since $U \sim \mathcal{U}(0,1)$, the probability density function of $U$ is defined as $f_U(u)=\mathds{1}_{[0,1]}(u)$, while its cumulative distribution function is $F_U(u)=u\mathds{1}_{[0,1]}(u)$.
The support of $Z$ is $[0,h]$, being $V$ defined on $[1,h]$.
% In particular, $Z_i=V_iU_i$ denotes variable $Z$ referred to a specific agent $i$.
\begin{align*}
F_{Z}(z) =&\Pr \left\{ V \, U \le z \right\} = \Pr \left\{  U \le \frac{z}{V} \right\}= \\
=& \mathds{1}_{[0,1)}(z) \iint_{\mathcal{D}'} f(v) f_U(u) \dd v \dd u +\mathds{1}_{[1, h]}(z) \iint_{\mathcal{D}''} f(v) f_U(u) \dd v \dd u= \\
=& \mathds{1}_{[0,1)}(z) \int_{1}^{h} f(v) \int_{0}^{z/v} f_U(u) \dd u \dd v \, + \\
& + \mathds{1}_{[1, h]}(z)\left(\int_{1}^{z} f(v) \dd v \int_{0}^{1} f_U(u) \dd u +\int_{z}^{h} f(v) \int_{0}^{z /v} f_U(u) \dd u \dd v\right) \\
=& \mathds{1}_{[0,1)}(z) \left(z \int_{1}^{h} \frac{1}{v} f(v) \dd v\right)+ \mathds{1}_{[1, h]}(z)\left(F(z)+z \int_{z}^{h} \frac{1}{v} f(v) \dd v\right) ,
\end{align*}
%
% The domain of integration is denoted by $\mathcal{D}'$ when $z \in [0,1)$, while it is called $\mathcal{D}''$ when $z \in [1,h]$.
%
where the domains of integration $\mathcal{D}^{\prime}$ and $\mathcal{D}^{\prime \prime}$ are defined as: 
\[ \mathcal{D}' \coloneqq \left\{(u,v) : 0 \leq u \leq \frac{z}{v}, 1 \leq v \leq h\right\} \]
\[ \mathcal{D}'' \coloneqq\mathcal{D}_{1}'' \cup \mathcal{D}_{2}'' \coloneqq\{(u,v) : 0 \leq u \leq 1,1 \leq v \leq z\} \cup\left\{(u,v) : 0 \leq u \leq \frac{z}{v}, z < v \leq h\right\} \]
See also Figure~\ref{Fig:domains} for a graphical representation of the domains.

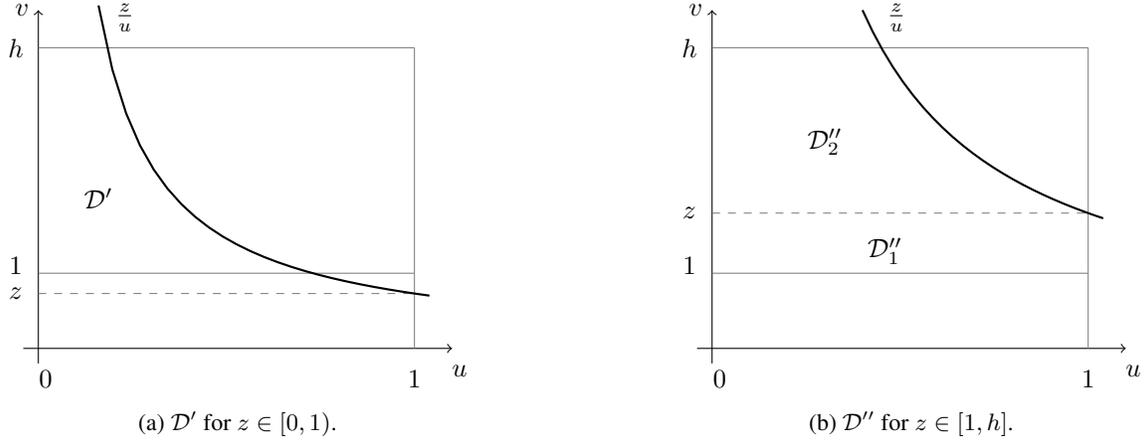
\begin{figure}[h]
	\begin{subfigure}{0.5\textwidth}
		\centering
		\begin{tikzpicture}[domain=0:2]
		\draw[->] (-0.2,0) -- (5.5,0) node at (5.6,-0.3) {$u$};
		\draw[->] (0,-0.2) -- (0,4.5)node[left] {$v$};
		\draw[gray, thin] (0,4) -- (5,4);
		\draw[gray, thin] (0,1) -- (5,1);
		\draw[gray, thin] (5,0) -- (5,4);
		\draw[gray, thin, dashed] (0,3.65/5) -- (5,3.65/5);
		\node at (-0.3,4) {$h$};
		\node at (-0.3,3.65/5) {$z$};
		\node at (-0.3,1.1) {$1$};
		\node at (5,-0.4) {$1$};
		\node at (0.1,-0.4) {$0$};
		\node at (1.15,4.4) {$\frac{z}{u}$};
		\node at (0.8,2) {$\mathcal{D}^{\prime}$};
		\draw[color=black, thick, domain=0.8:5.2] plot (\x,3.65/\x); %node[right] {$p(t)$};
		\end{tikzpicture}
		\caption{$\mathcal{D}'$ for $z \in  [0,1) $.}
	\end{subfigure} 
	\begin{subfigure}{0.5\textwidth}
		\centering
		\begin{tikzpicture}[domain=0:2]
		\draw[->] (-0.2,0) -- (5.5,0) node at (5.6,-0.3) {$u$};
		\draw[->] (0,-0.2) -- (0,4.5)node[left] {$v$};
		\draw[gray, thin] (0,4) -- (5,4);
		\draw[gray, thin] (0,1) -- (5,1);
		\draw[gray, thin] (5,0) -- (5,4);
		\draw[gray, thin, dashed] (0,9/5) -- (5,9/5);
		\node at (-0.3,4) {$h$};
		\node at (-0.3,9/5) {$z$};
		\node at (-0.3,1.1) {$1$};
		\node at (5,-0.4) {$1$};
		\node at (0.1,-0.4) {$0$};
		\node at (2.45,4.4) {$\frac{z}{u}$};
		\node at (1.5,2.75) {$\mathcal{D}^{\prime \prime}_{2}$};
		\node at (2.3,1.3) {$\mathcal{D}^{\prime \prime}_{1}$};
		\draw[color=black, thick, domain=2:5.2] plot (\x,9/\x); %node[right] {$p(t)$};
		\end{tikzpicture}
		\caption{$\mathcal{D}''$ for $z \in \left[ 1,h \right] $.}
	\end{subfigure} 
	\caption{Graphical representation of the domain of integration $\mathcal{D}'$ and $\mathcal{D}''$.}
	\label{Fig:domains}
\end{figure}

\section{Additional Experiments}
We provide other empirical evaluations of our mechanisms in the RV setting. We compare $\M_{\textsc{c}}$, $\M_{\textsc{pc}}$, and ESoES-SS when the distribution of the agents' valuations is \emph{not} MHR. Then, we show the performances of $\M_{\textsc{c}}$ and $\M_{\textsc{pc}}$ when valuations are linearly discounted over time.

\paragraph{Result \#3}
We perform an experiment similar to that of Result \#1. Here, agents' valuations are drawn from a truncated normal distribution with $\mu=\frac{h-1}{2}$, $\sigma^2=2$, and support $[1,h]$.
Figure~\ref{figure:empiricalevaluationbabaioff3} is similar to Figure~\ref{figure:empiricalevaluationbabaioff1}.
Observe that, in this setting, the performances of $\M_{\textsc{pc}}$ with $Nsub=13$ and $\M_{\textsc{pc}}$ with $Nsub=232$ are analogous.
This means that, tuning the parameters in a suitable way, we can impose a time constraint with almost no loss in the normalized mean revenue.
Moreover, the truncated normal distribution is \emph{not} MHR, hence, all the bounds on the competitive ratio of the mechanism do not hold.
Despite this fact, we see that, in this scenario, the behavior of the mean normalized revenue is comparable to that of Result \#1.
In particular, the loss of ESoES-SS w.r.t.~$\M_{\textsc{C}}$ averaged over the values of $\lambda$ is about $0.2\, h$ when $T=10$, and slightly larger when $T=50$.

\begin{figure}[h]
	\centering
	\begin{subfigure}{0.232\textwidth}
		\includegraphics[width=\textwidth]{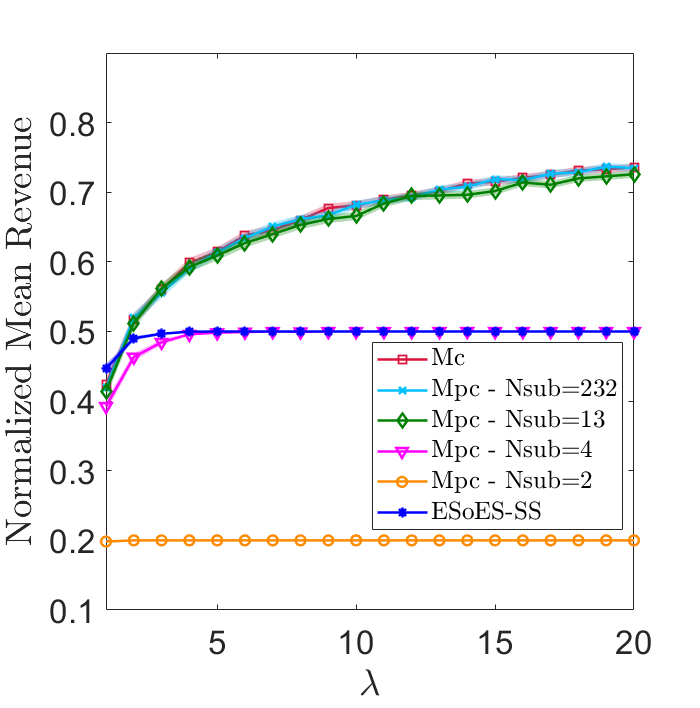}
		\caption{$T=10$} 
	\end{subfigure} \qquad
	\begin{subfigure}{0.232\textwidth}
		\includegraphics[width=\textwidth]{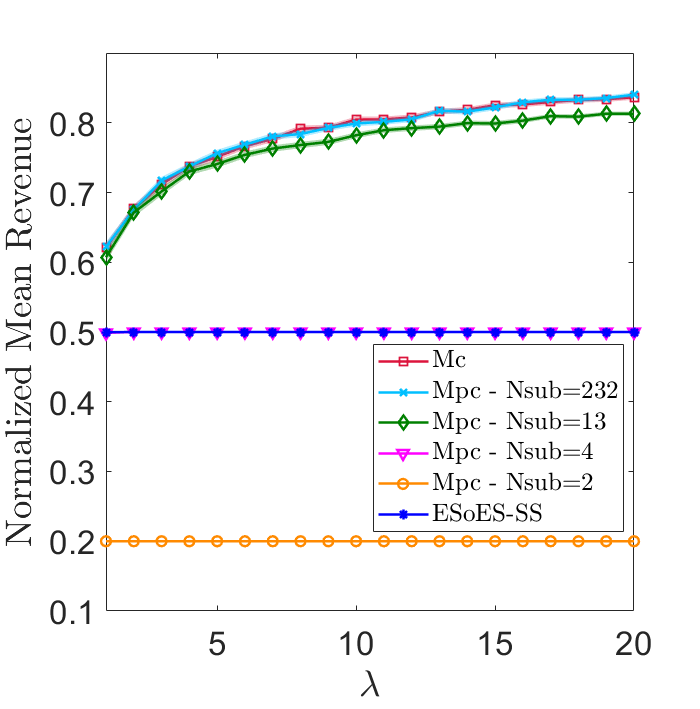}
		\caption{$T=50$}
	\end{subfigure}
	\caption{Average normalized revenue of $\M_{\textsc{c}}$, $\M_{\textsc{pc}}$, and ESoES-SS.}
	\label{figure:empiricalevaluationbabaioff3}
\end{figure}

%\paragraph{Result \#4}
%For every value of $v \in \{1.00, 1.05, 1.10,\ldots,h\}$, $\lambda$ and $h$ we run 1000 simulations.
%%
%Given parameters $\lambda$ and $h$, we simulated the arrivals of agents having the same valuation $v$ (setting IV) and we computed the revenue of mechanisms $\M_{\textsc{c}}$, ESoES-SS and $\M_{\textsc{pc}}$ with different Nsub. 
%%
%We normalized the results by $h$ and, for each value of $\lambda$, we average by the simulations and by the values of $h$.
%%
%Then, for each value of $\lambda$, we average by the values of $v$ and we plot in Fig.~\ref{figure:empiricalevaluationbabaioff5} the normalized mean revenues of the mechanisms and the 95\% confidence intervals with respect to the mean values of $v$.
%
%
%\begin{figure}[h]
%	\begin{subfigure}{0.232\textwidth}
%		\includegraphics[width=\textwidth]{Final_graphics/REV10.png}
%		\caption{$T=10$} 
%	\end{subfigure}
%	\begin{subfigure}{0.232\textwidth}
%		\includegraphics[width=\textwidth]{Final_graphics/REV50.png}
%		\caption{$T=50$}
%	\end{subfigure}
%	\caption{Average normalized revenue of $\M_{\textsc{c}}$, $\M_{\textsc{pc}}$, ESoES-SS.}
%	\label{figure:empiricalevaluationbabaioff5}
%\end{figure}

\paragraph{Result \#4}
We analyze mechanisms $\M_{\textsc{c}}$ and $\M_{\textsc{pc}}$ with different $Nsub$ values when the valuations of the agents are linearly discounted.
For every $\lambda$ we run $1000$ Monte Carlo simulations, with $h=10$.
Given parameters $\lambda$ and $h$, we simulated the arrivals of agents drawn from a uniform distribution with support $[1,h]$ and we computed the revenue of the mechanisms.
We normalized the results by $h$ and, for each value of $\lambda$, we average by the simulations.
Then, for each value of $\lambda$, we plot in Figure~\ref{figure:empiricalevaluationbabaioff6} the normalized mean revenues of the mechanisms, for $T=10$ and $T=50$.
We observe that $\M_{\textsc{c}}$ is no longer the best mechanism in terms of normalized mean revenue. 
The interesting fact is that, a suitably tuned mechanism $\M_{\textsc{pc}}$ can reach a better average revenue than $\M_{\textsc{c}}$ in some IV scenarios. 

\begin{figure}[h]
	\centering
	\begin{subfigure}{0.232\textwidth}
		\includegraphics[width=\textwidth]{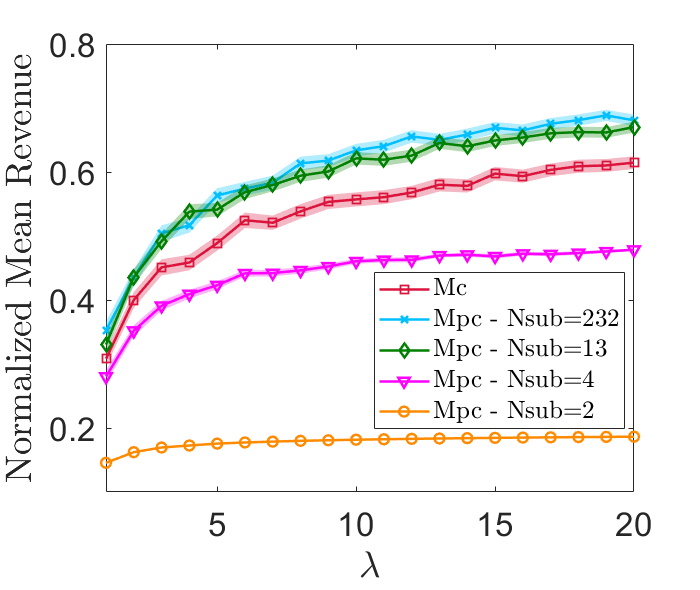}
		\caption{$T=10$} 
	\end{subfigure} \qquad
	\begin{subfigure}{0.232\textwidth}
		\includegraphics[width=\textwidth]{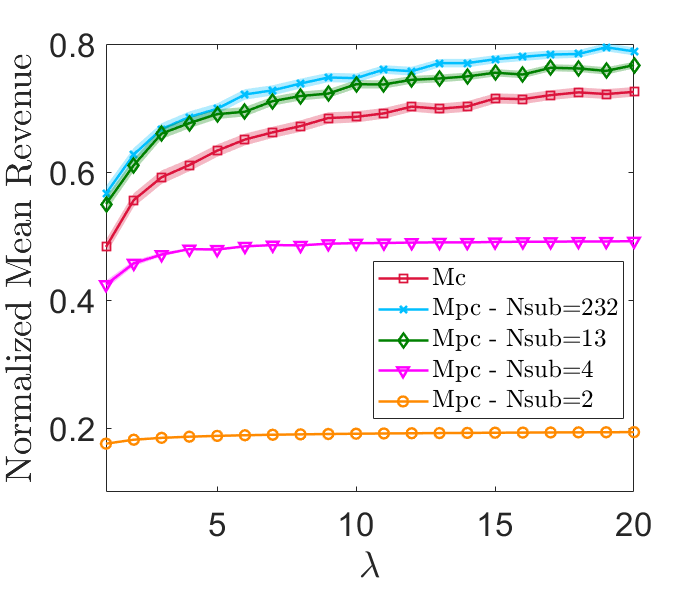}
		\caption{$T=50$}
	\end{subfigure}
	\caption{Average normalized revenue of $\M_{\textsc{c}}$, $\M_{\textsc{pc}}$, and ESoES-SS.}
	\label{figure:empiricalevaluationbabaioff6}
\end{figure}

%\paragraph{Result \#1 BIS}
%
%\textcolor{red}{Alternativo a Result \#1 : se non mediamo sugli h e diciamo che osservandome diversi troviamo lo stesso comportamento (con T diciamo così). In Fig.~\ref{figure:empiricalevaluationbabaioff2} h=10, Fig.~\ref{figure:empiricalevaluationbabaioff4} h=50}
%
%In Fig.~\ref{figure:empiricalevaluationbabaioff4} we can observe the same plots of Fig.~\ref{figure:empiricalevaluationbabaioff2} with a different $h$ parameter. We can see that as parameters $\l$ and $T$ grows, the gap has the same behavior
%
%\begin{figure}[h]
%	\begin{subfigure}{0.232\textwidth}
%		\includegraphics[width=\textwidth]{Final_graphics/GAP_T10_h50.png}
%		\caption{$T=10$} 
%	\end{subfigure}
%	\begin{subfigure}{0.232\textwidth}
%		\includegraphics[width=\textwidth]{Final_graphics/GAP_T50_h50.png}
%		\caption{$T=50$}
%	\end{subfigure}
%	\caption{Maximum normalized gap between $\M_{\textsc{c}}$ and ESoES-SS.}
%	\label{figure:empiricalevaluationbabaioff4}
%\end{figure}

\end{document}